\documentclass[11pt]{article}

\usepackage[utf8]{inputenc}
\usepackage{csquotes}
\usepackage[american]{babel}
\usepackage[margin=0.85in]{geometry}
\usepackage{enumitem}
\usepackage{hyperref}
\usepackage{amsmath}
\usepackage{cleveref}
\usepackage{amssymb}
\usepackage{amsthm}
\usepackage{mathtools}
\usepackage{bbm}
\usepackage{accents}
\usepackage{booktabs}
\usepackage{graphicx}
\usepackage{comment}
\usepackage[title]{appendix}
\usepackage{xcolor}
\usepackage{authblk}
\usepackage{multirow}
\usepackage{multicol}
\usepackage{setspace}
\usepackage{float} 
\usepackage{proof-at-the-end}
\usepackage{centernot}
\usepackage{subcaption}
\usepackage{threeparttable}
\usepackage{rotating}
\usepackage{pdflscape}

\usepackage[authordate,backend=bibtex,natbib]{biblatex-chicago}
\newtheorem{definition}{Definition}
\newtheorem{corollary}{Corollary}
\newtheorem{lemma}{Lemma}
\newtheorem{claim}{Claim}

\newtheorem{remark}{Remark}
\newtheorem{example}{Example}
\newtheorem{assumption}{Assumption}

\makeatletter
\@addtoreset{case}{proposition}
\makeatother

\makeatletter
\def\blfootnote{\gdef\@thefnmark{}\@footnotetext}
\makeatother

\makeatletter
\def\footnoterule{\kern-3\p@
  \hrule \@width 2in \kern 2.6\p@} 
\makeatother

\newcommand{\dbar}[1]{\Bar{\Bar{#1}}} 
\newtheorem{innercustomassumption}{Assumption}
\newenvironment{customassumption}[1]
  {\renewcommand\theinnercustomassumption{#1}\innercustomassumption}
  {\endinnercustomassumption}
\newcommand{\independent}{\perp\mkern-9.5mu\perp}  
\newcommand{\notindependent}{\centernot{\independent}} 

\newcommand{\conv}{\xrightarrow{\;\;\;}}

\graphicspath{{graphs/}}

\makeatletter
\def\input@path{{Tables/}}
\makeatother

\begin{document}

\title{Measuring Diagnostic Test Performance Using Imperfect Reference Tests: A Partial Identification Approach}
\author[1]{Filip Obradović\thanks{Northwestern University, Department of Economics, 2211 Campus Drive, Evanston 60208 IL, USA. Email: \href{mailto:obradovicfilip@u.northwestern.edu}{obradovicfilip@u.northwestern.edu}
\\
I am deeply grateful to Charles Manski for his guidance and support, and to Amilcar Velez for extensive discussions that have helped shape the paper. I am also thankful to Ivan Canay, Federico Bugni, Joel Horowitz, Eric Auerbach, Francesca Molinari, Peter Hansen, and Evan Rose for valuable suggestions. I thank the participants at the Econometrics Seminar, Econometrics Reading Group, 501 Seminar at Northwestern and the 2022 Annual Health Econometrics Workshop at Emory University for comments. Financial support from the Robert Eisner Memorial Fellowship and the Unicredit Crivelli Scholarship is gratefully acknowledged.
}}

\vspace{-1in}

\date{August, 2024}

\maketitle

\vspace{-0.3in}

\begin{center} 
\textbf{Abstract}
\end{center}
 
\begin{spacing}{1.2}

\bigskip
 {\small Diagnostic tests are almost never perfect. Studies quantifying their performance use knowledge of the true health status, measured with a reference diagnostic test. Researchers commonly assume that the reference test is perfect, which is often not the case in practice. When the assumption fails, conventional studies identify ``apparent'' performance or performance with respect to the reference, but not true performance. This paper provides the smallest possible bounds on the measures of true performance - sensitivity (true positive rate) and specificity (true negative rate), or equivalently false positive and negative rates, in standard settings. Implied bounds on policy-relevant parameters are derived: 1) Prevalence in screened populations; 2) Predictive values. Methods for inference based on moment inequalities are used to construct uniformly consistent confidence sets in level over a relevant family of data distributions. Emergency Use Authorization (EUA) and independent study data for the BinaxNOW COVID-19 antigen test demonstrate that the bounds can be very informative. Analysis reveals that the estimated false negative rates for symptomatic and asymptomatic patients are up to $3.17$ and $4.59$ times higher than the frequently cited ``apparent'' false negative rate. Further applicability of the results in the context of imperfect proxies such as survey responses and imputed protected classes is indicated.}
    
\end{spacing}

\

\noindent KEYWORDS: Sensitivity, specificity, partial identification, moment inequalities, gold standard bias

\

\noindent JEL classification codes: C14, I10

\thispagestyle{empty} 
\newpage
\setcounter{page}{1}

\section{Introduction}

Diagnostic tests are indispensable in modern clinical decision making. As they are almost never perfect, evaluation of test performance is a common research goal. Test performance studies seek to quantify test accuracy predominantly in the form of sensitivity and specificity, also referred to as performance measures or operating characteristics. Sensitivity (true positive rate) is the probability that a test will return a positive result for an individual who truly has the underlying condition, while specificity (true negative rate) is the probability that a test will produce a negative result for an individual who does not have the underlying condition. Equivalently, one can measure false positive and false negative rates. False negative rate and sensitivity sum to unity, as do specificity and the false positive rate. 

Determining sensitivity and specificity for a diagnostic test of interest, referred to as an index test, requires knowledge of the true health status for all participants in the study. The true health status is most often unobservable, so a reference test is commonly used in lieu of it. However, such tests are rarely perfect themselves. When the reference is imperfect, conventional studies only identify ``apparent'' sensitivity and specificity, or the so-called rates of positive and negative agreement with the reference.\footnote{``Apparent'' false negative rate and ``apparent'' sensitivity sum to unity. Similarly, ``apparent'' false positive rate and ``apparent'' specificity sum to unity.} They measure performance with respect to the reference test and not true performance. Hence, ``apparent'' parameters are typically not of interest. Moreover, the true performance measures are usually only partially identified, as shown in \Cref{identification_sect}. In other words, there exists a set of parameter values that are consistent with the observed data, called the identified set. The smallest such set under maintained assumptions, or the set that exhausts all information from the data, is known as the sharp identified set.  This paper addresses the issue of finding, estimating and doing inference on the points in the sharp identified set for sensitivity and specificity, or equivalently false negative and positive rates, under standard assumptions used in the literature.

I first characterize the sharp joint identified set for the true performance measures without imposing any assumptions on the latent statistical dependence between the index and reference tests conditional on health status, assuming exact or approximate knowledge of the reference test characteristics. The set is a line segment or a union of line segments in $[0,1]^2$, in contrast to rectangular sets following from comparable existing bounds (\citet{thibodeau1981evaluating}, \citet{emerson2018biomarker}). The proposed framework allows researchers to layer on additional assumptions regarding the latent dependence to further reduce the size of the set. This is demonstrated through a formalization of an informally stated assumption from the literature which is plausible when the two tests are physiologically related.

Sensitivity and specificity are frequently used to obtain other policy-relevant parameters. I present how the derived identified sets may be utilized to sharply bound prevalence, or the population rate of illness, in populations screened by the index test. Bounds may be markedly narrower than those implied by existing comparable methods, owing to the specific shape of the identified sets. Implied bounds on predictive values, i.e. probabilities that a patient is sick conditional on observing a test result, are discussed in Appendix \ref{sect_bounding_pv}.

The FDA Statistical Guidance on Reporting Results Evaluating Diagnostic Tests\footnote{Link: \href{https://www.fda.gov/media/71147/download}{https://www.fda.gov/media/71147/download} (Last accessed: 12/25/2022)} requires diagnostic performance studies to report confidence intervals for index test sensitivity and specificity to quantify the statistical uncertainty in the estimates. To conform to the practice, I construct confidence sets for points in the identified set using an inference method based on moment inequalities (\citet{romano2014practical}). The confidence sets are uniformly consistent in level over a large family of permissible distributions relevant in the application. Namely, they asymptotically cover all points in the identified set uniformly over the family of population distributions with probability of at least $1-\alpha$, where $\alpha$ is the chosen significance level.

Diagnostic test performance studies for rapid COVID-19 tests have a mandated RT-PCR reference test which is known to produce false negative results, and thus pertain to the setting analyzed in the paper.\footnote{Link: \href{https://www.fda.gov/media/137907/download}{https://www.fda.gov/media/137907/download} (Last accessed: 12/25/2022)} \citet{fitzpatrick2021buyer} emphasize that the false negative rate of the ubiquitous \textit{Abbott BinaxNOW COVID-19 Ag2 CARD} rapid antigen test may be substantially understated by the reported ``apparent'' analog due to imperfect reference tests.\footnote{The test held $75\%$ of the COVID-19 antigen test market share in the United States, according to Abbott Laboratories CEO Robert Ford on Q3 2021 Results - Earnings Call Transcript.}

Application of the method to the data from the original Emergency Use Authorization (EUA) performance study, as well as an independent study by \citet{shah2021performance} bolsters this claim and reveals that bounds can be very informative. Depending on interpretation, the results from both studies suggest that the test may not satisfy the initial FDA requirement for EUA of at least $80\%$ estimated sensitivity, despite fulfilling the criterion of high ``apparent'' sensitivity, implying the need for alternative testing protocols. Moreover, the estimated false negative rates for symptomatic and asymptomatic patients are up to $3.17$ and $4.59$ times higher than the frequently cited ``apparent'' false negative rate, warranting further attention. Comparison with existing bounds reveals that the proposed method can provide significant reductions in the size of the identified set for operating characteristics, and consequently in the width of implied bounds on prevalence when the test is used for screening.

The methodological framework developed in the paper offers solutions to two issues in the current research practice guidelines set forth by the FDA Statistical Guidance, as explained by \Cref{remark_fda}: 1) Inability to measure true test performance in common settings; 2) Inability to compare index and reference test performance. It also addresses the concerns raised in \citet{boyko1988reference}: \textit{``When two tests are strongly suspected of being conditionally dependent, then the performance of one of these tests should probably not be compared with that of the other, unless better methods are developed to sort out the degree of bias caused by reference test errors in the presence of conditional dependence."}

Provided replication files allow researchers to directly utilize the findings of the paper to obtain estimates and confidence sets in their own work.\footnote{Available from: \href{https://github.com/obradovicfilip/bounding\_test\_performance}{https://github.com/obradovicfilip/bounding\_test\_performance}} Since the method does not require any changes to the data-collection process of standard studies, it can also be readily applied to estimate test performance based on published data, as demonstrated by the application section of the paper. 

Broader applicability of the approach is discussed in \Cref{sect_broader_applications}. It can be used to study features of the joint distribution between a binary outcome and a binary latent variable measured with an imperfect proxy. Illustrative examples include variables such as program participation as indicated by a survey response, or race, imputed using the Bayesian Improved Surname Geocoding (BISG) algorithm. The method is appealing when validation studies measuring proxy misclassification rates exist. The bounds also readily apply whenever one wishes to learn performance of a binary classifier by comparing it to another imperfect classifier or label, rather than the ground truth, which is common in satellite imagery and other remote sensing applications.

\subsection{Related Literature}\label{sect_literature}

\citet{gart1966comparison2}, \citet{staquet1981methodology}, and \citet{zhou2009statistical} show that if the reference and index tests are statistically independent conditional on the true health status, index test sensitivity and specificity are point identified, assuming exactly known reference test performance measures. (see also \citet{hui1980estimating}) However, \citet{vacek1985effect}, \citet{valenstein1990evaluating}, \citet{hui1998evaluation} and \citet{emerson2018biomarker} elaborate that conditional independence may frequently be untenable. A salient case is when the two tests are physiologically related, such as when they rely on the same type of sample (e.g. nasal swab or capillary blood) or measure the same quantities (e.g. antibody reaction to tuberculin).

Tests are generally expected to be dependent, but the dependence structure is latent since true status is unobservable. Several authors explore how dependence between the tests may affect the direction of gold standard bias, defined as the difference between the ``apparent'' and true performance measures. (\citet{deneef1987evaluating}, \citet{boyko1988reference}) \citet{valenstein1990evaluating} concludes that when errors committed by index and reference tests are highly correlated, the ``apparent'' measures may overstate the true parameters. However, they do not precisely define highly correlated errors, prompting the formalization of the assumption in this paper. Authors focus on the direction of the effects of the conditional dependence, rather than on the magnitude. The purpose is to allow researchers to determine whether their estimates are biased upwards or downwards. However, since the dependence cannot be measured, the practical relevance of these findings is diminished. Additionally, one could argue that the magnitude is perhaps even more important than the direction of the bias.

A formal approach to the issue of unknown bias magnitude is found in \citet{thibodeau1981evaluating}. Assuming that the reference and index tests are positively correlated, and that the reference is at least as accurate as the index, the author bounds the bias. The framework presented below does not require such assumptions. More recently, \citet{emerson2018biomarker} sketch an argument for individual bounds on sensitivity and specificity under similar assumptions used in this paper. This study contributes to the literature on gold standard bias by presenting the sharp joint identified set for test performance measures, formalizing and incorporating existing dependence assumptions to further reduce its size, bounding derived parameters of interest, and providing an appropriate uniform inference procedure.

In doing so, it builds upon on existing work on partial identification (\citet{manski2003partial}, \citet{manski2007identification}). \Cref{bounds_theta_prop} revitalizes the general analysis of \citet{cross2002regressions} in the context of diagnostic test performance measurements. Technical contributions over their work primarily lie within the novel identification findings in Propositions \ref{bounds_wrongly_agree_prop} and \ref{prevalence_screening_prop}, as well as the inference procedure for the points in the identified set that is uniformly consistent in level. The procedure assumes a mild and easily interpretable restriction on the family of population distributions, and relies on existing methods for inference in moment inequality models (\citet{andrews2010inference}, \citet{andrews2012inference}, \citet{chernozhukov2013intersection}, \citet{romano2014practical}, \citet{canay2017practical}, \citet{bugni2017inference}, \citet{chernozhukov2019inference},  \citet{kaido2019confidence}, \citet{bai2021two}). 

This paper aligns with a growing body of literature concerning partial identification in medical and epidemiological research. (\citet{bhattacharya2012treatment}, \citet{manski2020bounding}, \citet{toulis2021estimation}, \citet{manski2021estimating}, \citet{ziegler2021binary}, \citet{stoye2022bounding}) It also aims to contribute to the corpus of COVID-19 test performance studies by estimating the true sensitivity and specificity for COVID-19 antigen tests despite reference test imperfections under plausible assumptions. (\citet{shah2021performance}, \citet{pollock2021performance}, \citet{siddiqui2021implementation})

The remainder of the paper is organized as follows. \Cref{identification_sect} provides the identification argument. \Cref{sect_bounding_prevalence} discusses identification of prevalence. \Cref{finite_sample_sect} explains estimation and inference. \Cref{application_sect} presents confidence and estimated identified sets for the operating characteristics of the COVID-19 antigen test. \Cref{sect_broader_applications} indicates uses of the results beyond the context of diagnostic test performance studies. \Cref{conclusion_sect} concludes. Appendix \ref{sect_bounding_pv} derives bounds on predictive values. All proofs are collected in Appendix \ref{sect_proofs}.

\section{Identification}\label{identification_sect}

Studies quantifying the performance of a test of interest, also known as an index test, require knowledge of the true health status. Health status is usually unobservable, so it is determined by an alternative test, called the reference test. Even though the reference test should be the best available test for the underlying condition, it is almost always imperfect in practice, giving rise to identification issues. Let $t=1$ and $r=1$ if the index and reference tests, respectively, yield positive results and $t=0$, $r=0$ otherwise. Let $y=1$ denote the existence of the underlying condition we are testing for and $y=0$ the absence of it.\footnote{I interchangeably say that the person is ill when $y=1$ and when $y=0$, that they are healthy. This can be extended to encompass antibody tests with minor semantic changes, since they can also measure if a person has been ill.} 

We are interested in learning the sensitivity and specificity of the index test:
\begin{align}
\text{Sensitivity: }\theta_1 &= P(t=1|y=1)\label{eq_define_sens}\\
\text{Specificity: }\theta_0 &= P(t=0|y=0)\label{eq_define_spec}
\end{align}
which are defined when $P(y=1)\in(0,1)$ in the study population. Equivalently, one can study the false negative and false positive rates, $1-\theta_1$ and $1-\theta_0$. Similarly, define reference test sensitivity $s_1 = P(r=1|y=1)$ and specificity $s_0 = P(r=0|y=0)$. Data collection in test performance studies is commonly done by testing all participants with both the reference and index tests. The observed outcome for each participant is $(t,r)\in\{0,1\}^2$. The data identify the joint probability distribution $P(t,r)$. When $P(r=1)\in(0,1)$, ``apparent'' sensitivity $\tilde{\theta}_1 = P(t=1|r=1)$, and ``apparent'' specificity $\tilde{\theta}_0 = P(t=0|r=0)$ are also identified.
 
It is typically assumed that the reference test is perfect, so that $r=y$. Then $(\tilde{\theta}_1,\tilde{\theta}_0)=(\theta_1,\theta_0)$. This is rarely the case in practice. Generally, $\tilde{\theta}_j\neq\theta_j$ for some $j=0,1$, which referred to as \textit{gold standard bias}. Interpreting $(\Tilde{\theta}_1,\Tilde{\theta}_0)$ as true performance measures can lead to severely misleading conclusions due to the bias. Alternatively, researchers may explicitly study $(\Tilde{\theta}_1,\Tilde{\theta}_0)$. However, they only measure performance of $t$ with respect to $r$, and not $y$. If one wishes to learn about true performance $(\theta_1,\theta_0)$, then these parameters are not of interest. 

\begin{remark}\label{rem:availability_r}
    Index test $t$ is usually a novel test, evaluated against the best currently available test $r$. In some settings, use of $r$ may not be practical outside the performance study due to high costs, long turnaround time or invasiveness. For example, a reference test for some types of dementia is a postmortem neuropathological report which is not helpful for diagnosis. Viral antigen tests may be preferred over reference RT-PCR tests for screening purposes due to lower resource requirements.
\end{remark}

Focusing the analysis on binary tests and binary health statuses is standard procedure. FDA Statistical Guidance on Reporting Results Evaluating Diagnostic Tests recognizes only binary reference tests and health statuses, explicitly stating: \textit{``A reference standard ... divides the intended use population into only two groups (condition present or absent)."} Many tests that yield discrete or continuous test results, such as RT-PCR tests, are reduced to binary tests by thresholding in practice. While the results of this paper can be extended to cases in which ranges of $t$ and $r$ are finite sets, I limit the analysis to the binary setting to conform to current research practice.

The section begins by outlining the formal assumptions used. I then provide the set of parameter values $(\theta_1,\theta_0)$ consistent with the observed data, also known as the identified set, without imposing any assumptions on the statistical dependence between $t$ and $r$. The set is sharp, or the smallest possible under maintained assumptions. For simplicity of exposition, this is first done when $(s_1,s_0)$ are known. I show how an additional assumption on the dependence structure between the two tests can be used to further reduce the size of the sharp identified set. Finally, I allow $(s_1,s_0)$ to be approximately known by assuming $(s_1,s_0)\in\mathcal{S}$, where $\mathcal{S}$ is some known set.

\subsection{Assumptions}\label{sect_assumptions}

The framework in this paper relies on common assumptions maintained in the literature.

\begin{assumption}{(Random Sampling)}\label{random_sampling_assumption}
The study sample is a sequence of i.i.d random vectors $W_i=(t_i,r_i)$, where each $W_i$ follows a categorical distribution $P(t,r)$ for $(t,r)\in\{0,1\}^2$ and $i=1,\hdots,n$.
\end{assumption}

The distribution $P(t,r)$ is a marginal of the joint distribution $P(t,r,y)$. Since $y$ is not observable, $P(t,r,y)$ is not point identified.

\begin{assumption}{(Reference Performance)}\label{reference_performance_assumption}
Sensitivity and specificity of the reference test $s_1 = P(r=1|y=1)$ and $s_0 = P(r=0|y=0)$ are known, and $s_1>1-s_0$.
\end{assumption}

The analysis is first done for the simple case when $(s_1,s_0)$ are known exactly. The approach is then generalized in \Cref{imperfect_knowledge} by assuming $(s_1,s_0)\in\mathcal{S}$, where $\mathcal{S}$ is some known set. Thus, reference test performance needs to be known only approximately. The generalization can also be used to perform sensitivity analyses. Knowledge of $(s_1,s_0)$ or $\mathcal{S}$ is commonly assumed in work dealing with gold standard bias correction, such as \citet{gart1966comparison2}, \citet{thibodeau1981evaluating}, \citet{staquet1981methodology}, and \citet{emerson2018biomarker}. The current norm of relying on the assumption that the reference test is perfect means that researchers regularly maintain $(s_1,s_0)=(1,1)$, which is a stronger condition since it implies Assumption \ref{reference_performance_assumption}. The framework hence weakens standard assumptions. However, knowledge of $(s_1,s_0)$ or $\mathcal{S}$ is a crucial identifying assumption and warrants further discussion.

Assumption \ref{reference_performance_assumption} is particularly appealing when one has access to a study that identifies performance of $r$. \citet{mathews2010estimating} and \citet{matos2011clinical} justify assumed $(s_1,s_0)$ based on such studies. In specific settings, $y$ may indeed be plausibly observable, allowing for identification of $(s_1,s_0)$. However, the procedure needed to observe $y$ may be exceedingly costly, invasive, or have an unsuitably long turnaround time for widespread use in performance measurement. Thus, $r$ is commonly used as a reference for $t$ instead. For example, a neuropathological autopsy report is the only way to pose a definitive diagnosis of Alzheimer's disease, i.e. observe $y$. (\citet{suemoto2023autopsy}) Studies of novel tests $t$ for Alzheimer's disease frequently use an amyloid positron emission tomography (PET) scan or a clinical diagnosis as $r$, since a neuropathological report may be unobtainable. (\citet{budelier2020biomarkers}, \citet{wang2023bayesian})  Autopsy performance studies for $r$ in which $y$ is directly observed post-mortem are available, identifying performance of $r$. (\citet{patwardhan2004alzheimer}) 

Unfortunately, in many cases such a performance study for $r$ does not exist. Nevertheless, tests are often expected to have precisely measured analytical performance measures. Analytical specificity and sensitivity are defined as performance measures obtained based on contrived, rather than clinical samples. This may provide some information on how $r$ will perform in clinical settings. For example, \citet{kucirka2020variation} and \citet{kanji2021false} maintain that COVID-19 RT-PCR tests are perfectly specific owing to the absence of cross-reactivity with other pathogens, that is, due to its perfect analytical specificity.\footnote{Specificity on contrived laboratory samples containing other pathogens, but not SARS-CoV-2.}

The commonly maintained assumption $(s_1,s_0)=(1,1)$ has been disputed for a plethora of reference tests. This fact indicates that at least a set $\mathcal{S}$ of more credible values $(s_1,s_0)$ exists for a variety of tests used as $r$. In these cases, the method will yield sharp bounds on $(\theta_1,\theta_0)$. However, if nothing can be credibly assumed about the performance of $r$, one cannot reasonably use it to identify performance of $t$. This is not a novel observation, but it highlights the importance of the assumption. \citet{emerson2018biomarker} state: ``\textit{If very little is known about the reference test performance, then it is clear that a comparison to such a reference test is a futile exercise and can provide no information about a new test.}'' \citet{gart1966comparison2} similarly note that when $(s_1,s_0)$ are unknown, both $t\independent y$ or $t\notindependent y$ will generally be consistent with $P(t,r)$. The choice of $(s_1,s_0)$ is context-specific, and should be carefully considered in each study.

I further assume that $s_1>1-s_0$, or that the reference test is reasonable.\footnote{The assumption does not require that both $s_1$ and $s_0$ are high. Indeed, it is possible that either $s_1$ or $s_0$ are close to $0$, but that their sum is higher than $1$.} If $s_1=1-s_0$, one can show that $r\independent y$, so the test provides no information on $y$. Tests are costly, and any use of such test is not rational. If $s_1<1-s_0$, it would be possible to redefine $r^* = 1-r$, so that $s_1^*=1-s_1$ and $s_0^*=1-s_0$. Now $s^*_1 >1-s^*_0$, since $1-s_1>s_0$.

\begin{assumption}{(Bounded Prevalence)}\label{prevalence_assumption}
Population prevalence $P(y=1)$ satisfies $0<P(y=1)<1$.
\end{assumption}

The assumption is implicitly found in all diagnostic test performance studies measuring sensitivity and specificity, since it is necessary for them to be defined. Assumptions \ref{reference_performance_assumption} and \ref{prevalence_assumption} imply that $P(r=1)\in(1-s_0,s_1)$. If the condition fails, at least one of the two assumptions are refuted.

If one additionally maintains that $t\independent r|y$, then $(\theta_1,\theta_0)$ are point identified (\citet{gart1966comparison2}, \citet{staquet1981methodology}, and \citet{zhou2009statistical}). However, it is well established that conditional independence is generally untenable (\citet{valenstein1990evaluating}, \citet{hui1998evaluation} and \citet{emerson2018biomarker}). Dependence may arise $t$ and $r$ are physiologically related, such as when they rely on the same type of sample or measure the same quantities. For example, tine and Mantoux tests may be dependent since they both rely on the antibody reaction to tuberculin (\citet{vacek1985effect}), and direct immunoassay and culture swab tests for \textit{Group A streptococci} may be related since they rely on the same type of sample (\citet{valenstein1990evaluating}). Since $y$ is unobserved, the dependence structure is latent, and multiple structures may be consistent with the data distribution $P(t,r)$. This leads to a possibly non-singleton set of values $(\theta_1,\theta_0)$ that are consistent with the data, called the identified set. We would first like to learn this set without imposing any assumptions on the statistical dependence structure between $t$ and $r$ conditional on $y$. Additional assumptions on the possible dependence structures may then be used to reduce the size of the identified set, as shown in \Cref{misclassification_assumptions_sect}.

\subsection{Identified Set for \texorpdfstring{$(\theta_1,\theta_0)$}{.}}

The data reveal $P(t,r)$, while probability distributions involving $y$ are not directly observable. Still, $P(r,y)$ can be determined using $(s_1,s_0)$ and $P(t,r)$. I henceforth use $P_{s_1,s_0}$ to denote probability distributions that are derived from observable distributions given $(s_1,s_0)$. All directly observable distributions, such as $P(t,r)$, do not have the subscript. By the law of total probability and $s_1\neq 1-s_0$ from Assumption \ref{reference_performance_assumption}:
\begin{align}\label{prevalence_expression} 
    P(r=1) &= s_1 P_{s_1,s_0}(y=1)+(1-s_0)P_{s_1,s_0}(y=0)\Leftrightarrow P_{s_1,s_0}(y=1)=\frac{P(r=1) +s_0 - 1}{s_1+s_0-1}.    
\end{align} 

$P_{s_1,s_0}(r,y)$ is then known from $P_{s_1,s_0}(r,y) = P_{s_1,s_0}(r|y)P_{s_1,s_0}(y)$, since $(s_1,s_0)$ fully characterize $P_{s_1,s_0}(r|y)$. To outline the idea of finding the identified set, first note that for $j=0,1$:
\begin{equation}\label{theta_decomposition}
    \theta_j =  P_{s_1,s_0}(t = j|y=j)= \frac{P_{s_1,s_0}(t = j,r=0,y=j)+P_{s_1,s_0}(t = j,r=1,y=j)}{P_{s_1,s_0}(y=j)}.
\end{equation}

Probabilities $P_{s_1,s_0}(t = j,r=k,y=j)$ for $k=0,1$ are unobservable. However, they can be bounded using the knowledge of $P(t,r)$ and $P_{s_1,s_0}(r,y)$. By the properties of probability measures, an upper bound on $P_{s_1,s_0}(t = j,r=k,y=j)$ is $min\Big(P(t = j, r=k), P_{s_1,s_0}(r=k,y=j)\Big)$. To form a lower bound, one can similarly find that $P_{s_1,s_0}(t = j,r=k,y=1-j)\leq min\Big(P(t = j, r=k), P_{s_1,s_0}(r=k,y=1-j)\Big)$ and use:
\begin{align}
\begin{split}
    P_{s_1,s_0}(t = j,r=k,y=j) &= P(t = j,r=k)-P_{s_1,s_0}(t=j,r=k,y=1-j)\\
    &\geq max\Big(0,P(t=j,r=k)-P_{s_1,s_0}(r=k,y=1-j)\Big).
\end{split}
\end{align}
Note that these coincide with Fréchet-Hoeffding bounds for $P(t=j,y=j|r=k)$ multiplied by $P(r=k)$. Proof of \Cref{bounds_theta_prop} demonstrates that any pair of values for $P_{s_1,s_0}(t = j,r=0,y=j)$ and $P_{s_1,s_0}(t = j,r=1,y=j)$ within their respective bounds is consistent with the observed data. Hence, sharp bounds on $P_{s_1,s_0}(t=j,y=j)$ are obtained by summing the two individual set of bounds. The sharp bounds for $\theta_j$ then follow directly from \eqref{prevalence_expression} and \eqref{theta_decomposition}. Finally, the sharp joint identified set for $(\theta_1,\theta_0)$ is derived using $P(t=1)=\theta_1 P_{s_1,s_0}(y=1)+(1-\theta_0)P_{s_1,s_0}(y=0)$. Observe that no restrictions beyond those set by the distribution $P(t,r)$ are imposed on the latent dependence structure of $t$ and $r$ conditional on $y$. 

\begin{theoremEnd}[restate, proof at the end, no link to proof]{proposition}\label{bounds_theta_prop}
The sharp identified set $\mathcal{H}_{(\theta_1,\theta_0)}(s_1,s_0)$ for $(\theta_1,\theta_0)$ given reference test sensitivity $s_1$ and specificity $s_0$ is:
\begin{equation}\label{bounds_joint}
    \mathcal{H}_{(\theta_1,\theta_0)}(s_1,s_0) =
        \Bigg\{(t_1,t_0):
        t_0=t_1\frac{P_{s_1,s_0}(y=1)}{P_{s_1,s_0}(y=0)}+1-\frac{P(t=1)}{P_{s_1,s_0}(y=0)},t_j\in\mathcal{H}_{\theta_j}(s_1,s_0)\Bigg\}
\end{equation}
where $\mathcal{H}_{\theta_j}(s_1,s_0)=[\theta_j^L,\theta_j^U]$ is the sharp bound on $\theta_j$ defined as:
\begin{align}\label{bounds_theta_eq}
\begin{split}
    \theta_j^L =&\Bigg[max\Big(0,P(t=j,r=j)-P_{s_1,s_0}(r=j,y=1-j)\Big)\\
    &+max\Big(0,P_{s_1,s_0}(r=1-j,y=j)-P(t=1-j,r=1-j)\Big)\Bigg]\frac{1}{P_{s_1,s_0}(y=j)}\\
    \theta_j^U =&\Bigg[min\Big(P(t = j, r=1-j), P_{s_1,s_0}(r=1-j,y=j)\Big)\\
    &+min\Big(P(t = j, r=j), P_{s_1,s_0}(r=j,y=j)\Big)\Bigg]\frac{1}{P_{s_1,s_0}(y=j)}.
\end{split}
\end{align}
\end{theoremEnd}
\begin{proofEnd}

Alternative proofs can be constructed using Theorem 3.10 from \citet{joe1997multivariate} or Artstein's inequalities (\citet{beresteanu2012partial}).\footnote{I am grateful to an anonymous referee and Gabriel Ziegler for bringing this to my attention.} Here, I offer a direct proof that follows through a series of claims. Intermediate results will be used to prove other propositions. First we derive bounds on $P_{s_1,s_0}(t=j,r=k,y=l)$ for $(j,k,l)\in\{0,1\}^3$. We then show that the pair of bounds on $P_{s_1,s_0}(t=j,r=j,y=j)$ and $P_{s_1,s_0}(t=j,r=1-j,y=j)$ for a fixed $j$ are sharp and that any two points in these bounds are attainable simultaneously. The sharp bound on $\theta_j$ follows by summing the individual bounds and dividing by $P_{(s_1,s_0)}(y=1)$. Finally, the sharp joint identified set for $(\theta_1,\theta_0)$ is immediate by the law of total probability.

\begin{claim}
Bounds on $P_{s_1,s_0}(t=j,r=k,y=l)$ for any $(j,k,l)\in\{0,1\}^3$ are:
\begin{equation}
\begin{split}\label{bounds_on_probability}
    P_{s_1,s_0}(t=j,r=k,y=l)\in\Bigg[&max\Big(0,P_{s_1,s_0}(r=k,y=l)-P(t=1-j,r=k)\Big),\\
    &min\Big(P(t = j, r=k), P_{s_1,s_0}(r=k,y=l)\Big)\Bigg].    
\end{split}
\end{equation}
\end{claim}
\begin{proof}

 Probability $P_{s_1,s_0}(t=j,r=k,y=l)$ for any $(j,k,l)\in\{0,1\}^3$ is the probability of the intersection of events $P_{s_1,s_0}(\{t=j,r=k\}\cap\{r=k,y=l\})$. An upper bound on $P_{s_1,s_0}(t=j,r=k,y=l)$ is then:
\begin{align}\label{upper_bound_on_theta}
  P_{s_1,s_0}(\{t=j,r=k\}\cap\{r=k,y=l\})&\leq min\Big(P(t = j, r=k), P_{s_1,s_0}(r=k,y=l)\Big).
\end{align}
The upper bound \eqref{upper_bound_on_theta} holds for any $(j,k,l)\in\{0,1\}^3$. Using the upper bound on $P_{s_1,s_0}(t=j,r=k,y=1-l)$, the lower bound on $P_{s_1,s_0}(t=j,r=k,y=l)$ is:
\begin{equation}
    \begin{split}\label{lower_bound_on_theta}
       P_{s_1,s_0}(t=j,r=k,y=l) &= P(t=j,r=k) - P_{s_1,s_0}(t=j,r=k,y=1-l)\\
       &\geq P(t=j,r=k)- min\Big(P(t = j, r=k), P_{s_1,s_0}(r=k,y=1-l)\Big)\\
       &= max\Big(0,P(t=j,r=k) - P_{s_1,s_0}(r=k,y=1-l)\Big)\\
       & = max\Big(0,P_{s_1,s_0}(r=k,y=l)-P(t=1-j,r=k)\Big)
    \end{split}
\end{equation}
where the final line of \eqref{lower_bound_on_theta} follows from \Cref{lemma_trick_equivalence}.
\end{proof}
\begin{claim}\label{claim_independent_bounds}
Bounds \eqref{upper_bound_on_theta} on $P_{s_1,s_0}(t=j,r=j,y=j)$ and $P_{s_1,s_0}(t=j,r=1-j, y=j)$ are sharp. Bounds are independent in the sense that any pair of points within the two bounds is attainable.
\end{claim}
\begin{proof}
Write all eight joint and observable probabilities as a matrix equation:
\begin{align}\label{matrix}
    \begin{split}
\underbrace{
\left(
\begin{array}{cccccccc}
 1 & 1 & 0 & 0 & 0 & 0 & 0 & 0 \\
 1 & 0 & 1 & 0 & 0 & 0 & 0 & 0 \\
 0 & 0 & 1 & 1 & 0 & 0 & 0 & 0 \\
 0 & 1 & 0 & 1 & 0 & 0 & 0 & 0 \\
 0 & 0 & 0 & 0 & 1 & 1 & 0 & 0 \\
 0 & 0 & 0 & 0 & 1 & 0 & 1 & 0 \\
 0 & 0 & 0 & 0 & 0 & 0 & 1 & 1 \\
 0 & 0 & 0 & 0 & 0 & 1 & 0 & 1 \\
\end{array}
\right)}_{\mathbf{A}}
\left(
\begin{array}{cccccccc}
P_{s_1,s_0}(t=1,r=1,y=1)\\
P_{s_1,s_0}(t=1,r=1,y=0)\\
P_{s_1,s_0}(t=0,r=1,y=1)\\
P_{s_1,s_0}(t=0,r=1,y=0)\\
P_{s_1,s_0}(t=1,r=0,y=1)\\
P_{s_1,s_0}(t=1,r=0,y=0)\\
P_{s_1,s_0}(t=0,r=0,y=1)\\
P_{s_1,s_0}(t=0,r=0,y=0)\\
\end{array}
\right)=&
\left(
\begin{array}{cccccccc}
P(t=1,r=1)\\
P_{s_1,s_0}(r=1,y=1)\\
P(t=0,r=1)\\
P_{s_1,s_0}(r=1,y=0)\\
P(t=1,r=0)\\
P_{s_1,s_0}(r=0,y=1)\\
P(t=0,r=0)\\
P_{s_1,s_0}(r=0,y=0)\\
\end{array}
\right).
    \end{split}
\end{align}
Matrix $\mathbf{A}$ has rank $6$. The bottom four rows cannot be represented as a linear combination using any of the top four rows. The bottom four rows are only mutually linearly dependent. Similarly, the top four rows are only mutually linearly dependent. Therefore, the value of $P_{s_1,s_0}(t=j,r=1-k,y=l)$ does not affect the values of $P_{s_1,s_0}(t=j,r=k,y=l)$ for $(j,l)\in\{0,1\}^2$ within their respective bounds. 

There exist two separate systems of equations, one for each value of $r$. Focus on one system for an arbitrary $r=k$:
\begin{align}\label{matrix_reduced}
    \begin{split}
\underbrace{
\left(
\begin{array}{cccc}
 1 & 1 & 0 & 0  \\
 1 & 0 & 1 & 0 \\
 0 & 0 & 1 & 1 \\
 0 & 1 & 0 & 1  \\
\end{array}
\right)}_{\mathbf{A}'}
\left(
\begin{array}{cccc}
P_{s_1,s_0}(t=j,r=k,y=j)\\
P_{s_1,s_0}(t=j,r=k,y=1-j)\\
P_{s_1,s_0}(t=1-j,r=k,y=j)\\
P_{s_1,s_0}(t=1-j,r=k,y=1-j)\\
\end{array}
\right)=&
\left(
\begin{array}{cccc}
P(t=j,r=k)\\
P_{s_1,s_0}(r=k,y=j)\\
P(t=1-j,r=k)\\
P_{s_1,s_0}(r=k,y=1-j)\\
\end{array}
\right).
    \end{split}
\end{align}

Matrix $\mathbf{A}'$ has rank $3$. I first show that both the upper and lower bounds on any of the joint probabilities $P_{s_1,s_0}(t=j,r=k,y=l)$ in \eqref{matrix_reduced} are attainable for $(j,l)\in\{0,1\}^2$.\footnote{Attainable (or equivalently feasible) is meant in the sense that it is consistent with the observed distribution $P(t,r)$ and assumed values for $(s_1,s_0)$.} Then, I demonstrate that any value in the interior of the bounds is attainable, proving that bounds on $P_{s_1,s_0}(t=j,r=k,y=l)$ are sharp. Focus on $P_{s_1,s_0}(t=j,r=k,y=j)$. Assume that it is equal to its upper bound, $P_{s_1,s_0}(t=j,r=k,y=j) = min\Big(P(t = j, r=k), P_{s_1,s_0}(r=k,y=j)\Big)$. Let first $P(t = j, r=k)< P_{s_1,s_0}(r=k,y=j)$. From \Cref{lemma_trick_equivalence}, $P_{s_1,s_0}(r=k,y=1-j)<P(t=1-j,r=k)$. Then from \eqref{matrix_reduced}:
\begin{align}\label{solving_sharpness}
\begin{split}
&P_{s_1,s_0}(t=j,r=k,y=j) = P(t = j, r=k)\\
&P_{s_1,s_0}(t=j,r=k,y=1-j) = 0\\
&P_{s_1,s_0}(t=1-j,r=k,y=j) = P_{s_1,s_0}(r=k,y=j) - P(t=j,r=k)\\
&P_{s_1,s_0}(t=1-j,r=k,y=1-j) = P_{s_1,s_0}(r=k,y=1-j)    .
\end{split}
\end{align}
By assumption, $P_{s_1,s_0}(t=j,r=k,y=j)$ is equal to its upper bound. Consequently, $P_{s_1,s_0}(t=j,r=k,y=1-j)$ is equal to $0=max(0,P_{s_1,s_0}(r=k,y=1-j)-P(t=1-j,r=k))$ which is its lower bound. Similarly, $P_{s_1,s_0}(t=1-j,r=k,y=j) = P_{s_1,s_0}(r=k,y=j) - P(t=j,r=k) = max(0,P_{s_1,s_0}(r=k,y=j) - P(t=j,r=k)) $, which is its lower bound. Finally, $P_{s_1,s_0}(t=1-j,r=k,y=1-j) = P_{s_1,s_0}(r=k,y=1-j)=min(P(t=1-j,r=k),P_{s_1,s_0}(r=k,y=1-j))$, representing the upper bound. All four probabilities achieve their corresponding upper and lower bounds. 

Let now $P(t = j, r=k)\geq P_{s_1,s_0}(r=k,y=j)$, or equivalently $P_{s_1,s_0}(r=k,y=1-j)\geq P(t=1-j,r=k)$. The system then is:
\begin{align}\label{solving_sharpness_2}
\begin{split}
&P_{s_1,s_0}(t=j,r=k,y=j) = P_{s_1,s_0}(r=k,y=j)\\
&P_{s_1,s_0}(t=j,r=k,y=1-j) = P(t=j,r=k)-P_{s_1,s_0}(r=k,y=j)\\
&P_{s_1,s_0}(t=1-j,r=k,y=j) =  0\\
&P_{s_1,s_0}(t=1-j,r=k,y=1-j) = P(t=1-j,r=k).
\end{split}
\end{align}
As before, $P_{s_1,s_0}(t=j,r=k,y=j)$ and $P_{s_1,s_0}(t=1-j,r=k,y=1-j)$ are equal to their respective upper bounds. $P_{s_1,s_0}(t=j,r=k,y=1-j)$ and $P_{s_1,s_0}(t=1-j,r=k,y=j)$ attain the lower bounds. That $P_{s_1,s_0}(t=j,r=k,y=j)$ and $P_{s_1,s_0}(t=1-j,r=k,y=1-j)$ attain lower bounds when $P_{s_1,s_0}(t=j,r=k,y=1-j)$ and $P_{s_1,s_0}(t=1-j,r=k,y=j)$ are equal to their upper bounds can be shown symmetrically. Thus, for an arbitrary $r=k$, all probabilities can be equal to their upper and lower bounds.

From \eqref{matrix_reduced}, reducing any probability that is on the upper bound will lead to an increase in the probabilities at lower bounds and a decrease in the remaining probability at the upper bound. Any value in the interior of the bounds must be feasible. Therefore, the bounds \eqref{bounds_on_probability} must be sharp for $P(t=j,r=k,y=l)$ and any $(j,l)\in\{0,1\}^2$. This is true for an arbitrary $r=k$, hence the bounds are sharp for any $P(t=j,r=k,y=l)$ such that $(j,k,l)\in\{0,1\}^3$.

Finally, from \eqref{matrix}, the value which $P_{s_1,s_0}(t=j,r=j,y=j)$ takes does not influence the value of $P_{s_1,s_0}(t=j,r=1-j,y=j)$. Any pair of values coming from the Cartesian product of the bounds on the two probabilities is feasible.
 \end{proof}
By Claim \ref{claim_independent_bounds}, the sharp bounds on $P_{s_1,s_0}(t=j,y=1) =  P_{s_1,s_0}(t=j,r=j,y=j)+P_{s_1,s_0}(t=j,r=1-j,y=j)$ are a sum of the sharp bounds on individual probabilities. Hence, the sharp bounds on $\theta_j$ are:
\begin{align}\label{bounds_on_theta_prop}
\begin{split}
    \theta_j &\geq  \frac{max\Big(0,P(t=j,r=j)-P_{s_1,s_0}(r=j,y=1-j)\Big)}{P_{s_1,s_0}(y=j)}\\
    &+\frac{max\Big(0,P_{s_1,s_0}(r=1-j,y=j)-P(t=1-j,r=1-j)\Big)}{P_{s_1,s_0}(y=j)}\\
    \theta_j &\leq \frac{min\Big(P(t = j, r=1-j), P_{s_1,s_0}(r=1-j,y=j)\Big)+min\Big(P(t = j, r=j), P_{s_1,s_0}(r=j,y=j)\Big)}{P_{s_1,s_0}(y=j)}
\end{split}
\end{align}

\begin{claim}\label{claim_sharp_joint_bounds}
The sharp joint identified set for $(\theta_1,\theta_0)$ is:
\begin{equation*}
    \mathcal{H}_{(\theta_1,\theta_0)}(s_1,s_0) =
        \Bigg\{(t_1,t_0):
        t_0=t_1\frac{P_{s_1,s_0}(y=1)}{P_{s_1,s_0}(y=0)}+1-\frac{P(t=1)}{P_{s_1,s_0}(y=0)},t_1\in\mathcal{H}_{\theta_j}(s_1,s_0)\Bigg\}.
\end{equation*}
\end{claim}
\begin{proof}
\begin{align}
    \begin{split}
        P(t=1) &= P_{s_1,s_0}(t=1,y=1)+P_{s_1,s_0}(t=1,y=0) =\\
        &=\theta_1P_{s_1,s_0}(y=1)+P_{s_1,s_0}(y=0)-\theta_0P_{s_1,s_0}(y=0).
    \end{split}
\end{align}
Set $j=1$ without loss of generality. For any value $t_1\in \mathcal{H}_{\theta_1}(s_1,s_0)$, it must be that $t_0P_{s_1,s_0}(y=0) =t_1P_{s_1,s_0}(y=1)+P_{s_1,s_0}(y=0)-P(t=1)$. Since $\mathcal{H}_{\theta_1}(s_1,s_0)$ is sharp, $\mathcal{H}_{(\theta_1,\theta_0)}(s_1,s_0)$ is a sharp joint identification region for $(\theta_1,\theta_0)$.
\end{proof}
\end{proofEnd}

\Cref{bounds_theta_prop} revitalizes the general identification result of \citet{cross2002regressions}. Arguments therein can also be used to derive the identified set for $(\theta_1,\theta_0)$.\footnote{I thank an anonymous referee for bringing this to my attention.} Namely, $P(t|r)$ and $P_{s_1,s_0}(y|r)$ are identified from $P(t,r)$ and $(s_1,s_0)$, yielding the identified set for the vector $\left(E[t|r=j,y=k]\right)_{(j,k)\in\{0,1\}^2}$.\footnote{$P_{s_1,s_0}(y=j|r=1-j) = 0$ when $s_j=1$ for some $j$, which violates the assumptions of \citet{cross2002regressions}. It is still possible to utilize their results by bounding only $\left(E[t|r=j,y=1],E[t|r=j,y=0]\right)$.} Sharp bounds on $(\theta_1,\theta_0) = (E[t|y=1],1-E[t|y=0])$ follow. This paper relies on a constructive proof approach which can be adapted to accommodate assumptions on the latent dependence structure of $t$ and $r$ conditional on $y$. \Cref{misclassification_assumptions_sect} introduces a novel formalization of an informally stated assumption from the literature and exploits this feature to further tighten the bounds.

The set $\mathcal{H}_{(\theta_1,\theta_0)}(s_1,s_0)$ is a line segment on $[0,1]^2$ for a given value of reference test operating characteristics $s_1$ and $s_0$. It collapses to a point $(\Tilde{\theta}_1,\Tilde{\theta}_0)$ when $(s_1,s_0)=(1,1)$. \citet{emerson2018biomarker} sketch an argument for individual bounds on $\theta_j$ as in \eqref{bounds_theta_eq} and do not discuss sharpness or the joint identified set. \Cref{bounds_theta_prop} goes further by proving that both individual bounds and the joint identified sets are the smallest possible under the assumptions. \Cref{sect_bounding_prevalence} shows that the linear structure of the set $\mathcal{H}_{(\theta_1,\theta_0)}(s_1,s_0)$ is crucial for sharpness of bounds on certain derived policy-relevant parameters, such as the illness rate in populations screened with the test $t$, otherwise known as prevalence. Bounds on prevalence are unnecessarily wide if any pair of values $(\theta_1,\theta_0)$ from their individual bounds is considered feasible, so that the joint identified set is a rectangle $\mathcal{H}_{\theta_1}(s_1,s_0)\times \mathcal{H}_{\theta_0}(s_1,s_0)$. The sharp joint identified set for false negative and false positive rates $(1-\theta_1,1-\theta_0)$ also directly follows from $\mathcal{H}_{(\theta_1,\theta_0)}(s_1,s_0)$. The same will hold for other identified sets for $(\theta_1,\theta_0)$ below.

\begin{example}\label{example_1}
Consider a study in which $(s_1,s_0)=(0.9,0.9)$, $P(t=j,r=j)=0.45$ and $P(t=j,r=1-j)=0.05$ for $j=0,1$. $\mathcal{H}_{(\theta_1,\theta_0)}(s_1,s_0)$ is a line segment connecting $(0.8,0.8)$ and $(1,1)$.
\end{example}

The identified set for $(\theta_1,\theta_0)$ is sharp. Encountering wide bounds on sensitivity and specificity implies that is not possible to learn the operating characteristics more precisely without additional assumptions that may be untenable, or without changing the reference test. Since the reference test is supposed to be the best available test, researchers and practitioners may have to embrace the ambiguity regarding the index test performance.

\begin{remark}\label{remark_better_index}
 Conventional studies maintain that the reference test is perfect $(s_1,s_0) = (1, 1) \geq (\theta_1,\theta_0)$ component-wise, so $t$ is irrefutably assumed to perform at most as well as $r$ in both dimensions. The bounds allow one to empirically compare index and reference test performance for certain $P(t,r)$ and $(s_1,s_0)$. This is possible in the following cases:
\begin{enumerate}
    \item $(\theta_1,\theta_0)<(s_1,s_0)$ component-wise for all $(\theta_1,\theta_0)\in \mathcal{H}_{(\theta_1,\theta_0)}(s_1,s_0)$, demonstrating that $r$ outperforms $t$ in both dimensions;
    \item $\theta_j>s_j$ for a single $j$ and all $(\theta_1,\theta_0)\in \mathcal{H}_{(\theta_1,\theta_0)}(s_1,s_0)$, demonstrating that $t$ outperforms $r$ in one dimension.
\end{enumerate}

\Cref{lem:comparison} in Appendix \ref{sect_aux_results} shows that $(\theta_1,\theta_0)>(s_1,s_0)$ component-wise is impossible for all $(\theta_1,\theta_0)\in \mathcal{H}_{(\theta_1,\theta_0)}(s_1,s_0)$. That is, a single study cannot show that $t$ outperforms $r$ in both dimensions. As illustrated by \Cref{example_1}, it is also possible that there exist $(\theta_1,\theta_0),(\theta'_1,\theta'_0)\in \mathcal{H}_{(\theta_1,\theta_0)}(s_1,s_0)$ such that $(\theta_1,\theta_0)>(s_1,s_0)$ and $(\theta'_1,\theta'_0)\leq(s_1,s_0)$ component-wise. This indicates that $t$ or $r$ may outperform the other test in both dimensions, but it is inconclusive. One cannot exclude the possibility that either test outperforms the other.
\end{remark}

\begin{remark}
Depending on $(s_1,s_0)$ and $P(t,r)$, ``apparent'' measures $(\tilde{\theta}_1,\tilde{\theta}_0)$ need not be contained in the identified set for $(\theta_1,\theta_0)$. In that sense, $(\tilde{\theta}_1,\tilde{\theta}_0)$ may be over- or understating $(\theta_1,\theta_0)$. A relevant empirical example is found in \Cref{application_sect}.
\end{remark}

\begin{remark}\label{remark_fda}
The FDA Statistical Guidance defines a \textit{reference standard} for a condition as: \textit{``The best available method for establishing the presence or absence of the target condition. ... established by opinion and practice within the medical, laboratory, and regulatory community."} The guidance does not require a reference standard to be perfect, as it rarely is. When used as a reference test, the estimates may be reported as pertaining to sensitivity and specificity, even though the estimands are ``apparent'' measures when it is imperfect. This practice can be misleading. Tests other than the reference standard may be used as reference tests. However, then the estimates should be reported as ``apparent''. If one wishes to learn true test performance, they are typically not of interest. 

The guidance does not require or suggest any corrections that would allow researchers to form adequate estimates of the true operating characteristics in either case. The method in this paper proposes a solution by forming the smallest possible bounds on true performance measures under standard assumptions. Furthermore, the guidance emphasizes that in a conventional study one cannot determine whether $t$ or $r$ has better performance. \Cref{remark_better_index} clarifies that bounds allow for comparison of performance between the two tests in certain cases.
\end{remark}

\begin{remark}\label{remark_ordinary_frechet}
Bounds in \eqref{bounds_theta_eq} could be formed from the marginals $P(t=j)$ and $P_{s_1,s_0}(y=j)$ as $\theta_j\in\Big[max\Big(0,P(t=j)+P_{s_1,s_0}(y=1-j)\Big),min\Big(P(t=j),P_{s_1,s_0}(y=j)\Big)\Big]\frac{1}{P_{s_1,s_0}(y=j)}$.
The literature on data combination suggests that these are not sharp, as outlined by \citet{ridder2007econometrics}. \Cref{lemma_ordinary_frechet} in Appendix \ref{sect_aux_results} shows that they are at least as wide as those in \Cref{bounds_theta_prop}.
\end{remark}

\begin{subsection}{Misclassification Assumptions}\label{misclassification_assumptions_sect}

Points in the identified set $\mathcal{H}_{(\theta_1,\theta_0)}(s_1,s_0)$ derived in the previous section correspond to different non-observable probability distributions $P_{s_1,s_0}(t,r,y)$ that are consistent with the identified distribution $P(t,r)$ and $(s_1,s_0)$. Until this point no additional restrictions on the dependence structure between $t$, $r$ and $y$ were imposed.
Literature on gold standard bias suggests that $t$ and $r$ may frequently be dependent conditional on $y$ in ways that would further restrict the set of distributions $P_{s_1,s_0}(t,r,y)$ consistent with the data, resulting in more informative identified sets for $(\theta_1,\theta_0)$. It is thus important to incorporate assumptions on the dependence structure into the framework.

A particular kind of restrictions that researchers may be willing to consider concern the error probabilities of $t$ conditional on $r$ making a misclassification error for a specific value of $y$. Researchers may scrutinize the credibility of such assumptions based on physical properties of the two tests. \citet{valenstein1990evaluating} informally discusses one such restriction. The author analyzes the magnitude of the difference $\theta_j-\tilde{\theta}_j$ for $j=0,1$ by means of a numerical example when the two tests have classification errors that are referred to as ``highly correlated''. The meaning of highly correlated errors is not formally defined, and in the numerical example the assumption is imposed as $P(t\neq y|r\neq y,y)=P(t=1-y|r=1- y,y)=1$ for all $y$. I formalize this assumption and derive the resulting sharp identified set for $(\theta_1,\theta_0)$. Given that its plausibility may vary across health statuses, I allow it to hold only for a particular value of $y$. 

\begin{definition}{(Tendency to wrongly agree)}\label{tendency_def}
An index test has a tendency to wrongly agree with the reference test for disease status $\bar{y}$ given $(s_1,s_0)$ if $P_{s_1,s_0}(t=1-\bar{y}|r=1-\bar{y},y=\bar{y})\geq P_{s_1,s_0}(t=\bar{y}|r=1-\bar{y},y=\bar{y})$.
\end{definition}

If an index test exhibits a tendency to wrongly agree with the reference test for $\bar{y}$, conditional on the reference test making a classification error, the index test is more likely to misdiagnose the patient than to diagnose them correctly. \citet{valenstein1990evaluating} explains that the tendency may arise if the two tests have common properties, such as the type of sample used, e.g. the same swab type.

\begin{theoremEnd}[restate, proof at the end, no link to proof]{proposition}\label{bounds_wrongly_agree_prop}
Let $\theta_j^L$ be as in \eqref{bounds_theta_eq}. When the index and reference tests have a tendency to wrongly agree only for $y=j$, the sharp bounds on $\theta_j$ given $(s_1,s_0)$ are $\Bar{\mathcal{H}}_{\theta_j}(s_1,s_0)=[\theta_j^L,\Bar{\theta}_j^U]$, where:
\begin{align}
\begin{split}\label{bounds_theta_agree_one}
    \Bar{\theta}_j^U =&\Bigg[min\Big(P(t = j, r=1-j), \frac{P_{s_1,s_0}(r=1-j,y=j)}{2}\Big)\\
    &+min\Big(P(t = j, r=j), P_{s_1,s_0}(r=j,y=j)\Big)\Bigg]\frac{1}{P_{s_1,s_0}(y=j)}.
\end{split}
\end{align}
If the index and reference tests have a tendency to wrongly agree for $y=0$ and $y=1$, the sharp bounds on $\theta_j$ for $j=0,1$ given $(s_1,s_0)$ are $\dbar{\mathcal{H}}_{\theta_j}(s_1,s_0)=[\theta_j^L,\dbar{\theta}_j^U]$, where:
\begin{align}\label{bounds_theta_wrongly_agree_two}
\begin{split}
    \dbar{\theta}_j^U =&\Bigg[min\Big(P(t = j, r=1-j), \frac{P_{s_1,s_0}(r=1-j,y=j)}{2}\Big)\\
    &+min\Big(P(t = j, r=j) - \frac{P_{s_1,s_0}(r=j,y=1-j)}{2}, P_{s_1,s_0}(r=j,y=j)\Big)\Bigg]\frac{1}{P_{s_1,s_0}(y=j)}.
\end{split}
\end{align}
Sharp joint identified sets $\Bar{\mathcal{H}}_{(\theta_1,\theta_0)}(s_1,s_0)$ and $\dbar{\mathcal{H}}_{(\theta_1,\theta_0)}(s_1,s_0)$ for $(\theta_1,\theta_0)$ given $(s_1,s_0)$ follow from \eqref{bounds_joint}, $\Bar{\mathcal{H}}_{\theta_j}(s_1,s_0)$, and $\dbar{\mathcal{H}}_{\theta_j}(s_1,s_0)$.
\end{theoremEnd}
\begin{proofEnd}

First, I prove a lemma used below. The proof then follows through a series of claims.

\begin{lemma}\label{lemma_wrongly_agree}
The index test has a tendency to wrongly agree with the reference test for $y=j$ for a given $(s_1,s_0)$, if and only if $P_{s_1,s_0}(t=1-j,r=1-j,y=j) \geq \frac{P_{s_1,s_0}(r=1-j,y=j)}{2}$. 
\end{lemma} 
\begin{proof}
It holds that $P_{s_1,s_0}(t=1-j,r=1-j,y=j)+P_{s_1,s_0}(t=j,r=1-j,y=j) = P(r=1-j,y=j)$. For sufficiency, note that $2P_{s_1,s_0}(t=1-j,r=1-j,y=j)=P_{s_1,s_0}(r=1-j,y=j)-P_{s_1,s_0}(t=j,r=1-j,y=j)+P_{s_1,s_0}(t=1-j,r=1-j,y=j)\geq P_{s_1,s_0}(r=1-j,y=j)$, since by assumption $P_{s_1,s_0}(t=1-j,r=1-j,y=j)\geq P_{s_1,s_0}(t=j,r=1-j,y=j)$. Necessity is immediate. 
\end{proof}

\begin{claim}\label{claim_wrongly_agree_once}
Assume that the tests have a tendency to wrongly agree only for $y=j$. The sharp identified set for $(\theta_1,\theta_0)$ is $\Bar{\mathcal{H}}_{(\theta_1,\theta_0)}(s_1,s_0)$.
\end{claim}
\begin{proof}
 From \Cref{lemma_wrongly_agree}, $P_{s_1,s_0}(t=1-j,r=1-j,y=j) \geq \frac{P_{s_1,s_0}(r=1-j,y=j)}{2}$. Then, $P_{s_1,s_0}(t=j,r=1-j,y=j)\leq \frac{P_{s_1,s_0}(r=1-j,y=j)}{2}\leq P_{s_1,s_0}(r=1-j,y=j)$.  
Using this and following the steps taken to obtain \eqref{upper_bound_on_theta}:
\begin{equation}\label{upper_bound_wrongly_agree_once}
P_{s_1,s_0}(t=j,r=1-j,y=j)\leq min\Big(P(t = j, r=1-j), \frac{P_{s_1,s_0}(r=1-j,y=j)}{2}\Big).
\end{equation}

The lower bound on $P_{s_1,s_0}(t=j,r=1-j,y=j)$ is derived from the upper bound on $P_{s_1,s_0}(t=1-j,r=1-j,y=j)$ which is unaffected by the assumption. Substituting the upper bound into the system \eqref{matrix_reduced} yields the lower bound $P_{s_1,s_0}(t=j,r=1-j,y=j)\geq\max\Big(0, P_{s_1,s_0}(r=1-j, y=j)-P(t=1-j,r=1-j)\Big)$, as in \eqref{lower_bound_on_theta}. 

For the bounds defined by \eqref{lower_bound_on_theta} and \eqref{upper_bound_wrongly_agree_once} on $P_{s_1,s_0}(t=j,r=1-j,y=j)$ to be sharp, all values contained between them must be feasible for a given population distribution. The lower bound is identical as in \Cref{bounds_theta_prop}. The upper bound in \eqref{upper_bound_wrongly_agree_once} is at most as large as the upper bound \eqref{upper_bound_on_theta} in \Cref{bounds_theta_prop}. Thus, all points in the bounds on $P_{s_1,s_0}(t=j,r=1-j,y=j)$ are attainable by the same argument as in Claim \ref{claim_independent_bounds} in the proof of \Cref{bounds_theta_prop}. Hence, the bounds defined by \eqref{lower_bound_on_theta} and \eqref{upper_bound_wrongly_agree_once} are sharp. Sharp bounds on probabilities $P_{s_1,s_0}(t=k,r=j,y=l)$ from \eqref{bounds_on_probability} are unaffected by the assumption for $(k,l)\in\{0,1\}^2$ as they form an independent system of equations from \eqref{matrix}. Using the reasoning in Claims \ref{claim_independent_bounds}, and \ref{claim_sharp_joint_bounds} of \Cref{bounds_theta_prop}, $\Bar{\mathcal{H}}_{(\theta_1,\theta_0)}(s_1,s_0)$ is a sharp identification region for $(\theta_1,\theta_0)$.

\end{proof}

\begin{claim}
Assume that the tests have a tendency to wrongly agree for $y=0$ and $y=1$. The sharp identified set for $(\theta_1,\theta_0)$ is $\dbar{\mathcal{H}}_{(\theta_1,\theta_0)}(s_1,s_0)$.
\end{claim}

\begin{proof}
By \Cref{lemma_wrongly_agree}, $P_{s_1,s_0}(t=1-j,r=1-j,y=j) \geq \frac{P_{s_1,s_0}(r=1-j,y=j)}{2}$ for $j\in\{0,1\}$. The sharp upper bound on $P_{s_1,s_0}(t=j,r=1-j,y=j)$ is again as in \eqref{upper_bound_wrongly_agree_once}. The sharp upper bound on $P_{s_1,s_0}(t=j,r=j,y=j)$ is no longer equivalent to \eqref{upper_bound_on_theta}. Analogously to the steps used to derive \eqref{upper_bound_wrongly_agree_once}:
\begin{equation}\label{upper_bound_wrongly_agree_twice}
P_{s_1,s_0}(t=j,r=j,y=j)\leq min\Big(P(t = j, r=j) - \frac{P_{s_1,s_0}(r=j,y=1-j)}{2}, P_{s_1,s_0}(r=j,y=j)\Big),
\end{equation}
where the first value in the minimum is derived using \Cref{lemma_wrongly_agree} and:
\begin{align}
    \begin{split}
        P_{s_1,s_0}(t=j,r=j,y=j) &= P(t = j, r=j)-P_{s_1,s_0}(t=j,r=j,y=1-j)\\
        &\leq P(t = j, r=j) - \frac{P_{s_1,s_0}(r=j,y=1-j)}{2}.
    \end{split}
\end{align}
\begin{remark}
Only the upper bounds on $P_{s_1,s_0}(t=j, r=1-j, y=j)$ and $P_{s_1,s_0}(t=j, r=j, y=j)$ are changed by the assumption that tests have a tendency to wrongly agree for $y\in\{0,1\}$. The lower bounds remain as in \eqref{lower_bound_on_theta}.
\end{remark}

To see this, observe from \eqref{matrix} that the bounds on $P_{s_1,s_0}(t=j, r=1-j, y=j)$ and $P_{s_1,s_0}(t=j, r=j, y=j)$ belong to separate systems of equations and will not affect each other. The bounds on $P_{s_1,s_0}(t=j, r=1-j, y=j)$ hold as in the Claim \ref{claim_wrongly_agree_once}. The bounds on $P_{s_1,s_0}(t=j, r=j, y=j)$ are derived using $P_{s_1,s_0}(t=j,r=j, y=1-j)$ which is affected only from below by the assumption. From \eqref{matrix_reduced} it can be seen that substituting $P_{s_1,s_0}(t=j,r=j, y=1-j)$ with its upper bound $\min\Big(P(t=j,r=j), P_{s_1,s_0}(r=j,y=1-j)\Big)$ yields an identical lower bound for $P_{s_1,s_0}(t=j, r=j, y=j)$ as in \eqref{lower_bound_on_theta}.

Bounds \eqref{lower_bound_on_theta} and \eqref{upper_bound_wrongly_agree_once} on $P_{s_1,s_0}(t=j, r=1-j, y=j)$ were shown to be sharp in the previous claim. Using the same argument, bounds  \eqref{lower_bound_on_theta} and \eqref{upper_bound_wrongly_agree_twice} on $P_{s_1,s_0}(t=j, r=j, y=j)$ are also sharp. Any pair of points in the bounds for the two probabilities is feasible. Hence, $\dbar{\mathcal{H}}_{(\theta_1,\theta_0)}(s_1,s_0)$ is the sharp identified set for $(\theta_1,\theta_0)$.

\end{proof}
\end{proofEnd}
\end{subsection}

\Cref{bounds_wrongly_agree_prop} provides sharp identified sets for $(\theta_1,\theta_0)$ when the researcher maintains that the tests have a tendency to wrongly agree for only one or both health statuses.\footnote{One can also define the tendency to correctly disagree for disease status $\bar{y}$ as $P_{s_1,s_0}(t=1-\bar{y}|r=1-\bar{y},y=\bar{y})\leq P_{s_1,s_0}(t=\bar{y}|r=1-\bar{y},y=\bar{y})$. Identified sets that follow can easily be derived symmetrically. \citet{thibodeau1981evaluating} emphasizes that tests are generally not expected to exhibit negative dependence. However, the formulation may be beneficial in applications described in \Cref{sect_broader_applications}.} Both sets given $(s_1,s_0)$ are again line segments in $[0,1]^2$. The bounds $[\theta_j^L,\Bar{\theta}_j^U]$, and $[\theta_j^L,\dbar{\theta}_j^U]$ imply that the sets may be reduced in size only from above compared to $\mathcal{H}_{(\theta_1,\theta_0)}(s_1,s_0)$. This can be seen in \Cref{example_2} below.

It is important to note that the assumption may or may not have identifying power for a given $P(t,r)$ and $(s_1,s_0)$. This is evident in the empirical application. \Cref{remark_no_power_application} notes that the assumption effectively halves the size of the estimated identified set in one population, but has no effect in the remaining two. \Cref{remark_identifying_power} characterizes sufficient and necessary conditions for the assumption to have identifying power, and provides an easily verifiable necessary condition. 

\begin{remark}\label{remark_identifying_power}
\Cref{lem:ident_power} in Appendix \ref{sect_aux_results} shows that the tendency to wrongly agree for $y=j$ has identifying power if and only if $P(t=j,r=1-j)>\frac{P_{s_1,s_0}(r=1-j,y=j)}{2}>0$. If $s_j=1$, the assumption cannot have identifying power.
\end{remark}

For the purpose of interpreting this result, suppose that tendency to wrongly agree holds for $y=j$. If $s_j=1$, $\{r=1-j,y=j\}$ is a probability zero event and properties of $t$ on it are inconsequential. Having $s_j<1$ is thus necessary for reducing the size of the identified set. Assumption \ref{prevalence_assumption} and $s_j<1$ imply $P_{s_1,s_0}(r=1-j,y=j)>0$. For the assumption to have identifying power we then only need $P(t=j,r=1-j)$ to be ``large enough''. The definition of ``large enough'' is contingent upon $P(r=1)$ and $(s_1,s_0)$. In \Cref{example_2} below, the threshold is $P(t=j,r=1-j)>2.5\%$ for $j\in\{0,1\}$.

\begin{example}\label{example_2}
    Consider the study as in \Cref{example_1}. If the tests have a tendency to wrongly agree for $y=1$, $\Bar{\mathcal{H}}_{(\theta_1,\theta_0)}(s_1,s_0)$ is a line segment with end points $(0.8,0.8)$ and $(0.95,0.95)$. If they have a tendency to wrongly agree for any $y$, $\dbar{\mathcal{H}}_{(\theta_1,\theta_0)}(s_1,s_0)$ is a line segment with end points $(0.8,0.8)$ and $(0.9,0.9)$.
\end{example}

\begin{remark}
    The identified set $\mathcal{H}_{(\theta_1,\theta_0)}(s_1,s_0)$ was derived by finding all distributions $P_{s_1,s_0}(t,r,y)$ that are consistent with the data given $(s_1,s_0)$. It thus represents a domain of consensus for the values of $(\theta_1,\theta_0)$ under additional assumptions restricting the set of $P_{s_1,s_0}(t,r,y)$ that are considered to be feasible. In other words, any sharp identified set obtained under further assumptions on the statistical dependence of $t$, $r$, and $y$ will be a subset of $\mathcal{H}_{(\theta_1,\theta_0)}(s_1,s_0)$. 
\end{remark}

Given that SARS-CoV-2 RT-PCR and rapid antigen swab tests rely on the same type of sample usually taken from the same location (e.g. nasopharynx, nares or oropharynx), it may be plausible to maintain that the two tests have a tendency to wrongly agree. We will use the assumption in the empirical application in \Cref{application_sect}. More examples can be found in the literature. \citet{hadgu1999discrepant} observes that the same assumption is credible for the ligase chain reaction (LCR) and culture tests for \textit{Chlamydia trachomatis} by the same reasoning. \citet{valenstein1990evaluating} indicates that when determining the performance of direct immunoassay swab tests for \textit{Group A streptococci} using a culture as a reference, the tendency to wrongly agree may hold for $y=1$ due to inadequately obtained samples leading to false negatives. The same is suggested for $y=0$. Patients who are ill with viral pharyngitis, but incidentally carry the bacteria elsewhere, may appear falsely positive on both tests. \citet{vacek1985effect} argues that tine and Mantoux tuberculin tests may have a tendency to wrongly agree for any $y$ as both rely on the antibody reaction to tuberculin.

\subsection{Imperfect Knowledge of Reference Test Characteristics}\label{imperfect_knowledge}

For simplicity of exposition, previously derived identified sets for $(\theta_1,\theta_0)$ were presented under the premise that $(s_1,s_0)$ are known exactly. That assumption might be implausible depending on the setting. Researchers may instead prefer to maintain that they do not possess exact, but rather approximate knowledge of $(s_1,s_0)$. I thus relax Assumption \ref{reference_performance_assumption} by supposing that we only have knowledge of a set $\mathcal{S}$ that contains true sensitivity and specificity of the reference test.
\begin{customassumption}{2A}\label{reference_performance_assumption_relaxed}
Sensitivity and specificity of the reference test are contained in a known compact set $\mathcal{S}\subset[0,1]^2$. All values $(s_1,s_0)\in\mathcal{S}$ satisfy $s_1>1-s_0$.
\end{customassumption}

Assumption \ref{reference_performance_assumption_relaxed} is a weaker form of Assumption \ref{reference_performance_assumption}, since it is implied by it. Similarly, jointly with Assumption \ref{prevalence_assumption}, Assumption \ref{reference_performance_assumption_relaxed} implies that $\forall(s_1,s_0)\in\mathcal{S}:P(r=1)\in(1-s_0,s_1)$. If the condition fails, at least one of the two assumptions is refuted. Compactness of $\mathcal{S}$ is not relevant for identification, but it is utilized in the inference procedure construscted in \Cref{finite_sample_sect}. 

For an element $(s_1,s_0)\in\mathcal{S}$, let the identified set $\mathcal{G}_{(\theta_1,\theta_0)}(s_1,s_0)$ for $(\theta_1,\theta_0)$ be found using \Cref{bounds_theta_prop} or \Cref{bounds_wrongly_agree_prop}, depending on which of the discussed assumptions the researcher is willing to maintain. Denote by $\mathcal{G}_{(\theta_1,\theta_0)}(\mathcal{S})$ the corresponding identified set for $(\theta_1,\theta_0)$ when $(s_1,s_0)$ is known to be in $\mathcal{S}$. All values $(\theta_1,\theta_0)$ that are found in at least one set $G$ within a collection of sets $G\in\{\mathcal{G}_{(\theta_1,\theta_0)}(s_1,s_0): (s_1,s_0)\in\mathcal{S}\}$ then constitute $\mathcal{G}_{(\theta_1,\theta_0)}(\mathcal{S})$. In other words, the set $\mathcal{G}_{(\theta_1,\theta_0)}(\mathcal{S})$ contains all values of $(\theta_1,\theta_0)$ that are consistent with the observed data and at least one $(s_1,s_0)\in\mathcal{S}$. We can formally define:
\begin{equation}\label{general_identified_set}
    \mathcal{G}_{(\theta_1,\theta_0)}(\mathcal{S})=\bigcup_{(s_1,s_0)\in\mathcal{S}}\mathcal{G}_{(\theta_1,\theta_0)}(s_1,s_0).
\end{equation}

\begin{corollary} 
Suppose that Assumption \ref{reference_performance_assumption_relaxed} holds. Let $\mathcal{G}_{(\theta_1,\theta_0)}(s_1,s_0)$ be a sharp identified set for $(\theta_1,\theta_0)$ given a value $(s_1,s_0)$ as defined in \Cref{bounds_theta_prop}, or \Cref{bounds_wrongly_agree_prop}. Then $\mathcal{G}_{(\theta_1,\theta_0)}(\mathcal{S})$ in \eqref{general_identified_set} is a sharp identified set for $(\theta_1,\theta_0)$ if $(s_1,s_0)\in\mathcal{S}$.
\end{corollary}

Any set $\mathcal{G}_{(\theta_1,\theta_0)}(s_1,s_0)$ contains only the values of $(\theta_1,\theta_0)$ that are consistent with the observed data and $(s_1,s_0)$. The union of sets $\mathcal{G}_{(\theta_1,\theta_0)}(s_1,s_0)$ over all possible $(s_1,s_0)\in\mathcal{S}$ then only contains the values of $(\theta_1,\theta_0)$ that are consistent with the observed data and at least one $(s_1,s_0)\in\mathcal{S}$. Hence, the identified set $\mathcal{G}_{(\theta_1,\theta_0)}(\mathcal{S})$ is the smallest possible under the maintained assumptions.

The set $\mathcal{S}$ may take different forms. Expected ones include finite sets, line segments or rectangles. In general, within $\mathcal{G}_{(\theta_1,\theta_0)}(\mathcal{S})$ test performance measures $\theta_1$ and $\theta_0$ will no longer necessarily be linearly dependent. The set $\mathcal{G}_{(\theta_1,\theta_0)}(\mathcal{S})$ may not be a line segment in $[0,1]^2$, but rather a union of line segments of positive and bounded slopes. Hence, it will not be rectangular. It is still possible to demonstrate that $r$ is more precise than $t$. As in \Cref{remark_better_index}, it is feasible for all $(s_1,s_0)\in\mathcal{S}$ and $(\theta_1,\theta_0)\in\mathcal{G}_{(\theta_1,\theta_0)}(\mathcal{S})$ to have $(s_1,s_0)>(\theta_1,\theta_0)$ component-wise.

\section{Bounding Prevalence in Screened Populations}\label{sect_bounding_prevalence}

Sensitivity and specificity in the performance study population are often extrapolated to other populations and used to identify different parameters of interest. Notable examples are prevalence in a population undergoing screening and predictive values. In this section, I show how the specific structure of the identified set for $(\theta_1,\theta_0)$ helps reduce the width of bounds on prevalence when test $t$ is used for screening. Bounds on predictive values are discussed in Appendix \ref{sect_bounding_pv}.

Population disease prevalence is both a research- and policy-relevant parameter. Suppose that we are interested in learning the true prevalence in a certain population that is being screened. Identification of prevalence based on results of an imperfect screening test is a standard epidemiological problem.  Assume that each individual is tested exactly once using only the test $t$. Maintain that $r$ is not used for screening. Depending on the setting, $t$ may be preferred over $r$ for the purpose due to resource constraints, turnaround time or invasiveness, as explained in \Cref{rem:availability_r}. A prominent recent example was the use of antigen testing in university and institutional settings to monitor prevalence during the COVID-19 pandemic. 

To make the distinction between the screened and performance study populations explicit, let $Q(t,y)$ denote the probability distribution which generates the data in the screened population. Unlike in the performance study, $r$ is not available, making the use of $Q(t,r,y)$ superfluous. The data alone identify only $Q(t=1)$. As before, $y$ is not observed, and we are interested in learning $Q(y=1)$. Let $(\tau_1,\tau_0) = \left(Q(t=1|y=1),Q(t=0|y=0)\right)$ be the sensitivity and specificity of $t$ among the screened individuals. We can then write the following identities for the two populations:
\begin{align}\label{eq:contrast}
\begin{split}
    Q(t=1) &= \tau_1 Q(y=1)+(1-\tau_0)Q(y=0)\\
    P(t=1) &= \theta_1 P_{s_1,s_0}(y=1)+(1-\theta_0)P_{s_1,s_0}(y=0).
\end{split}
\end{align}

\begin{remark}
It is important to emphasize that $Q(y=1)$ may differ from $P_{s_1,s_0}(y=1)$ in the performance study population. In the performance study, one can point identify or bound $P_{s_1,s_0}(y=1)$ using $P(r=1)$ and knowledge of $(s_1,s_0)$ or $\mathcal{S}$ as shown in \eqref{prevalence_expression}. In the context of this section, this is not possible since $r$ is not used for screening, making $Q(r=1)$ unidentified.
\end{remark}

Exact knowledge of $(\tau_1,\tau_0)$ identifies the prevalence. (\citet{gart1966comparison2}, \citet{diggle2011estimating}) In the population of interest it directly follows from \eqref{eq:contrast}:
\begin{align}\label{prevalence_eq}
    Q(y=1) = \frac{Q(t=1)+\tau_0-1}{\tau_1+\tau_0-1}.
\end{align}

\citet{walter1988estimation}, and \citet{greenland1996basic} explain that knowledge of $(\tau_1,\tau_0)$ is commonly extrapolated from test performance studies. (see also \citet{gastwirth1987statistical}) That is, researchers maintain the extrapolation assumption: $(\tau_1,\tau_0) =(\theta_1,\theta_0)$, where $(\theta_1,\theta_0)$ are assumed to be identified in a performance study. I follow this practice and generalize \eqref{prevalence_eq} to the case when $(\tau_1,\tau_0) = (\theta_1,\theta_0)$, but $(\theta_1,\theta_0)$ are partially identified. Denote by $\mathcal{G}_{(\theta_1,\theta_0)}(\mathcal{S})$ the identified set for $(\theta_1,\theta_0)$ obtained in a performance study. $\mathcal{S}$ can be a singleton, as when $(s_1, s_0)$ are assumed to be known, in which case I write $\mathcal{G}_{(\theta_1,\theta_0)}(s_1,s_0)$.

\begin{assumption}{(Test Performance Extrapolation)}\label{ass:extrapolation}
    $(\tau_1,\tau_0) =(\theta_1,\theta_0)$.
\end{assumption}

The assumption maintains that test performance is identical in the performance study and the screened populations. It does not require $(\theta_1,\theta_0)$ to be known exactly, and it implies that $(\tau_1,\tau_0)\in \mathcal{G}_{(\theta_1,\theta_0)}(\mathcal{S})$. While \citet{walter1988estimation} state that sensitivity and specificity may readily extrapolate to other populations in many cases, one should be aware that credibility of Assumption \ref{ass:extrapolation} critically depends on the details of the empirical setting. One notable potential threat to its validity is variability of test performance across subpopulations, otherwise known as \textit{spectrum effects}. For example, it is known that test sensitivity may vary across subpopulations with different illness severity.
 
 When spectrum effects exist and the two populations differ in terms of relevant subpopulation proportions, \citet{willis2008spectrum} argues that $(\theta_1,\theta_0) = (\tau_1,\tau_0)$ may be implausible. To see this, observe that in this case $(\theta_1,\theta_0)$ and $(\tau_1,\tau_0)$ are weighted averages of sensitivity and specificity across relevant subpopulations, but with different weights. However, if test performance is known for all relevant subpopulations, one could maintain Assumption \ref{ass:extrapolation} and identify prevalence at the subpopulation level using arguments that follow. (see also \citet{mulherin2002spectrum})

\begin{theoremEnd}[restate, no link to proof, proof at the end]{proposition}\label{prevalence_screening_prop}%

Suppose that Assumption \ref{ass:extrapolation} holds and the population is screened only using $t$. Let $\mathcal{G}_{(\theta_1,\theta_0)}(s_1,s_0)$ be a known sharp identified set from \Cref{bounds_theta_prop} or \Cref{bounds_wrongly_agree_prop}. Denote by $\theta_j^L$ and $\theta_j^U$ the smallest and largest values of $\theta_j$ in $\mathcal{G}_{(\theta_1,\theta_0)}(s_1,s_0)$. The sharp bounds on prevalence $Q(y=1)$ are:
\begin{align}\label{prevalence_bounds}
\begin{split}
    Q(y=1)\in \Pi_{s_1,s_0}:=\Bigg[&\min\left\{\frac{Q(t=1)+\theta_0^L-1}{\theta_1^L+\theta_0^L-1},\frac{Q(t=1)+\theta_0^U-1}{\theta_1^U+\theta_0^U-1}\right\},\\
    &\max\left\{\frac{Q(t=1)+\theta_0^L-1}{\theta_1^L+\theta_0^L-1},\frac{Q(t=1)+\theta_0^U-1}{\theta_1^U+\theta_0^U-1}\right\}\Bigg]\cap [0,1]
\end{split}
\end{align}
when $\forall(\theta_1,\theta_0)\in\mathcal{G}_{(\theta_1,\theta_0)}(s_1,s_0):\enskip \theta_1\neq 1- \theta_0$, and $Q(y=1)\in[0,1]$ otherwise.
\end{theoremEnd}

\begin{proofEnd}

The proof proceeds in two steps. First we show that the bounds are valid. For sharpness, we consider an arbitrary point in the bounds. We then construct a distribution consistent with the assumptions that generates the point, and that is observationally equivalent to the data. The data alone identify only $Q(t=1)$.

We study two distinct cases.
 
\textbf{Case 1: $\forall(\theta_1,\theta_0)\in\mathcal{G}_{(\theta_1,\theta_0)}(s_1,s_0):\enskip \theta_1\neq 1- \theta_0$}

The bounds on $Q(y=1)$ can be imposed as:
\begin{align*}
    Q(y=1) &= \frac{Q(t=1)+\tau_0-1}{\tau_1+\tau_0-1}\\
    &\in \Bigg[\min_{(\theta_1,\theta_0)\in\mathcal{G}_{(\theta_1,\theta_0)}(s_1,s_0)}\frac{Q(t=1)+\theta_0-1}{\theta_1+\theta_0-1}, \max_{(\theta_1,\theta_0)\in\mathcal{G}_{(\theta_1,\theta_0)}(s_1,s_0)}\frac{Q(t=1)+\theta_0-1}{\theta_1+\theta_0-1}\Bigg]\cap[0,1]
\end{align*}

where the first line follows by \eqref{prevalence_eq}, and the second by Assumption \ref{ass:extrapolation}. The intersection with $[0,1]$ is added by definition $Q(y=1)$ and the fact that $\frac{Q(t=1)+\theta_0-1}{\theta_1+\theta_0-1}\not\in[0,1]$ if and only if $(\theta_1,\theta_0)$ are such that $Q(t=1)\not\in\left[\min(\theta_1,1-\theta_0),\max(\theta_1,1-\theta_0)\right]$, which is possible. The expression $\frac{Q(t=1)+\theta_0-1}{\theta_1+\theta_0-1}$ is increasing in $\theta_0$ and decreasing in $\theta_1$. We will show that it attains extrema at the end-points of the line segment $\mathcal{G}_{(\theta_1,\theta_0)}(s_1,s_0)$. By Propositions \ref{bounds_theta_prop} and \ref{bounds_wrongly_agree_prop}, for any $(\theta_1,\theta_0)\in\mathcal{G}_{(\theta_1,\theta_0)}(s_1,s_0)$:
\begin{equation}\label{eq:linearity}
    \theta_0 = 1 - a + \theta_1 b
\end{equation}
for known constants $a$ and $b\in(0,\infty)$. Given that $\mathcal{G}_{(\theta_1,\theta_0)}(s_1,s_0)\subset [0,1]^2$ is a line segment, it is a connected set. By assumption it does not contain $(\theta_1,\theta_0)$ such that $\theta_1+\theta_0 = 1$, therefore it does not intersect the negatively-sloped diagonal of the unit rectangle. Thus, all $(\theta_1,\theta_0)\in\mathcal{G}_{(\theta_1,\theta_0)}(s_1,s_0)$ are such that either $\theta_1+\theta_0>1$ or $\theta_1+\theta_0<1$. By \eqref{eq:linearity}, for all $\theta_1$ in the identified set we have either $\theta_1 (b+1)>a$ or $\theta_1 (b+1)<a$, so $\theta_1 (b+1)\neq a$. For any $(\theta_1,\theta_0)\in\mathcal{G}_{(\theta_1,\theta_0)}(s_1,s_0)$ we can then write:
\begin{align}\label{prevalence_eq_proof}
    \frac{Q(t=1)+\theta_0-1}{\theta_1+\theta_0-1}= \frac{Q(t=1)+\theta_1 b-a}{\theta_1(b+1)-a}.
\end{align}

First derivative of \eqref{prevalence_eq_proof} with respect to $\theta_1$ is $\frac{a-(b+1) Q(t=1)}{(a-(b+1) \theta_1)^2}$ which has the same sign for all $\theta_1$ in the identified set. If (and only if) $a=(b+1) Q(t=1)$, the expression is a constant function of $\theta_1$.\footnote{Note that this is equivalent to prevalence being point identified. Expressions for $a$ and $b$ in \eqref{bounds_joint} reveal that this happens if only if $P(t=1)$ in the performance population equals $Q(t=1)$ in the screened population.} When $a>(b+1) Q(t=1)$, it is minimized at $\theta_1^L$ and maximized at  $\theta_1^U$. Conversely, if $a<(b+1) Q(t=1)$, the expression is minimized at $\theta_1^U$ and maximized at $\theta_1^L$. By \eqref{eq:linearity} and $b>0$,  $\theta_1^L$ and $\theta_1^U$ correspond to $\theta_0^L$ and $\theta_0^U$ in $\mathcal{G}_{(\theta_1,\theta_0)}(s_1,s_0)$, respectively, showing that $\Pi_{s_1,s_0}$ are valid bounds for $Q(y=1)$. 

To show sharpness, pick an arbitrary point $\pi\in\Pi_{s_1,s_0}$. We demonstrate that one can construct a distribution $Q(t,y)$ such that: $i)$ it consistent with the observed data $Q(t=1)$; $ii)$ $Q(y=1)=\pi$; $iii)$ it is consistent with the assumptions $\left(Q(t=1|y=1),Q(t=0|y=0)\right) = (\theta_1,\theta_0)\in\mathcal{G}_{(\theta_1,\theta_0)}(s_1,s_0)$. All marginals of the distribution $Q(t,y)$ are completely determined by observational data and $\pi$. To complete the proof, we only need to appropriately specify the dependence structure $(\theta_1,\theta_0)$ such that it is feasible, i.e. in the identified set for the parameters. 

If Assumption \ref{ass:extrapolation} holds, $\left(Q(t=1|y=1),Q(t=0|y=0)\right)\in\mathcal{G}_{(\theta_1,\theta_0)}(s_1,s_0)$, so $\Pi_{s_1,s_0} \neq \emptyset$. Denote then by $\Pi_{s_1,s_0} = [\pi^L,\pi^U]\cap[0,1]$. Consider the case where $\pi^L = \frac{Q(t=1)+\theta_0^L-1}{\theta_1^L+\theta_0^L-1}\leq\frac{Q(t=1)+\theta_0^U-1}{\theta_1^U+\theta_0^U-1} = \pi^U$. The converse case  $\frac{Q(t=1)+\theta_0^L-1}{\theta_1^L+\theta_0^L-1}>\frac{Q(t=1)+\theta_0^U-1}{\theta_1^U+\theta_0^U-1}$ follows a symmetric argument. Let $\theta_1^\beta = \beta\theta_1^L+(1-\beta)\theta_1^U$ for any $\beta\in[0,1]$. Define:
\begin{align}
\begin{split}
    \pi^\beta &= \frac{Q(t=1)+\theta_1^\beta b-a}{\theta_1^\beta(b+1)-a} =\frac{\beta\left(Q(t=1)+\theta_1^Lb-a\right)+(1-\beta)\left(Q(t=1)+\theta_1^U b-a\right)}{\beta\left(\theta_1^L(b+1)-a\right)+(1-\beta)\left(\theta_1^U(b+1)-a\right)}  \\
    &=\pi^L\frac{\beta\left(\theta_1^L(b+1)-a\right)}{\beta\left(\theta_1^L(b+1)-a\right)+(1-\beta)\left(\theta_1^U(b+1)-a\right)} +\pi^U\frac{(1-\beta)\left(\theta_1^U(b+1)-a\right)}{\beta\left(\theta_1^L(b+1)-a\right)+(1-\beta)\left(\theta_1^U(b+1)-a\right)}
\end{split}
\end{align}
where the second line follows by \eqref{prevalence_eq_proof} and definition of $\pi^L$ and $\pi^U$. Since $\pi\in[\pi^L,\pi^U]$, then $\exists\alpha\in[0,1]$ $\pi = \alpha\pi^L+(1-\alpha)\pi^U$. We can define:
\begin{align}\label{eq:beta}
    \beta = \frac{ \alpha \left(\theta_1^U(b+1)-a\right)}{\alpha\left(\theta_1^U(b+1)-a\right)+(1-\alpha)\left(\theta_1^L(b+1)-a\right)}\in[0,1].
\end{align}

For $\beta$ in \eqref{eq:beta}, we have $\pi^\beta=\pi$. Then let $(\theta_1^\beta,\theta_0^\beta) = \beta(\theta_1^L,\theta_0^L)+(1-\beta)(\theta_1^U, \theta_0^U)$. Since $(\theta_1^\beta,\theta_0^\beta)$ is a linear combination of endpoints of a line segment, it must also be an element of the line segment. Hence $(\theta_1^\beta,\theta_0^\beta)\in\mathcal{G}_{(\theta_1,\theta_0)}(s_1,s_0)$, proving that for any $\pi\in \Pi_{s_1,s_0}$ we can construct $Q(t,y)$ such that it is consistent with the observed data and assumptions, with $Q(y=1)=\pi$. Hence $\Pi_{s_1,s_0}$ is sharp.

\textbf{Case 2: $\exists(\theta_1,\theta_0)\in\mathcal{G}_{(\theta_1,\theta_0)}(s_1,s_0):\theta_1+\theta_0=1$}

Fix $(\theta_1,\theta_0)\in\mathcal{G}_{(\theta_1,\theta_0)}:\theta_1+\theta_0=1$. Then $Q(t=1|y=1) = Q(t=1|y=0)$, so $t\independent y$ is consistent with Assumption \ref{ass:extrapolation}. Hence $Q(y=1)\in[0,1]$ is consistent with any observed $Q(t=1)$. Sharpness is also immediate since for an arbitrary point $\pi\in[0,1]$, we can fix $(\theta_1,\theta_0)\in\mathcal{G}_{(\theta_1,\theta_0)}:\theta_1+\theta_0=1$ and define $Q(t,y)$ with $Q(t,y=1) = Q(t)\pi$ and $Q(t,y=0)=Q(t)(1-\pi)$ for any $Q(t)$. Thus, there exists a distribution $Q(t,y)$ which is consistent with $Q(t=1)$, $(\theta_1,\theta_0)\in\mathcal{G}_{(\theta_1,\theta_0)}$ and $Q(y=1)=\pi$. 

\end{proofEnd}

\begin{remark}
The proof of \Cref{prevalence_screening_prop} remains valid if we replace Assumption \ref{ass:extrapolation} with a weaker condition $(\tau_1,\tau_0)\in\mathcal{G}_{(\theta_1,\theta_0)}(s_1,s_0)$. The benefits of doing so are primarily technical, as it difficult to think of settings in which the weaker condition is plausible and Assumption \ref{ass:extrapolation} is not.
\end{remark}

\Cref{prevalence_screening_prop} extends the identity \eqref{prevalence_eq} to the case when $(\theta_1,\theta_0)$ extrapolate to $Q(t,y)$. Note that it maintains that $(s_1,s_0)$ are known exactly in the performance study. \Cref{cor:prevalence} generalizes the results to the case when $(s_1,s_0)$ are known to lie in $\mathcal{S}$. The resulting bounds on the screened population prevalence $Q(y=1)$ are sharp in the absence of additional data, namely results of other tests such as $r$.

Whenever there exist $(\theta_1,\theta_0)\in\mathcal{G}_{(\theta_1,\theta_0)}(s_1,s_0):\theta_1+\theta_0=1$, prevalence is unidentified as $t$ may not be informative of $y$. Such tests are not useful for screening purposes. When $t$ is informative of $y$, that is $\forall(\theta_1,\theta_0)\in\mathcal{G}_{(\theta_1,\theta_0)}(s_1,s_0):\enskip \theta_1\neq 1- \theta_0$, the importance of the linear structure of $\mathcal{G}_{(\theta_1,\theta_0)}(s_1,s_0)$ for bounding $Q(y=1)$ becomes apparent. Identifying power of the structure can be substantial as highlighted by \Cref{rem:point_identified_prevalence}. For certain $Q(t,y)$, resulting bounds may even point identify $Q(y=1)$ using data only on $t$, despite $(\theta_1,\theta_0)$ being partially identified. 

It is important to highlight that Assumption \ref{ass:extrapolation} is refutable. If $t$ is informative of $y$, then it is possible that $\Pi_{s_1,s_0} = \emptyset$. By \eqref{prevalence_eq}, that happens if all assumed values $(\tau_1,\tau_0)$ consistent with the assumption $(\tau_1,\tau_0)=(\theta_1,\theta_0)\in\mathcal{G}_{(\theta_1,\theta_0)}(s_1,s_0)$ result in $Q(y=1)\not\in[0,1]$. This would contradict the definition of $Q(y=1)$, implying that $(\tau_1,\tau_0)\neq (\theta_1,\theta_0)$.

\begin{remark}\label{rem:point_identified_prevalence}
    Proof of \Cref{prevalence_screening_prop} reveals that when $(s_1,s_0)$ are known, $Q(y=1)$ is point identified if $Q(t=1)=P(t=1)$. One can intuitively see this from the fact that $Q(t=1)=P(t=1)$, $(\tau_1,\tau_0)=(\theta_1,\theta_0)$ and \eqref{eq:contrast} jointly imply that $P_{s_1,s_0}(y=1) = Q(y=1)$. Since $P_{s_1,s_0}(y=1)$ is point-identified when $(s_1,s_0)$ are known, then $Q(y=1)$ is too.
\end{remark}

Let $\mathcal{G}_{\theta_j}(s_1,s_0)=\{\theta_j:(\theta_1,\theta_0)\in\mathcal{G}_{(\theta_1,\theta_0)}(s_1,s_0)\}$ denote the individual bounds on $\theta_j$ for $j=0,1$. The sets $\mathcal{G}_{\theta_1}(s_1,s_0)$ and $\mathcal{G}_{\theta_0}(s_1,s_0)$ are also referred to as projection bounds on $\theta_1$ and $\theta_0$.

\begin{remark}\label{rem:disregarding_structure}
Let $\forall(\theta_1,\theta_0)\in\mathcal{G}_{(\theta_1,\theta_0)}(s_1,s_0):\enskip \theta_1\neq 1- \theta_0$ so that $t$ is informative for screening.
If we were to disregard the linear structure of the sharp identified set by supposing that it is a rectangle $\mathcal{G}_{\theta_1}(s_1,s_0)\times \mathcal{G}_{\theta_0}(s_1,s_0)$, then the bounds on the prevalence would be:
\begin{align}\label{prevalence_bounds_no_dependence}
\begin{split}
    Q(y=1)\in\Bar{\Pi}_{s_1,s_0}:= \Bigg[\frac{Q(t=1)+\theta_0^L-1}{\theta_1^U+\theta_0^L-1},  \frac{Q(t=1)+\theta_0^U-1}{\theta_1^L+\theta_0^U-1}\Bigg]\cap[0,1].
\end{split}
\end{align}
It is direct that $\Pi_{s_1,s_0}\subset \Bar{\Pi}_{s_1,s_0}$ whenever $\mathcal{G}_{(\theta_1,\theta_0)}(s_1,s_0)$ is not a singleton, so that $\theta_j^U>\theta_j^L$ for $j\in\{0,1\}$. Disregarding the linear structure of the identified set for $(\theta_1,\theta_0)$ yields strictly wider bounds on prevalence. For any $Q(t=1)$, $\Bar{\Pi}_{s_1,s_0}$ is an infinite set. If $P(t=1)=Q(t=1)$, $\Pi_{s_1,s_0}$ is a singleton.
\end{remark}

Bounds in \eqref{prevalence_bounds_no_dependence} would follow from methods that do not establish the linear structure of  $\mathcal{G}_{(\theta_1,\theta_0)}(s_1,s_0)$, such as \citet{thibodeau1981evaluating} and \citet{emerson2018biomarker}. \Cref{rem:disregarding_structure} shows that using such methods to bound test performance will yield wider bounds on prevalence in a population screened by an informative $t$. Moreover, depending on $P(t,r)$ and $Q(t)$ the difference in width can be extreme, since $\Pi_{s_1,s_0}$ can be a singleton, while $\Bar{\Pi}_{s_1,s_0}$ is always an infinite set. \Cref{application_sect} illustrates how relying on rectangular identified sets affects the width of prevalence bounds using empirical examples. It compares prevalence bound widths implied by estimated identified sets for $(\theta_1,\theta_0)$ constructed using \citet{thibodeau1981evaluating}, \citet{emerson2018biomarker}, and the method described here, for hypothetical screened populations with different $Q(t=1)$.

\begin{corollary}\label{cor:prevalence}
Suppose that Assumption \ref{ass:extrapolation} holds and the population is screened only using $t$. Let $\mathcal{G}_{(\theta_1,\theta_0)}(\mathcal{S}) =  \bigcup_{(s_1,s_0)\in\mathcal{S}}\mathcal{G}_{(\theta_1,\theta_0)}(s_1,s_0)$, where $\mathcal{G}_{(\theta_1,\theta_0)}(s_1,s_0)$ are known sharp identified sets from \Cref{bounds_theta_prop} or \Cref{bounds_wrongly_agree_prop}. The sharp bounds on prevalence $Q(y=1)$ are:
\begin{align}\label{prevalence_bounds_general}
\begin{split}
    Q(y=1)\in \Pi_\mathcal{S}:= \bigcup_{(s_1,s_0)\in\mathcal{S}}\Pi_{s_1,s_0} 
\end{split}
\end{align}
when $\forall(\theta_1,\theta_0)\in\mathcal{G}_{(\theta_1,\theta_0)}(\mathcal{S}):\enskip \theta_1\neq 1- \theta_0$, and $Q(y=1)\in[0,1]$ otherwise.
\end{corollary}

\Cref{cor:prevalence} generalizes \Cref{prevalence_screening_prop} to the case when $(s_1,s_0)$ are not known exactly. If the shape of $\mathcal{G}_{(\theta_1,\theta_0)}(\mathcal{S})$ was disregarded by assuming that the identified set was a rectangle, bounds $\Bar{\Pi}_\mathcal{S}$ analogous to the ones in \eqref{prevalence_bounds_no_dependence} can still be formed, and it would hold that $\Pi_\mathcal{S}\subset \Bar{\Pi}_\mathcal{S}$.

Throughout this section we have assumed that screening is performed once in the population. If screening is done repeatedly, a time series of prevalence bounds can be constructed. When there is selection into testing, bounds on prevalence by \citet{stoye2022bounding} may be used, for which the bounds on $(\theta_1,\theta_0)$ derived in \Cref{identification_sect} are natural inputs.

\section{Estimation and Inference}\label{finite_sample_sect}

Identified sets in \Cref{identification_sect} can be found when $P(t,r)$ is known. In practice, researchers must use sample data to estimate the identified set and conduct inference. This section demonstrates consistent estimation of the identified set and construction of confidence sets for the points in the identified set that are uniformly consistent in level over a large family of permissible distributions.

Let $W_i = (t_i,r_i)\in\{0,1\}^2$ for $i=1,\hdots,n$ constitute the observed data of $n$ i.i.d observations from the distribution $P(t,r)\in\mathbf{P}$, where $\mathbf{P}$ is a family of categorical distributions with 4 categories. Let $\mathcal{G}_{(\theta_1,\theta_0)}(s_1,s_0)$ denote an arbitrary identified set for $(\theta_1,\theta_0)$ given $(s_1,s_0)$ from any of the propositions above, and $\mathcal{G}_{\theta_j}(s_1,s_0)$ the corresponding identified set for $\theta_j$ with $j=0,1$. Replacing population parameters with their consistent estimators in closed form expressions for $\mathcal{G}_{\theta_j}(s_1,s_0)$ and $\mathcal{G}_{(\theta_1,\theta_0)}(s_1,s_0)$ yields the consistent plug-in estimator of the identified sets (\citet{manski1998monotone}, \citet{tamer2010partial}).

Let $\mathbbm{1}\{\cdot\}$ denote the indicator function. Suppose first that $(s_1,s_0)$ are known. $\hat{P}(t=j,r=k)=\frac{\sum_{i=1}^n\mathbbm{1}\{t_i=j,r_i=k\}}{n}$ are consistent estimators of $P(t=j,r=k)$ for all $(j,k)\in\{0,1\}^2$. Combining $\hat{P}(t=j,r=k)$ with the knowledge of $(s_1,s_0)$ yields $\hat{P}_{s_1,s_0}(r=k,y=l)$ for every ${k,l}\in\{0,1\}^2$. The plug-in estimator $\hat{\mathcal{G}}_{\theta_j}(s_1,s_0)$ for the identified set of a single parameter $\theta_j$ follows immediately by inputting $\hat{P}(t=j,r=k)$ and $\hat{P}_{s_1,s_0}(r=k,y=l)$ into the bounds in \Cref{bounds_theta_prop}, or \Cref{bounds_wrongly_agree_prop}. Consistent estimator $\hat{\mathcal{G}}_{(\theta_1,\theta_0)}(s_1,s_0)$ of the joint identified set for $(\theta_1,\theta_0)$ follows from \eqref{bounds_joint}. 

In the case when $(s_1,s_0)$ are only known to be bounded by some compact set $\mathcal{S}$, one can obtain the consistent estimator $\hat{\mathcal{G}}_{(\theta_1,\theta_0)}(\mathcal{S})=\bigcup_{(s_1,s_0)\in\mathcal{S}}\hat{\mathcal{G}}_{(\theta_1,\theta_0)}(s_1,s_0)$. This is done by finding a union of $\hat{\mathcal{G}}_{(\theta_1,\theta_0)}(s_1,s_0)$ over a fine grid of $(s_1,s_0)$ covering $\mathcal{S}$. The procedure requires two nested grid-search algorithms, and the level of coarseness of the two grids can impact computation time.

FDA Statistical Guidance on Reporting Results Evaluating Diagnostic Tests requires all diagnostic performance studies to report confidence intervals for $\theta_1$ and $\theta_0$. I show how one can use the method for inference based on moment inequalities from \citet{romano2014practical} to form confidence sets that cover the true parameters with at least some pre-specified probability $1-\alpha$ and that are uniformly consistent over a large family of permissible distributions $\mathbf{P}$. 

Let $C_n$ be the confidence set of interest and let $\mathbf{\Theta}(P)= \bigcup_{(s_1,s_0)\in\mathcal{S}}\Big(\mathcal{G}_{(\theta_1,\theta_0)}(s_1,s_0)\times\{(s_1,s_0)\}\Big)$ be an identification region for $\theta = (\theta_1,\theta_0,s_1,s_0)$ that depends on $P\in\mathbf{P}$ through $\mathcal{G}_{(\theta_1,\theta_0)}(s_1,s_0)$, and where $\mathcal{S}$ can be a singleton.\footnote{More precisely, we are interested in $C_n$ for the points in $\mathcal{G}_{(\theta_1,\theta_0)}(\mathcal{S})=\bigcup_{(s_1,s_0)\in\mathcal{S}}\mathcal{G}_{(\theta_1,\theta_0)}(s_1,s_0)=\{(\theta_1,\theta_0):\theta\in\mathbf{\Theta}(P)\}$. When $P(t,r)$ is known, whether one defines the identified set as $\mathcal{G}_{(\theta_1,\theta_0)}(\mathcal{S})$ or $\mathbf{\Theta}(P)$ is inconsequential.} Note that $\theta$ includes reference test performance $(s_1,s_0)$. This is done to facilitate convenient definition of moment inequalities that represent the identified set of interest, regardless of whether $(s_1,s_0)$ are known exactly or not. The confidence set $C_n$ should satisfy:
\begin{align}\label{conf_set_definition}
    \liminf_{n\conv\infty}\inf_{P\in\mathbf{P}}\inf_{\theta\in\mathbf{\Theta}(P)}P(\theta\in C_n)\geq (1-\alpha).
\end{align}

\citet{canay2017practical} provide an overview of the recent advances in inference based on moment inequalities that are focused on finding $C_n$ in partially identified models. They underline the importance of uniform consistency of $C_n$ in level in these settings. If it fails, it may be possible to construct a distribution of the data $P(t,r)$ such that for any sample size finite-sample coverage probability of some points in the identified set is arbitrarily low. In that sense, inference based on confidence intervals that are consistent only pointwise may be severely misleading in finite samples. To exploit existing inference methods based on moment inequalities to construct $C_n$, the identified set $\mathbf{\Theta}(P)$ must be equivalent to some set $\tilde{\mathbf{\Theta}}(P)$:
\begin{align}\label{def_moment_ineq}
    \tilde{\mathbf{\Theta}}(P) = \{\theta\in[0,1]^2\times\mathcal{S}:E_P\big(m_j(W_i,\theta)\big)\leq0\text{ for $j\in J_1$ }, E_P\big(m_j(W_i,\theta)\big)=0\text{ for $j\in J_2$}\}
\end{align}
where $m_j(W_i,\theta)$ for $j\in J_1\cup J_2$ are the components of a random function $m:\{0,1\}^2\times[0,1]^2\times\mathcal{S}\conv\mathbb{R}^k$ such that $|J_1|+|J_2|=k$. Construction of the confidence set for points in the identified set $\tilde{\mathbf{\Theta}}(P)$ is done by imposing a fine grid over the parameter space $[0,1]^2\times\mathcal{S}$ for $\theta$ and performing test inversion.

Identified sets derived in the previous section are representable by \eqref{def_moment_ineq}. Focus on the bounds for $\theta_1 $ in \Cref{bounds_wrongly_agree_prop} when the tests have the tendency to wrongly agree for $y=1$ for intuition. Observe that there are four values that are all lower bounds on $\theta_1$ given $(s_1,s_0)$. Similarly there are four values that are all upper bounds. One of the lower bounds is trivial: $\theta_1\geq 0$. One upper bound is $\theta_1\leq \left( \frac{P_{s_1,s_0}(r=0,y=1)}{2}+P_{s_1,s_0}(r=1,y=1)\right)\frac{1}{P_{s_1,s_0}(y=1)} = \frac{1+s_1}{2}$. There are no parameters pertaining to the population distribution in the bound. This is a restriction on the parameter space, under which $\theta\in\bigcup_{(s_1,s_0)\in\mathcal{S}}[0,\frac{1+s_1}{2}]\times[0,1]\times\{(s_1,s_0)\}$.  With the appropriate parameter space, there are six relevant values for the bounds on $\theta_1$ that depend on parameters of the population distribution, three for the upper and three for the lower bound. Hence, we can represent the bounds on $\theta_1$ using six moment inequalities. 

Proposition \ref{bounds_wrongly_agree_prop} implies that we only need to include one additional moment equality to represent the joint identification region $\Bar{\mathcal{H}}_{(\theta_1,\theta_0)}(s_1,s_0)$ for $(\theta_1,\theta_0)$. Then the moment function $\Bar{m}^1(W_i,\theta)$ representing the identified set $\mathbf{\Theta}(P)= \bigcup_{(s_1,s_0)\in\mathcal{S}}\Big(\Bar{\mathcal{H}}_{(\theta_1,\theta_0)}(s_1,s_0)\times\{(s_1,s_0)\}\Big)$  will have $k=7$, where $J_1=\{1,\hdots,6\}$ and $J_2=\{7\}$.

\begin{theoremEnd}[restate, proof at the end, no link to proof]{proposition}\label{prop_inequalities_wrongly_agree}
Assume that the index and reference tests have a tendency to wrongly agree only for $y=1$. Let the moment function $\Bar{m}^1$ be:
\begin{align}
\begin{split}\label{moment_inequalities_wrongly_agree_1}
    \Bar{m}^1(W_i,\theta) =\begin{pmatrix} \Bar{m}^1_1(W_i,\theta)\\
    \Bar{m}^1_2(W_i,\theta)\\
    \Bar{m}^1_3(W_i,\theta)\\
    \Bar{m}^1_4(W_i,\theta)\\
    \Bar{m}^1_5(W_i,\theta)\\
    \Bar{m}^1_6(W_i,\theta)\\
    \Bar{m}^1_7(W_i,\theta)
    \end{pmatrix} = 
    \begin{pmatrix} (-\theta_1+s_1)\frac{r_i-1+s_0}{s_1-1+s_0}+(t_i-1)r_i\\
     (-\theta_1+1-s_1)\frac{r_i-1+s_0}{s_1-1+s_0}+(r_i-1)(1-t_i)\\
    (-\theta_1+1)\frac{r_i-1+s_0}{s_1-1+s_0}+(t_i-1)\\
    \theta_1\frac{r_i-1+s_0}{s_1-1+s_0}-t_i\\
   (\theta_1-s_1)\frac{r_i-1+s_0}{s_1-1+s_0}-t_i(1-r_i)\\
     \Big(\theta_1+\frac{-1+s_1}{2}\Big)\frac{r_i-1+s_0}{s_1-1+s_0}-t_ir_i\\
    (\theta_0-1)(1-\frac{r_i-1+s_0}{s_1-1+s_0})-\theta_1\frac{r_i-1+s_0}{s_1-1+s_0}+t_i
    \end{pmatrix}.
\end{split}
\end{align}
Moment inequalities and equalities defined by $\Bar{m}^1$ for $J_1=\{1,\hdots,6\}$ and $J_2=\{7\}$ represent the joint identification region $\mathbf{\Theta}(P)= \bigcup_{(s_1,s_0)\in\mathcal{S}}\Big(\Bar{\mathcal{H}}_{(\theta_1,\theta_0)}(s_1,s_0)\times\{(s_1,s_0)\}\Big)$ for $\Bar{\mathcal{H}}_{(\theta_1,\theta_0)}(s_1,s_0)$ defined in Proposition \ref{bounds_wrongly_agree_prop} for $y=1$. For each $\theta\in\bigcup_{(s_1,s_0)\in\mathcal{S}}[0,\frac{1+s_1}{2}]\times[0,1]\times\{(s_1,s_0)\}$ such that $E_P\big(\Bar{m}^1_j(W_i,\theta)\big)\leq0\textrm{ for $j=1,\hdots,6$ and } E_P\big(\Bar{m}^1_7(W_i,\theta)\big)=0$, it must be that $\theta\in\mathbf{\Theta}(P)$. Conversely, if $\theta\in\mathbf{\Theta}(P)$, then $E_P\big(\Bar{m}^1_j(W_i,\theta)\big)\leq0\textrm{ for $j=1,\hdots,6$ and } E_P\big(\Bar{m}^1_7(W_i,\theta)\big)=0$. 
\end{theoremEnd}
\begin{proofEnd}
The proof is analogous to the proof of \Cref{prop_original_inequalities}. From the definition of $\Bar{\mathcal{H}}_{\theta_1}(s_1,s_0)$ for $y=1$ in \Cref{bounds_wrongly_agree_prop}:
\begin{align}
\begin{split}\label{prop3_bounds_proof}
    \theta_1P_{s_1,s_0}(y=1) &\geq max\Big(0,P(t=1,r=1)-P_{s_1,s_0}(r=1,y=0)\Big)\\
    &+max\Big(0,P_{s_1,s_0}(r=0,y=1)-P(t=0,r=0)\Big)\\
    \theta_1P_{s_1,s_0}(y=1) &\leq min\Big(P(t = 1, r=0), \frac{P_{s_1,s_0}(r=0,y=1)}{2}\Big)\\
    &+min\Big(P(t = 1, r=1), P_{s_1,s_0}(r=1,y=1)\Big).
\end{split}
\end{align}  Suppose that $E_P\big(\Bar{m}^1_j(W_i,\theta)\big)\leq 0$ for $j=1,2\hdots,6$ and $E_P\big(\Bar{m}^1_7(W_i,\theta)\big)= 0$. From \eqref{moment_inequalities_wrongly_agree_1}:
\begin{align}\label{expectations_moment_inequalities_wrongly_agree_upper}
    \begin{split}
    E_P\big(\Bar{m}^1_6(W_i,\theta)\big) &= \Big(\theta_1+\frac{-1+s_1}{2}\Big)P_{s_1,s_0}(y=1)-P(t=1,r=1)\\
     &=\theta_1 P_{s_1,s_0}(y=1)-\frac{P_{s_1,s_0}(r=0,y=1)}{2}-P(t=1,r=1)\leq 0
    \end{split}
\end{align}
Using $E_P\big(\Bar{m}^1_j(W_i,\theta)\big) = E_P\big(m_j(W_i,\theta)\big)$ for $j=1,\hdots,5$, $E_P\big(\Bar{m}^1_7(W_i,\theta)\big) = E_P\big(m_7(W_i,\theta)\big)$, \eqref{expectations_moment_inequalities_original_lower}, \eqref{expectations_moment_inequalities_original_upper}, \eqref{expectations_moment_equalities}, and \eqref{expectations_moment_inequalities_wrongly_agree_upper}, yields that $E_P\big(\Bar{m}^1_j(W_i,\theta)\big)\leq0\textrm{ for $j=1,\hdots,7$ and } E_P\big(\Bar{m}^1_7(W_i,\theta)\big)=0$ represent the joint identification region $\mathbf{\Theta}(P)= \bigcup_{(s_1,s_0)\in\mathcal{S}}\Big(\Bar{\mathcal{H}}_{(\theta_1,\theta_0)}(s_1,s_0)\times\{(s_1,s_0)\}\Big)$ by the same argument as in the proof of \Cref{prop_original_inequalities}.
\end{proofEnd}

Similarly, it is possible to define moment inequality functions that represent remaining identified sets in Propositions \ref{bounds_theta_prop} and \ref{bounds_wrongly_agree_prop}. They are found in equations \eqref{moment_inequalities_original}, \eqref{moment_inequalities_wrongly_agree_0}, and \eqref{moment_inequalities_wrongly_agree_both} in Appendix \ref{sect_additional_mom_inequalities}. 

\citet{romano2014practical}, Theorem 3.1 provides sufficient conditions for uniform consistency of confidence sets over a large family of distributions. Assumption \ref{family_restriction_assumption} defines a family $\mathbf{P}$ to which the conclusions of Theorem 3.1 apply. This is demonstrated by \Cref{family_theorem} below.
\begin{assumption}\label{family_restriction_assumption}
There exists a number $\varepsilon>0$ such that $P(t=j,r=k)\geq\varepsilon$ for all $(j,k)\in\{0,1\}^2$ and any $P(t,r)\in\mathbf{P}$.
\end{assumption}
The assumption restricts $\mathbf{P}$ to distributions $P(t,r)$ such that all outcomes $(t,r)\in\{0,1\}^2$ have probability that is bounded away from zero. It serves a technical purpose, ensuring that the uniform integrability condition required by \citet{romano2014practical}, Theorem 3.1 holds. The assumption is easily interpretable and it appears reasonable in the analyzed data, as discussed in \Cref{application_sect}.

\begin{theoremEnd}[restate, proof at the end, no link to proof]{theorem}\label{family_theorem}
Suppose that Assumptions \ref{random_sampling_assumption}, \ref{reference_performance_assumption_relaxed}, \ref{prevalence_assumption}, and \ref{family_restriction_assumption} hold. Then for any component $m_j(W_i,\theta)$ in \eqref{moment_inequalities_wrongly_agree_1}, \eqref{moment_inequalities_original}, \eqref{moment_inequalities_wrongly_agree_0}, and \eqref{moment_inequalities_wrongly_agree_both}:
\begin{enumerate}
    \item $Var_P(m_j(W_i,\theta))>0$ and for all $P\in\mathbf{P}$ and $\theta\in[0,1]^2\times\mathcal{S}$;
    \item $\limsup_{\lambda\conv\infty}\sup_{P\in\mathbf{P}}\sup_{\theta\in\mathbf{\Theta}(P)}E_P\left[\left(\frac{m_j(W_i,\theta)-\mu_j(\theta,P)}{\sigma_j(\theta,P)}\right)^2\mathbbm{1}\left\{|\frac{m_j(W_i,\theta)-\mu_j(\theta,P)}{\sigma_j(\theta,P)}|>\lambda\right\}\right]=0$;
\end{enumerate}
where $\mu_j(\theta,P) = E_P(m_j(W_i,\theta))$ and $\sigma_j(\theta,P)=Var_P(m_j(W_i,\theta))$.
\end{theoremEnd}
\begin{proofEnd}

I first show that under the assumptions $Var_P(m_j(W_i,\theta))>\frac{1}{M_j^2}>0$, for any $j\in1,\hdots,7$ in \eqref{moment_inequalities_original}, where $M_j$ do not depend on $P$ and $\theta$. I then demonstrate the same for components \eqref{moment_inequalities_wrongly_agree_1}, \eqref{moment_inequalities_wrongly_agree_0}, and \eqref{moment_inequalities_wrongly_agree_both} that are not identical. Finally, I show that $m_j(W_i,\theta)$ are bounded irrespective of $P$ and $\theta$, and use that to prove that the second claim is true.

Let $\rho_P(X,Y) = \frac{Cov_P(X,Y)}{\sqrt{Var_P(X)Var_P(Y)}}$ for some binary random vector $(X,Y)$ with distribution $P\in\mathbf{P}$. The following Lemma will be used to bound the variances from below. 

\begin{lemma}\label{correlation_bounding_lemma}
Suppose that Assumption \ref{family_restriction_assumption} holds. Then for any $P\in\mathbf{P}$, the following are true:
\begin{enumerate}
    \item $  \max\limits_{P\in\mathbf{P}}{\rho_P(r_i,t_i)^2}  = (1-4\varepsilon)^2<1$;
    \item $  \max\limits_{P\in\mathbf{P}}{\rho_P(r_i,(1-t_i))^2}  = (1-4\varepsilon)^2<1$;
    \item $  \max\limits_{P\in\mathbf{P}}{\rho_P(r_i,r_it_i)^2} = h(\varepsilon) $;
    \item $  \max\limits_{P\in\mathbf{P}}{\rho_P(r_i,(1-r_i)t_i)^2} =h(\varepsilon) $
    \item $  \max\limits_{P\in\mathbf{P}}{\rho_P(r_i,r_i(1-t_i))^2} =h(\varepsilon) $
    \item $  \max\limits_{P\in\mathbf{P}}{\rho_P(r_i,(1-r_i)(1-t_i))^2} =h(\varepsilon) $
\end{enumerate}
where $h(\varepsilon) =\mathbbm{1}\{\varepsilon\in[0.2,0.25]\}\frac{2-6\varepsilon}{3-6\varepsilon}+\mathbbm{1}\{\varepsilon\in(0,0.2)\}\left(1-\frac{(1-\varepsilon)^2}{(1+\varepsilon)^2}\right)\in(0,1)$.
\end{lemma}
\begin{proof}
Denote $P(t_i=j,r_i=k)=P_{jk}$. Assumption \ref{family_restriction_assumption} states that for $(j,k)\in\{0,1\}^2$, $P_{jk}\geq\varepsilon>0$, and implies that $\varepsilon\leq \frac{1}{4}$.

\textbf{Statements 1 and 2}

Parameter $\rho_P(r_i,t_i)^2$ is the largest when either $P_{01}=P_{10}=\varepsilon$ or $P_{11}=P_{00}=\varepsilon$. I prove the statement for $P_{01}=P_{10}=\varepsilon$, and the argument for the $P_{11}=P_{00}=\varepsilon$ is symmetric. The maximal $\rho_P(r_i,t_i)^2$ must then be for $P_{11}+P_{00}=1-2\varepsilon$ and $P(t_i=1)=P(r_i=1)$. 

Next, let $P_{11}=\alpha(1-2\varepsilon)$, $P_{00}=(1-\alpha)(1-2\varepsilon)$ for some $\alpha\in[\frac{\varepsilon}{1-2\varepsilon},\frac{1-3\varepsilon}{1-2\varepsilon}]$, and $P(t_i=1)=P(r_i=1)=\alpha(1-2\varepsilon)+\varepsilon$. By plugging in the relevant probabilities, $\rho_P(r_i,t_i)$ becomes a function of $\alpha$:
\begin{align}\label{maximizing_correlation_m4}
\begin{split}
  \rho_\alpha(r_i,t_i) &= \frac{P_{11}-P(t_i=1)P(r_i=1)}{\sqrt{P(t_i=1)(1-P(t_i=1))P(r_i=1)(1-P(r_i=1))}} =\\
  &=\frac{\alpha(1-2\varepsilon)-\left(\alpha(1-2\varepsilon)+\varepsilon\right)^2}{\left(\alpha(1-2\varepsilon)+\varepsilon\right)\left(1-\alpha(1-2\varepsilon)-\varepsilon\right)}.
\end{split}
\end{align}

Since we are considering the case $P_{01}=P_{10}=\varepsilon$, the correlation is positive. By maximizing $\rho_\alpha(r_i,t_i)$ with respect to $\alpha$, we obtain the upper bound on $\rho_P(r_i,t_i)^2$. The second order condition confirms that this is a concave optimization problem. The first order condition yields the maximizing $\alpha^*=\frac{1}{2}$.

For any $\varepsilon\leq \frac{1}{4}$, it is true that $\alpha^*\in[\frac{\varepsilon}{1-2\varepsilon},\frac{1-3\varepsilon}{1-2\varepsilon}]$. To conclude the proof of statement $1$, plug in $\alpha^*$ into \eqref{maximizing_correlation_m4} to find $\max_{P\in\mathbf{P}}{\rho_P(r_i,t_i)}  = \rho_{\alpha^*}(r_i,t_i) = (1-4\varepsilon)$.

By using Statement $1$ and replacing $\tilde{t}_i=1-t_i$, it follows directly that $\max_{P\in\mathbf{P}}{\rho_P(r_i,(1-t_i))}=\max_{P\in\mathbf{P}}{\rho_P(r_i,\tilde{t}_i)} = \rho_{\alpha^*}(r_i,\tilde{t}_i) = (1-4\varepsilon)$.

\textbf{Statement 3}

From the definition of $\rho_P(r_i,r_it_i)$:

\begin{align}\label{maximizing_correlation_m6}
\begin{split}
        \rho_P(r_i,r_it_i) &= \frac{Cov_P(r_i,t_ir_i)}{\sqrt{Var_P(r_i)Var_P(t_ir_i)}} = \frac{E_P(t_ir_i)(1-E_P(r_i))}{\sqrt{E_P(r_i)(1-E_P(r_i))E_P(t_ir_i)(1-E_P(t_ir_i))}}\\
        &=\sqrt{\frac{E_P(t_ir_i)(1-E_P(r_i))}{E_P(r_i)(1-E_P(t_ir_i))}}=\sqrt{\frac{P_{11}(1-P(r_i=1))}{P(r_i=1)(1-P_{11})}}\\
        &=\sqrt{\frac{P_{11}(1-P_{11}-P_{01})}{(P_{11}+P_{01})(1-P_{11})}}.
\end{split}
\end{align}
Notice that $\rho_P(r_i,r_it_i)$ decreases in $P_{01}$, so at the maximum, $P_{01}=\varepsilon$. Therefore, we only need to maximize $\rho_P(r_i,r_it_i)^2$ with respect to feasible $P_{11}$. The maximization problem is:
\begin{align}\label{maximization_problem_correlation_m6}
    \max\limits_{P\in\mathbf{P}}{\rho_P(r_i,r_it_i)}= \max_{P_{11}\in{[\varepsilon,1-3\varepsilon]}}\sqrt{\frac{P_{11}(1-P_{11}-\varepsilon)}{(P_{11}+\varepsilon)(1-P_{11})}}.
\end{align}

The objective function is concave. The first order condition implies that for an interior maximum, the maximizing $P_{11}$ is $\frac{1-\varepsilon}{2}$. If $\varepsilon\in[0.2,0.25]$, the constraint $P_{11}\leq1-3\varepsilon$ will bind. Therefore, the value of the parameter at the maximum is $P_{11}^*=\min\left\{\frac{1-\varepsilon}{2},1-3\varepsilon\right\}$. The maximum of the objective function obtained by plugging in $P_{11}^*$ into \eqref{maximizing_correlation_m6} is:
\begin{align}\label{maximum_correlation_m6}
    \begin{split}
        \max\limits_{P\in\mathbf{P}}{\rho_P(r_i,r_it_i)^2}=\mathbbm{1}\{\varepsilon\in[0.2,0.25]\}\left(1-\frac{1}{3-6\varepsilon}\right)+\mathbbm{1}\{\varepsilon\in(0,0.2)\}\left(1-\frac{(1-\varepsilon)^2}{(1+\varepsilon)^2}\right)\in(0,1)
    \end{split}
\end{align}

\textbf{Statements 4, 5, and 6}

Following the definition of $\rho_P(r_i,(1-r_i)t_i)$:
\begin{align}\label{maximizing_correlation_m5}
\begin{split}
        \rho_P(r_i,(1-r_i)t_i) &= \frac{Cov_P(r_i,(1-r_i)t_i)}{\sqrt{Var_P(r_i)Var_P((1-r_i)t_i)}}\\
        &= \frac{-E_P(r_i)E_P((1-r_i)t_i)}{\sqrt{E_P(r_i)(1-E_P(r_i))E_P((1-r_i)t_i)(1-E_P((1-r_i)t_i))}}\\
        &=-\sqrt{\frac{E_P(r_i)E_P((1-r_i)t_i)}{(1-E_P(r_i))(1-E_P((1-r_i)t_i))}}\\
        &=-\sqrt{\frac{P(r_i=1)P_{10}}{(1-P(r_i=1))(1-P_{10})}}.
\end{split}
\end{align}
The square of the correlation is increasing in both $P(r_i=1)=P_{11}+P_{01}$ and $P_{10}$. Consequently, at the maximum, together they will be at the upper bound, meaning that $P_{11}+P_{01}+P_{10}=1-\varepsilon$, or equivalently, that $P(r_i=1)=1-\varepsilon-P_{10}$. We can then rewrite the problem as:
\begin{align}\label{maximization_problem_correlation_m5}
    \max\limits_{P\in\mathbf{P}}{\rho_P(r_i,(1-r_i)t_i)}^2= \max_{P_{10}\in{[\varepsilon,1-3\varepsilon]}}\frac{(1-\varepsilon-P_{10})P_{10}}{(\varepsilon+P_{10})(1-P_{10})}.
\end{align}

In this form, the problem is identical to the one in \eqref{maximization_problem_correlation_m6}. Following the same steps:
\begin{align}
    \begin{split}
        \max\limits_{P\in\mathbf{P}}{\rho_P(r_i,(1-r_i)t_i)^2}=\mathbbm{1}\{\varepsilon\in[0.2,0.25]\}\left(1-\frac{1}{3-6\varepsilon}\right)+\mathbbm{1}\{\varepsilon\in(0,0.2)\}\left(1-\frac{(1-\varepsilon)^2}{(1+\varepsilon)^2}\right)<1.
    \end{split}
\end{align}

Analogously to the proof of Statement $3$, for $\rho(r_i,(1-t_i)r_i)^2$ in Statement $5$, the optimization problem can be represented as:
\begin{align}\label{maximization_problem_correlation_m1}
    \max\limits_{P\in\mathbf{P}}{\rho_P(r_i,r_i(1-t_i))}^2= \max_{P_{01}\in{[\varepsilon,1-3\varepsilon]}}\frac{(1-\varepsilon-P_{01})P_{01}}{(\varepsilon+P_{01})(1-P_{01})}.
\end{align}

Following the steps in the proof of Statement $4$ $\rho(r_i,(1-r_i)(1-t_i))^2$ in Statement $6$, the optimization problem will be:
\begin{align}\label{maximization_problem_correlation_m2}
    \max\limits_{P\in\mathbf{P}}{\rho_P(r_i,(1-r_i)t_i)}^2= \max_{P_{10}\in{[\varepsilon,1-3\varepsilon]}}\frac{(1-\varepsilon-P_{00})P_{00}}{(\varepsilon+P_{00}P)(1-P_{00})}.
\end{align}

Consequently, from the solutions to \eqref{maximization_problem_correlation_m6} and \eqref{maximizing_correlation_m5}, \eqref{maximization_problem_correlation_m1} and \eqref{maximization_problem_correlation_m2} will yield the same upper bounds on their corresponding squares of correlations:
\begin{align}\label{maximum_correlation_m1_m2}
\begin{split}
     \max\limits_{P\in\mathbf{P}}{\rho_P(r_i,r_i(1-t_i))}^2 &= \max\limits_{P\in\mathbf{P}}{\rho_P(r_i,(1-r_i)t_i)}^2\\
     &= \mathbbm{1}\{\varepsilon\in[0.2,0.25]\}\left(1-\frac{1}{3-6\varepsilon}\right)+\mathbbm{1}\{\varepsilon\in(0,0.2)\}\left(1-\frac{(1-\varepsilon)^2}{(1+\varepsilon)^2}\right).   
\end{split}
\end{align}
\end{proof}

\begin{claim}\label{claim_moment_inequalities_original}
For any $P\in\mathbf{P}$ and $\theta\in[0,1]^2\times\mathcal{S}$ it holds that $Var_P(m_j(W_i,\theta))>0$ for all $m_j(W_i,\theta)$ in \eqref{moment_inequalities_original}.
\end{claim}
\begin{proof}
Consider first a component of $m$ pertaining to the upper bound of $\theta_1$. The variance $Var_P(m_4(W_i,\theta))$ for some $\theta$ and $P$ is defined as:
\begin{align}\label{m4_variance}
\begin{split}
   Var_P(m_4(W_i,\theta)) &= Var_P\left(\theta_1\frac{r_i-1+s_0}{s_1-1+s_0}-t_i\right) \\
  & =\left(\frac{\theta_1}{s_1-1+s_0}\right)^2Var_P(r_i)+Var_P(t_i)-2\frac{\theta_1}{s_1-1+s_0}Cov_P(r_i,t_i).
\end{split}
\end{align}
Fix any $(s_1,s_0)\in\mathcal{S}$. As discussed in \Cref{imperfect_knowledge}, Assumptions \ref{reference_performance_assumption_relaxed} and \ref{prevalence_assumption} imply $P(r=1)\in(1-s_0,s_1)$ so $Var_P(r_i)>0$. The value $\theta_1^*$ where $Var_P(m_4(W_i,\theta))$ is globally minimized given $s_1$ and $s_0$ from the first order condition is:
\begin{align}\label{m4_variance_foc}
    \begin{split}
        \frac{\partial Var_P(m_4(W_i,\theta))}{\partial \theta_1}: \enskip \theta_1^* = (s_1-1+s_0)\frac{Cov_P(r_i,t_i)}{Var_P(r_i)}.
    \end{split}
\end{align}
The second order condition shows that this indeed is a minimization problem. Let $\theta^* = (\theta_1^*,\theta_0,s_1,s_0)$, where I suppress the dependence $\theta_1^*(s_1,s_0)$ for clarity. The minimum variance for any $(s_1,s_0)\in\mathcal{S}$ is then:
\begin{align}\label{m4_variance_minimal}
    \begin{split}
        Var_P(m_4(W_i,\theta^*)) &= \frac{(Cov_P(r_i,t_i))^2}{Var_P(r_i)}+Var_P(t_i)-2\frac{(Cov_P(r_i,t_i))^2}{Var_P(r_i)}\\
        &=Var_P(t_i)\left(1-\rho_P(r_i,t_i)^2)\right).
    \end{split}
\end{align}
For any $\theta$ it follows:
\begin{align}\label{m4_variance_minimal_a_negative}
    \begin{split}
        Var_P(m_4(W_i,\theta))&\geq Var_P(m_4(W_i,\theta^*))= Var_P(t_i)\left(1-\rho_P(r_i,t_i)^2)\right)\\
        &\geq2\varepsilon(1-2\varepsilon)\left(1-(1-4\varepsilon)^2\right)=\frac{1}{M_4^2}>0
    \end{split}
\end{align}
where the first inequality follows from the definition of $\theta^*$. Focus on the second inequality. We wish to find the lower bound on the variance $\frac{1}{M_4^2}$ over all possible $P\in\mathbf{P}$. One such bound is equal to the expression at the smallest value of $Var_P(t_i)$ and the largest value of $\rho_P(r_i,t_i)^2$. The second is given by \Cref{correlation_bounding_lemma}, and the first follows directly from Assumption \ref{family_restriction_assumption} which implies that $P(t_i=1)\in[2\varepsilon,1-2\varepsilon]$, so $Var_P(t_i)\geq 2\varepsilon(1-2\varepsilon)$.\footnote{As long as $\varepsilon<0.25$, the inequality is strict, since the largest value of $\rho_P(r_i,t_i)^2$ warrants that $P(t_i=1,r_i=1)=\frac{1-2\varepsilon}{2}$ while the smallest $Var_P(t_i)$ requires $P(t_i=1,r_i=1)=\varepsilon$ or $P(t_i=1,r_i=1)=1-3\varepsilon$.} Therefore, $Var_P(m_4(W_i,\theta))\geq\frac{1}{M_4^2}>0$ for all $P\in\mathbf{P}$ and $\theta\in[0,1]^2\times\mathcal{S}$.

Following the same steps for the remaining components pertaining to the upper bound, the smallest variances for any $P\in\mathbf{P}$ and $\theta$ are:
\begin{align}
\begin{split}
    Var_P(m_5(W_i,\theta^*)) &= Var_P(t_i(1-r_i))\left(1-\rho_P(r_i,t_i(1-r_i))^2)\right)\geq\varepsilon(1-\varepsilon)\left(1-h(\varepsilon)\right)=\frac{1}{M_5^2}>0\\
    Var_P(m_6(W_i,\theta^*)) &= Var_P(t_ir_i)\left(1-\rho_P(r_i,t_ir_i)^2)\right)\geq\varepsilon(1-\varepsilon)\left(1-h(\varepsilon)\right)=\frac{1}{M_6^2}>0
\end{split}
\end{align}

where the inequalities follow from the definition of $\theta^*$, the fact that $Var_P(t_i(1-r_i))\geq\varepsilon(1-\varepsilon)$ and $Var_P(t_i(1-r_i))\geq\varepsilon(1-\varepsilon)$, and \Cref{correlation_bounding_lemma}.

Next observe the components pertaining to the lower bound. First for $Var_P(m_1(W_i,\theta))$ for any $\theta$ and $P$:
\begin{align}\label{m1_variance}
\begin{split}
   Var_P(m_1(W_i,\theta)) &= Var_P\left((-\theta_1+s_1)\frac{r_i-1+s_0}{s_1-1+s_0}+(t_i-1)r_i\right) \\
    & =\left(\frac{s_1-\theta_1}{s_1-1+s_0}\right)^2Var_P(r_i)+Var_P((t_i-1)r_i)-2\frac{s_1-\theta_1}{s_1-1+s_0}Cov_P((1-t_i)r_i,r_i)
\end{split}
\end{align}

Fix an arbitrary $s_1$ and $s_0$. The value $\theta_1^*$ where $Var_P(m_1(W_i,\theta))$ is globally minimized given $s_1$ and $s_0$ from the first order condition is:
\begin{align}\label{m1_variance_foc}
    \begin{split}
        \frac{\partial Var_P(m_1(W_i,\theta))}{\partial \theta_1}: \enskip \theta_1^* = (s_1-1+s_0)\frac{Cov_P((1-t_i)r_i,r_i)}{Var_P(r_i)}+s_1.
    \end{split}
\end{align}
The second order condition shows that this indeed is a minimization problem. The minimum variance $Var_P(m_1(W_i,\theta^*))$ for an arbitrary $(s_1,s_0)\in\mathcal{S}$ is:
 \begin{align}\label{m1_variance_minimal}
    \begin{split}
 Var_P(m_1(W_i,\theta^*))&=\left(-\frac{Cov_P((1-t_i)r_i,r_i))}{Var_P(r_i)}\right)^2Var_P(r_i)+Var_P((1-t_i)r_i)\\
 &-2\left(\frac{Cov_P((1-t_i)r_i,r_i)}{Var_P(r_i)}\right)Cov_P((1-t_i)r_i,r_i)\\
 &=Var_P((1-t_i)r_i)\left(1-\frac{(Cov_P((1-t_i)r_i,r_i))^2}{Var_P(r_i)Var_P((1-t_i)r_i)}\right) \\
 &= Var_P((1-t_i)r_i)\left(1-\rho_P((1-t_i)r_i,r_i)^2\right)\\
 &\geq\varepsilon(1-\varepsilon)(1-h(\varepsilon))=\frac{1}{M_1^2}>0
    \end{split}
\end{align}

where the first inequality follows from the definition of $\theta^*$. And the second follows from \Cref{correlation_bounding_lemma} and $Var_P(t_ir_i)\geq\varepsilon(1-\varepsilon)$. Therefore $Var_P(m_1(W_i,\theta))\geq\frac{1}{M_1^2}>0$ for all $P\in\mathbf{P}$ and $\theta\in[0,1]^2\times\mathcal{S}$.

Again, following the same steps for the remaining components pertaining to the lower bound, the smallest variances for any $P\in\mathbf{P}$ and $\theta$ are:
\begin{align}
\begin{split}
    Var_P(m_2(W_i,\theta^*)) &= Var_P((1-t_i)(1-r_i))\left(1-\rho_P(r_i,(1-t_i)(1-r_i))^2)\right)\\&
    \geq\varepsilon(1-\varepsilon)\left(1-h(\varepsilon)\right)=\frac{1}{M_2^2}>0\\
    Var_P(m_3(W_i,\theta^*)) &= Var_P(1-t_i)\left(1-\rho_P(r_i,(1-t_i))^2)\right)\geq 2\varepsilon(1-2\varepsilon)\left(1-(1-4\varepsilon)^2\right)=\frac{1}{M_3^2}>0
\end{split}
\end{align}

Finally, consider the component pertaining to the moment equality $Var(m_7(W_i,\theta))$. It is defined as:
\begin{align}\label{m7_variance}
    \begin{split}
    Var_P(m_7(W_i,\theta)) &= Var_P\left((1-\theta_0)\left(1-\frac{r_i-1+s_0}{s_1-1+s_0}\right)+\theta_1\frac{r_i-1+s_0}{s_1-1+s_0}-t_i\right)\\
    &= Var_P\left((\theta_0+\theta_1-1)\frac{r_i-1+s_0}{s_1-1+s_0}-t_i\right)\\
    &= Var_P\left(\Bar{\theta}\frac{r_i-1+s_0}{s_1-1+s_0}-t_i\right)\\  
    & =\left(\frac{\Bar{\theta}}{s_1-1+s_0}\right)^2Var_P(r_i)+Var_P(t_i)-2\frac{\Bar{\theta}}{s_1-1+s_0}Cov_P(r_i,t_i)
    \end{split}
\end{align}

for $\theta_1+\theta_0-1 = \Bar{\theta}$. Notice how the function \eqref{m7_variance} resembles \eqref{m4_variance}. Following the same steps as for finding $\frac{1}{M_4^2}$, we obtain that $Var(m_7(W_i,\theta))\geq 2\varepsilon(1-2\varepsilon)\left(1-(1-4\varepsilon)^2\right)=\frac{1}{M_7^2}>0$ for all $P\in\mathbf{P}$ and $\theta\in[0,1]^2\times\mathcal{S}$.
\end{proof}

\begin{claim}\label{claim_moment_inequalities_wrongly_agree_once}
For any $P\in\mathbf{P}$ and $\theta\in[0,1]^2\times\mathcal{S}$ it holds that $Var_P(m_j(W_i,\theta))>0$ for all $m_j(W_i,\theta)$ in \eqref{moment_inequalities_wrongly_agree_1}, \eqref{moment_inequalities_wrongly_agree_0}, and \eqref{moment_inequalities_wrongly_agree_both}.
\end{claim}

\begin{proof}
Functions $\Bar{m}$ and $m$ are such that $\Bar{m}^1_j(W_i,\theta)\neq m_j(W_i,\theta)$ only if $j=6$. Thus for all components that are equal, the proof follows from Claim \ref{claim_moment_inequalities_original}, so $Var_P(\Bar{m}^1_j(W_i,\theta))\geq\frac{1}{M^2_j}>0$ for $j\neq 6$.

The variance $Var_P(\Bar{m}^1_6(W_i,\theta))$ for some $\theta$ and $P$ is:
\begin{align}\label{m6bar_variance}
\begin{split}
   Var_P(\Bar{m}^1_6(W_i,\theta)) &= Var_P\left(\Big(\theta_1+\frac{-1+s_1}{2}\Big)\frac{r_i-1+s_0}{s_1-1+s_0}-t_ir_i\right) \\
  & =\left(\frac{\theta_1+\frac{-1+s_1}{2}}{s_1-1+s_0}\right)^2Var_P(r_i)+Var_P(r_it_i)-2\frac{\theta_1+\frac{-1+s_1}{2}}{s_1-1+s_0}Cov_P(r_i,r_it_i).
\end{split}
\end{align}
Fix any $(s_1,s_0)\in\mathcal{S}$. The value $\theta_1^*$ where $Var_P(m_6(W_i,\theta))$ is globally minimized given $s_1$ and $s_0$ from the first order condition is:
\begin{align}\label{m6bar_variance_foc}
    \begin{split}
        \frac{\partial Var_P(\Bar{m}^1_6(W_i,\theta))}{\partial \theta_1}: \enskip \theta_1^* = (s_1-1+s_0)\frac{Cov_P(r_i,r_it_i)}{Var_P(r_i)}+\frac{1-s_1}{2}.
    \end{split}
\end{align}
The second order condition shows that this indeed is a minimization problem. Following the same steps as before, for any $\theta\in[0,1]^2\times\mathcal{S}$:
\begin{align}\label{m6bar_variance_minimal_a_negative}
    \begin{split}
        Var_P(Cov_P(r_i,r_it_i)(W_i,\theta))&\geq Var_P(Cov_P(r_i,r_it_i)(W_i,\theta^*))= Var_P(r_it_i)\left(1-\rho_P(r_i,r_it_i)^2)\right)\\
        &\geq\varepsilon(1-\varepsilon)\left(1-h(\varepsilon)\right)=\frac{1}{M_6^2}>0
    \end{split}
\end{align}
The case for $\Bar{m}_j^0(W_i,\theta)$ is symmetric, $Var_P(\Bar{m}_j^0(W_i,\theta))\geq \frac{1}{M_j^2}>0$ for all $j=1,\hdots,7$, $P\in\mathbf{P}$ and $\theta\in[0,1]^2\times\mathcal{S}$.

Likewise, for $\dbar{m}$ note that $\dbar{m}_j(W_i,\theta) = \Bar{m}_j(W_i,\theta)$ except for $j\in\{4,6\}$. From \eqref{moment_inequalities_wrongly_agree_both}:

\begin{align}\label{variance_moment_inequalities_wrongly_agree_both}
    \begin{split}
    V_P(\dbar{m}_4(W_i,\theta)) &= V_P\left(\theta_1\frac{r_i-1+s_0}{s_1-1+s_0}-t_i+\frac{1}{2}\Big(r_i-s_1\frac{r_i-1+s_0}{s_1-1+s_0}\Big)\right) \\
    & =V_P\left(\left(\frac{\theta_1-\frac{s_1}{2}}{s_1-1+s_0}+\frac{1}{2}\right)r_i-t_i\right)\\
    V_P(\dbar{m}_6(W_i,\theta)) &= V_P\left(\Big(\theta_1+\frac{-1+s_1}{2}\Big)\frac{r_i-1+s_0}{s_1-1+s_0}-t_ir_i+\frac{1}{2}\Big(r_i-s_1\frac{r_i-1+s_0}{s_1-1+s_0}\Big)\right)\\
    & =V_P\left(\left(\frac{\theta_1-\frac{1}{2}}{s_1-1+s_0}+\frac{1}{2}\right)r_i-t_ir_i\right).
    \end{split}
\end{align}

As above, for any $\theta\in[0,1]^2\times\mathcal{S}$:
\begin{align}\label{minimum_variance_moment_inequalities_wrongly_agree_both}
    \begin{split}
    V_P(\dbar{m}_4(W_i,\theta))\geq V_P(\dbar{m}_4(W_i,\theta^*)) &= V_P(t_i)\left(1-\rho_P(r_i,t_i)^2\right) \\
    & \geq 2\varepsilon(1-2\varepsilon)\left(1-(1-4\varepsilon)^2\right)=\frac{1}{{M}_4^2}>0\\
      V_P(\dbar{m}_6(W_i,\theta))\geq V_P(\dbar{m}_6(W_i,\theta^*) &= V_P(t_ir_i)\left(1-\rho_P(r_i,t_ir_i)^2\right)\\
    & \geq\varepsilon(1-\varepsilon)\left(1-h(\varepsilon)\right)=\frac{1}{{M}_6^2}>0.
    \end{split}
\end{align}

It is true that $Var_P(\dbar{m}_j^0(W_i,\theta))\geq \frac{1}{M_j^2}>0$ for all $j=1,\hdots,7$, $P\in\mathbf{P}$ and $\theta\in[0,1]^2\times\mathcal{S}$.
\end{proof}

\begin{claim}\label{claim_uniform_integrability_rsw}
For any $P\in\mathbf{P}$, $\theta\in[0,1]^2\times\mathcal{S}$, and $m_j(W_i,\theta)$ in \eqref{moment_inequalities_original}, \eqref{moment_inequalities_wrongly_agree_1}, \eqref{moment_inequalities_wrongly_agree_0}, and \eqref{moment_inequalities_wrongly_agree_both}:
\begin{align}\label{uniform_integrability_rsw}
    \limsup_{\lambda\conv\infty}\sup_{P\in\mathbf{P}}\sup_{\theta\in\mathbf{\Theta}(P)}E_P\left[\left(\frac{m_j(W_i,\theta)-\mu_j(\theta,P)}{\sigma_j(\theta,P)}\right)^2\mathbbm{1}\left\{|\frac{m_j(W_i,\theta)-\mu_j(\theta,P)}{\sigma_j(\theta,P)}|>\lambda\right\}\right]=0.
\end{align}
\end{claim}

We have shown above that for any $\sigma_j(\theta,P)$ corresponding to components $m_j(W_i,\theta)$ in \eqref{moment_inequalities_original}, \eqref{moment_inequalities_wrongly_agree_1}, \eqref{moment_inequalities_wrongly_agree_0}, and \eqref{moment_inequalities_wrongly_agree_both}, there exists a finite constant $M_j>0$ such that $\sigma_j(\theta,P)\geq\frac{1}{M_j}>0$ for all $P\in\mathbf{P}$ and $\theta\in[0,1]^2\times\mathcal{S}$.

Then for any $j$, $P\in\mathbf{P}$, $\theta\in[0,1]^2\times\mathcal{S}$ and $\lambda$:
\begin{align}\label{expectation_sandwich}
     &E_P\left[M_j^2\left(m_j(W_i,\theta)-\mu_j(\theta,P)\right)^2\mathbbm{1}\left\{|m_j(W_i,\theta)-\mu_j(\theta,P)|>\frac{\lambda}{M_j}\right\}\right]\\
     &\geq E_P\left[\left(\frac{m_j(W_i,\theta)-\mu_j(\theta,P)}{\sigma_j(\theta,P)}\right)^2\mathbbm{1}\left\{|\frac{m_j(W_i,\theta)-\mu_j(\theta,P)}{\sigma_j(\theta,P)}|>\lambda\right\}\right]\geq 0.
\end{align}

As $W_i=(t_i,r_i)\in\{0,1\}^2$, and $\mathcal{S}$ is a compact set such that $\forall{(s_1,s_0)\in\mathcal{S}}:s_1>1-s_0$, $|m_j(W_i,\theta)|\leq B_j(\theta)\leq B_j^*<\infty$ for each $j$, where $B_j^*=\max_{\theta\in[0,1]^2\times\mathcal{S}}B_j(\theta)$. That implies that $|\mu_j(\theta,P)|\leq B_j^*<\infty$, and $(m_j(W_i,\theta)-\mu_j(\theta,P))^2\leq 4{B^*}_j^2$. Consequently:
\begin{align}\label{expectation_sandwich_bounds_determined}
     4M_j^2{B^*}_j^2\mathbbm{1}\left\{2{B^*}_j>\frac{\lambda}{M_j}\right\}\geq E_P\left[\left(\frac{m_j(W_i,\theta)-\mu_j(\theta,P)}{\sigma_j(\theta,P)}\right)^2\mathbbm{1}\left\{|\frac{m_j(W_i,\theta)-\mu_j(\theta,P)}{\sigma_j(\theta,P)}|>\lambda\right\}\right]\geq 0.
\end{align}

Finally, since neither $B_j^*$ nor $M_j$ depend on $P$ or $\theta$, it follows that \eqref{uniform_integrability_rsw} holds, concluding the proof.

\end{proofEnd}

\Cref{family_theorem} enables us to use the method from \citet{romano2014practical} to construct confidence sets $C_n$ for points $(\theta_1,\theta_0)$ in the identified sets defined by \Cref{bounds_theta_prop} and \Cref{bounds_wrongly_agree_prop} that satisfy \eqref{conf_set_definition} when the relevant family of population distributions conforms to Assumption \ref{family_restriction_assumption}.

\begin{remark}\label{rem:sidak}
    In certain cases researchers may estimate $(s_1,s_0)$ based on an independent sample of size $m$, rather than assume them to be known. For example, this happens if one has access to a sample from an independent validation study identifying performance of $r$. It is possible to account for statistical imprecision of both samples using \v{S}id\'{a}k's correction for independent tests. (\citet{lehmann2022testing}, Chapter 9.1.2) Let $\alpha_{S} = 1-(1-\alpha)^{\frac{1}{2}}$. One can construct an asymptotic confidence set for $(s_1,s_0)$ at the  $1-\alpha_{S}$ confidence level, and treat it as $\mathcal{S}$ in the inference procedure. Then, the confidence set $C_n$ at the significance level $\alpha_{S}$ ensures at least $1-\alpha$ asymptotic coverage of $(\theta_1,\theta_0)$ as $n,m\conv\infty$. For example, if $\alpha = 5\%$, then $\alpha_S = 2.53\%$ which modestly improves upon the Bonferroni correction.
\end{remark}

\section{Application - Abbott BinaxNOW COVID-19 Antigen Test}\label{application_sect}

In this section, I apply the developed method to existing study data to provide confidence and estimated identified sets for $(\theta_1,\theta_0)$ of the rapid antigen COVID-19 test with the currently highest market share in the United States - \textit{Abbott BinaxNOW COVID-19 Ag2 CARD} test.

 Template for Developers of Antigen Tests required by the FDA for EUA mandates that the reference for all COVID-19 antigen test studies must be an approved RT-PCR test.\footnote{Link: \href{https://www.fda.gov/media/137907/download}{https://www.fda.gov/media/137907/download} (Last accessed: 12/25/2022)} However, \citet{arevalo2020false}, \citet{kucirka2020variation}, \citet{drame2020should}, \citet{huerta2020should}, and \citet{kanji2021false} explain that these tests are imperfectly sensitive. Using them as a reference yields ``apparent'' and generally not true sensitivity and specificity. Interpreting the results as measures of true performance may be severely misleading. \citet{fitzpatrick2021buyer} emphasize that the false negative rate of the \textit{BinaxNOW} test may be substantially understated by the reported ``apparent'' analog due to imperfect reference tests. They highlight that ``apparent'' measures can unjustifiably lead the users to believe that the test must have high sensitivity, even when this is not true. 

To verify the claim, I revisit the test performance study data from the submitted EUA documentation\footnote{Link: \href{https://www.fda.gov/media/141570/download}{https://www.fda.gov/media/141570/download} (Last accessed: 12/25/2022)}, as well data from an independent study by \citet{shah2021performance}. I compare the results with the corresponding ``apparent'' estimates from the original documentation and the instructions for use pamphlet. I also use the data to compare the developed method with existing bounds by \citet{thibodeau1981evaluating} and \citet{emerson2018biomarker}, henceforth referred to as comparable methods.

 By established notation, $t$ is the antigen test, $r$ is the RT-PCR test and $y$ determines whether the person truly has COVID-19. To construct the confidence sets, I implement the test from \citet{romano2014practical} denoted by $\phi_n^{RSW2}$ in \citet{bai2021two}. The test relies on the maximum statistic $T_n = \max\left\{\max_{1\leq j\leq k} \frac{\sqrt{n}\Bar{m}_j}{S_j},0\right\}$, where $\Bar{m}_j =\frac{1}{n} \sum_{i=1}^n m_j(W_i,\theta)$ and $S_j^2 =\frac{1}{n} \sum_{i=1}^n( m_j(W_i,\theta)-\Bar{m}_j)^2$ for a value $\theta$ and components of the appropriate moment function $m_j(W_i,\theta)$ $j=1,\hdots,k$. The testing procedure has two steps: 1) Construction of confidence regions for the moments; 2) Formation of a critical value incorporating information on which moment inequalities are ``negative''. I perform test inversion over a fine grid of $10^5$ points for the relevant parameter space for $(\theta_1,\theta_0)$, and additionally over $10$ points over $\mathcal{S}$, where applicable. Following the original paper, I use $500$ bootstrap samples to find the critical values and set $\beta=\alpha/10$. The results do not change significantly with alternative values $\beta=\alpha/5$ and $\beta=\alpha/20$.

 \subsection{Identification Assumptions}

Use of comparable and bounds developed in the paper requires a credible set of values $(s_1,s_0)\in\mathcal{S}$ for the reference RT-PCR test. I maintain that $s_0=1$ following \citet{kucirka2020variation} who do the same, citing perfect analytical specificity.\footnote{Specificity on contrived laboratory samples containing other pathogens, but not SARS-CoV-2.} The same assumption has been used in other existing work, such as \citet{manski2020bounding}, \citet{manski2021estimating} \citet{kanji2021false}, \citet{ziegler2021binary}, and \citet{stoye2022bounding}.

In the absence of a perfect gold standard, a conventional diagnostic test performance study cannot identify sensitivity $s_1$ of the RT-PCR tests. Some studies rely on a different set of assumptions to identify the parameter of interest. \citet{kanji2021false} provide a discordant result analysis of the RT-PCR test used for frontline testing of symptomatic individuals. The authors define discordant results as initially negative RT-PCR findings followed by a positive test result within the incubation period. The negative samples were retested by three alternative RT-PCR assays targeting different genes. If at least two alternative assays yielded positive results, the initial result was considered to be a false negative finding. Assuming perfect specificities of each assay, and perfect sensitivity of the combined testing procedure, they estimate the sensitivity of the used RT-PCR test at $90.3\%$. Perfect specificity is maintained based on perfect analytical specificity. \citet{arevalo2020false} use data from published studies to estimate sensitivity, defining false negatives as patients who were symptomatic and negative, but subsequently positive on an RT-PCR test within the incubation period. It is implicitly maintained that all initial results must have been false negatives. Three estimates are based on data from the United States. Sensitivity of $90\%$ is the only estimate which is not considered to be at high risk of bias according to the QUADAS-2 tool (\citet{whiting2011quadas}). Following the two references, I assume that $s_1=0.9$. One should note that assumed $(s_1,s_0)$ are a critical identifying assumption which directly affect the obtained estimates. Appendix \ref{sect_sensitivity_analysis} thus discusses robustness of findings to different assumed values of $s_1$.

Since the antigen and RT-PCR test rely on the same type of sample, I maintain that they have a tendency to wrongly agree for $y=1$. Since it is assumed that $s_0=1$, the tendency to wrongly agree for $y=0$ has no identifying power, as explained by \Cref{remark_identifying_power}. It is thus not maintained.

\subsection{Data and Results}

EUA documentation and the instructions for use outline the initial performance study. The estimates were obtained on a sample of $460$ participants tested within $7$ days of symptom onset. \citet{shah2021performance} perform the same analysis on an independent sample of $2110$ individuals enrolled at a community testing site. This includes $1188$ symptomatic individuals, of which $929$ were tested within $7$ days of symptom onset. I omit symptomatic individuals tested more than $7$ days after initial symptoms for comparability with the EUA study. I separately analyze the performance on $877$ asymptomatic participants to provide plausible estimates of performance in the absence of symptoms. The data are summarized in \Cref{tab:data}. In all three samples, estimates of joint probabilities $\hat{P}(t=j,r=k)$ for $(r,k)\in\{0,1\}^2$ are bounded away from zero. That the distributions which have generated the data lie in a family of distributions satisfying Assumption \ref{family_restriction_assumption} is reasonable.

\begin{table}[htbp]
  \centering
  \begin{threeparttable} 
    \caption{Study Data}
      \label{tab:data}%
    \begin{tabular}{cccccc}
    \toprule
     &
       &
      \multicolumn{4}{c}{$(t_i,r_i)$}
      \\
\cmidrule{3-6}    Data &
      $N$ &
      $(1,1)$ &
      $(0,1)$ &
      $(1,0)$ &
      $(0,0)$
      \\
    \midrule
    EUA Sx &
      $460$ &
      $99$ &
      $18$ &
      $5$ &
      $338$
      \\
    \citet{shah2021performance} Sx &
      $929$ &
      $199$ &
      $44$ &
      $2$ &
      $684$
      \\
    \citet{shah2021performance} ASx &
      $877$ &
      $33$ &
      $15$ &
      $5$ &
      $824$
      \\
    \bottomrule
    \bottomrule
    \end{tabular}%
        \begin{tablenotes}
      \small
      \item \textit{Note: }Number of outcomes $(t_i,r_i)$ in analyzed studies. Sx denotes the symptomatic, and ASx the asymptomatic individuals.
    \end{tablenotes}
    \end{threeparttable}
\end{table}

Panels (\subref{figure_EUA:conf_set}), (\subref{figure_shah_sx:conf_set}), and (\subref{figure_shah_asx:conf_set}) of \Cref{joint_figure} represent the estimated ``apparent'' operating characteristics and joint identified sets for $(\theta_1,\theta_0)$, as well as corresponding $95\%$ confidence sets. The joint confidence set for apparent measures is the projection Clopper-Pearson exact confidence set. The areas of the two confidence sets are similar. \Cref{tab:projected_bounds} show the estimates of projected individual bounds on sensitivity and specificity. The bounds are revealing.

\begin{table}[htbp]
  \centering
  \begin{threeparttable}
    \caption{Estimates amd Estimated Projection Bounds}
    \label{tab:projected_bounds}%
    \begin{tabular}{ccccc}
    \toprule
     &
       &
      \multicolumn{3}{c}{Data}
      \\
\cmidrule{3-5}     &
       &
      EUA Sx &
      Shah et al. Sx &
      \multicolumn{1}{l}{Shah et al. ASx}
      \\
    \midrule
    \multicolumn{1}{c}{\multirow{4}[2]{*}{$\theta_1$ Estimates}} &
      Apparent &
      $0.846$ &
      $0.819$ &
      $0.688$
      \\
     &
      Projection &
      $[0.761,0.800]$ &
      $[0.737,0.744]$ &
      $[0.619, 0.669]$
      \\
     &
      \citet{emerson2018biomarker} &
      $[0.761,0.800]$ &
      $[0.737,0.744]$ &
      $[0.619, 0.712]$
      \\
     &
      \citet{thibodeau1981evaluating} &
      $[0.761,0.846]$ &
      $[0.737,0.819]$ &
      $[0.619, 0.688]$
      \\
    \midrule
    \multicolumn{1}{c}{\multirow{4}[2]{*}{$\theta_0$ Estimates}} &
      Apparent &
      $0.985$ &
      $0.997$ &
      $0.994$
      \\
     &
      Projection &
      $[0.985,1.000]$ &
      $[0.997,1.000]$ &
      $[0.994,0.997]$
      \\
     &
      \citet{emerson2018biomarker} &
      $[0.985,1.000]$ &
      $[0.997,1.000]$ &
      $[0.994, 1.000]$
      \\
     &
      \citet{thibodeau1981evaluating} &
      $[0.985,1.000]$ &
      $[0.997,1.000]$ &
      $[0.994,0.998]$
      \\
    \bottomrule
    \bottomrule
    \end{tabular}%
        \begin{tablenotes}
      \small
      \item \textit{Note: }Estimates of ``apparent'' performance measures and projections of estimated identified sets for $(\theta_1,\theta_0)$ shown in \Cref{joint_figure}. Sx denotes the symptomatic, and ASx the asymptomatic individuals.
    \end{tablenotes}

    \end{threeparttable}
\end{table}%

\newgeometry{left=3mm, right=5mm, top=5mm, bottom=12mm}
\begin{landscape}
\begin{figure}
    \centering
    \begin{subfigure}[t]{0.41\textwidth}
        \centering
        \includegraphics[width=\linewidth]{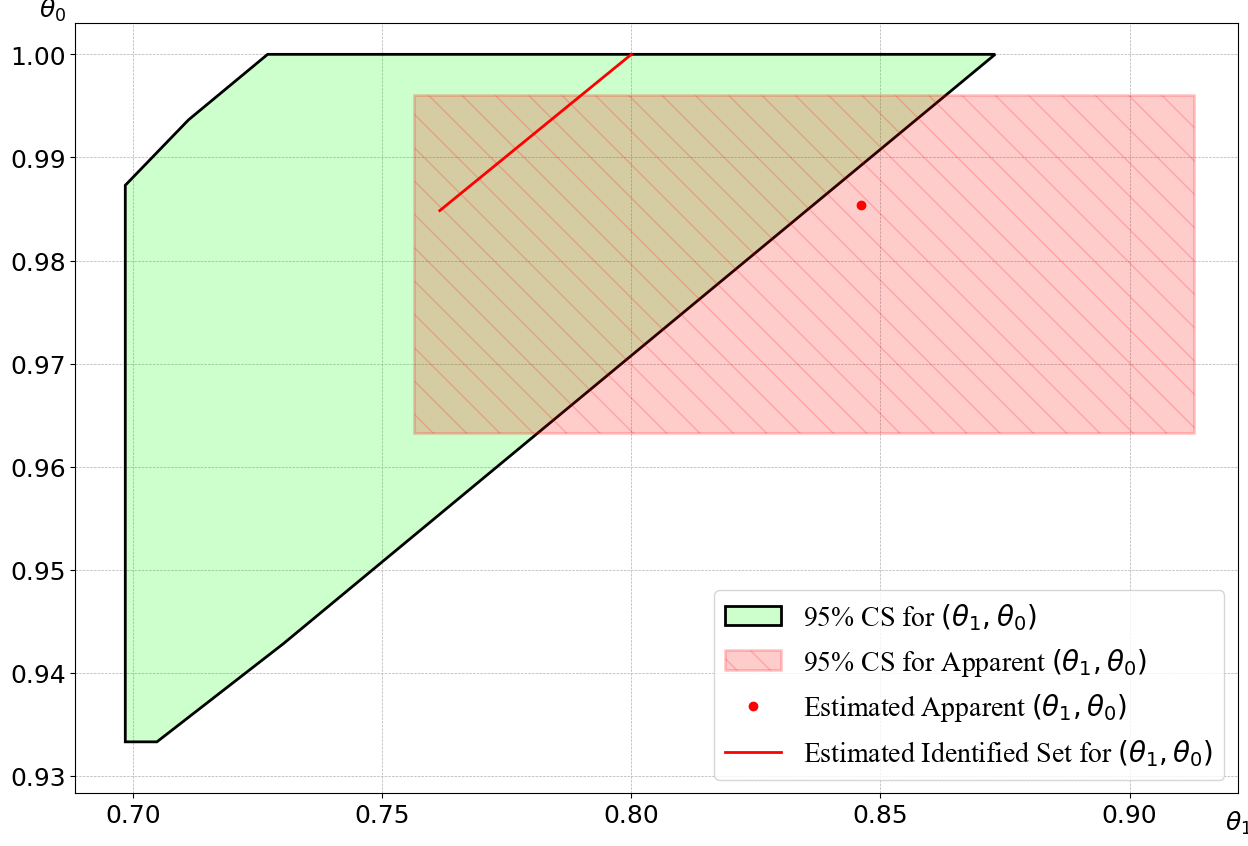} 
        \caption{EUA: Results} \label{figure_EUA:conf_set}
    \end{subfigure}
    \hfill
    \begin{subfigure}[t]{0.41\textwidth}
        \centering
        \includegraphics[width=\linewidth]{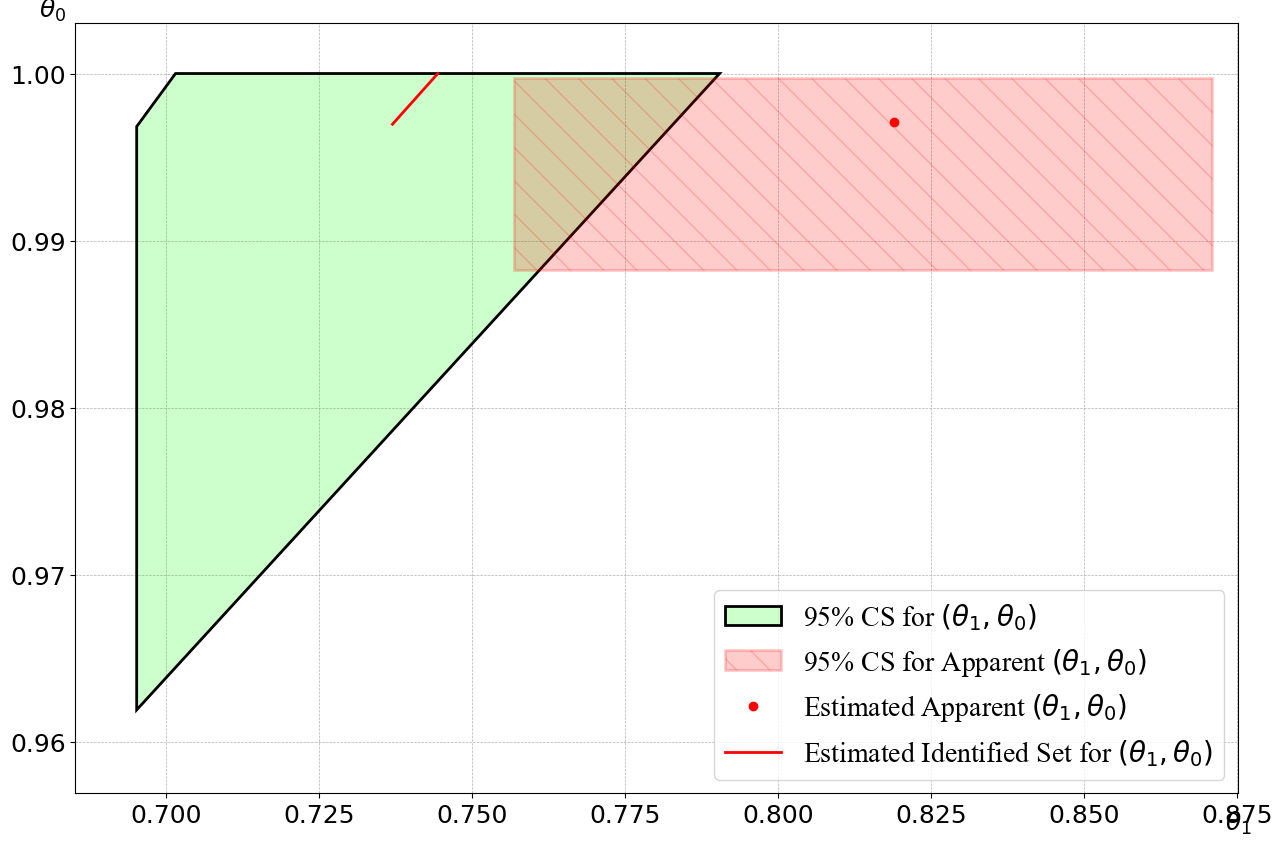} 
        \caption{\citet{shah2021performance} Sx: Results} \label{figure_shah_sx:conf_set}
    \end{subfigure}
     \hfill
    \begin{subfigure}[t]{0.41\textwidth}
        \centering
        \includegraphics[width=\linewidth]{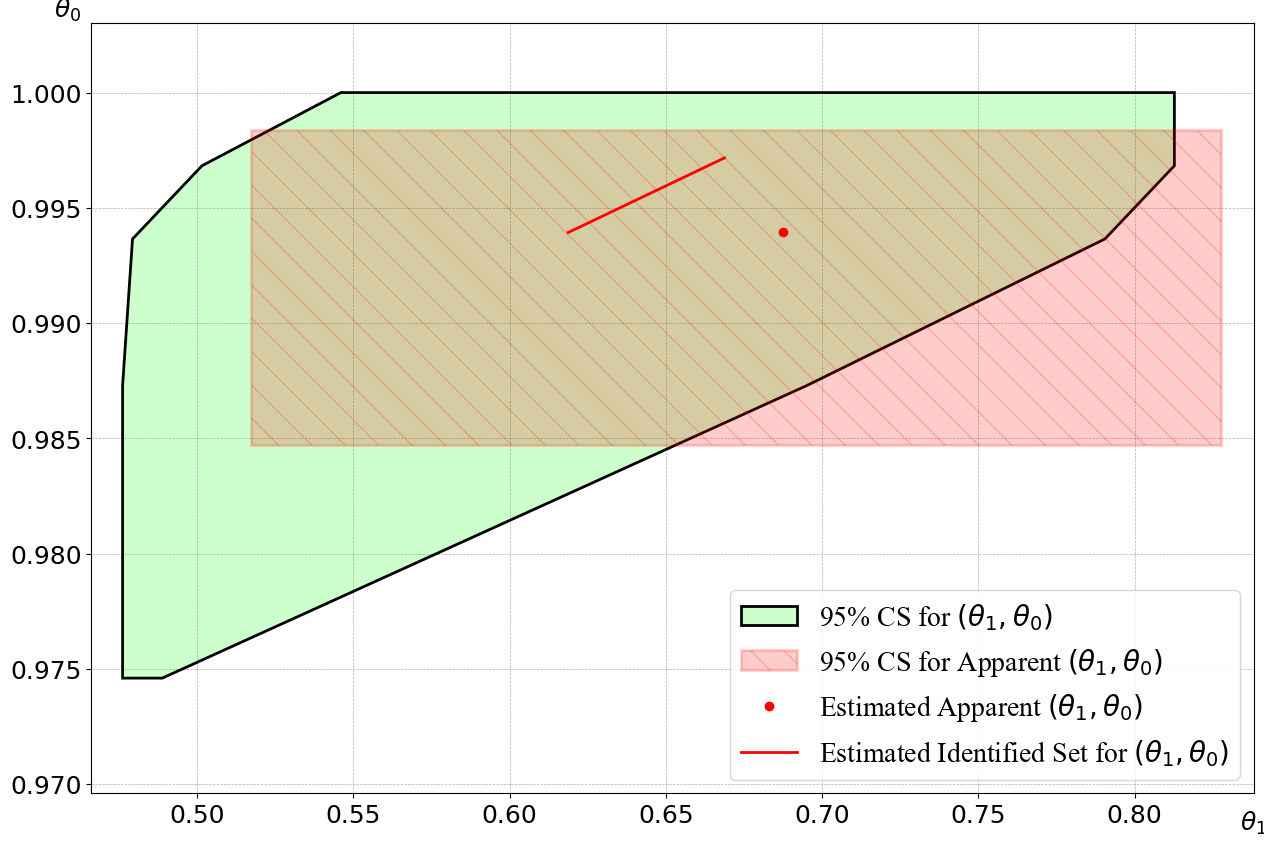} 
        \caption{\citet{shah2021performance} ASx: Results} \label{figure_shah_asx:conf_set}
    \end{subfigure}
    \begin{subfigure}[t]{0.41\textwidth}
        \centering
        \includegraphics[width=\linewidth]{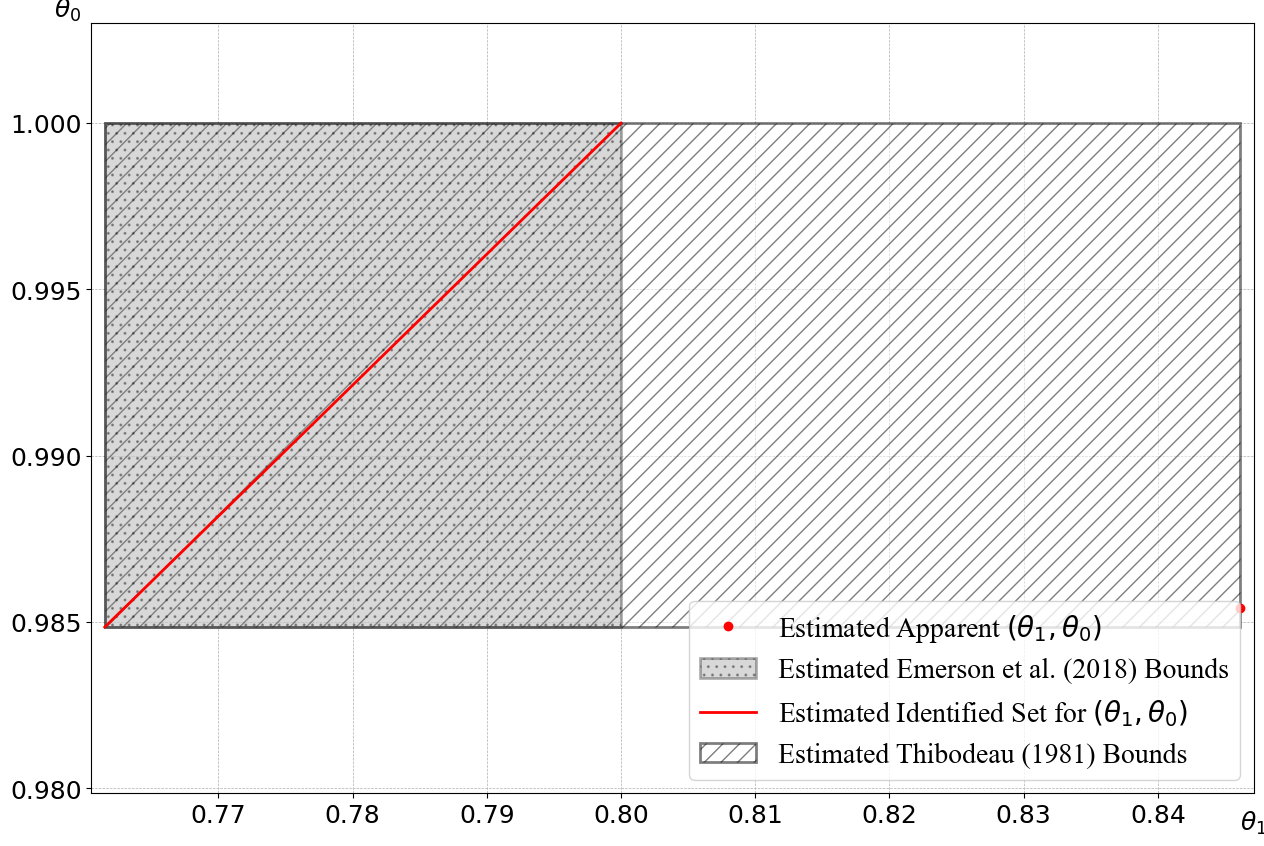} 
        \caption{EUA: Methods Comparison} \label{figure_EUA:comparison}
    \end{subfigure}
    \hfill
    \begin{subfigure}[t]{0.41\textwidth}
        \centering
        \includegraphics[width=\linewidth]{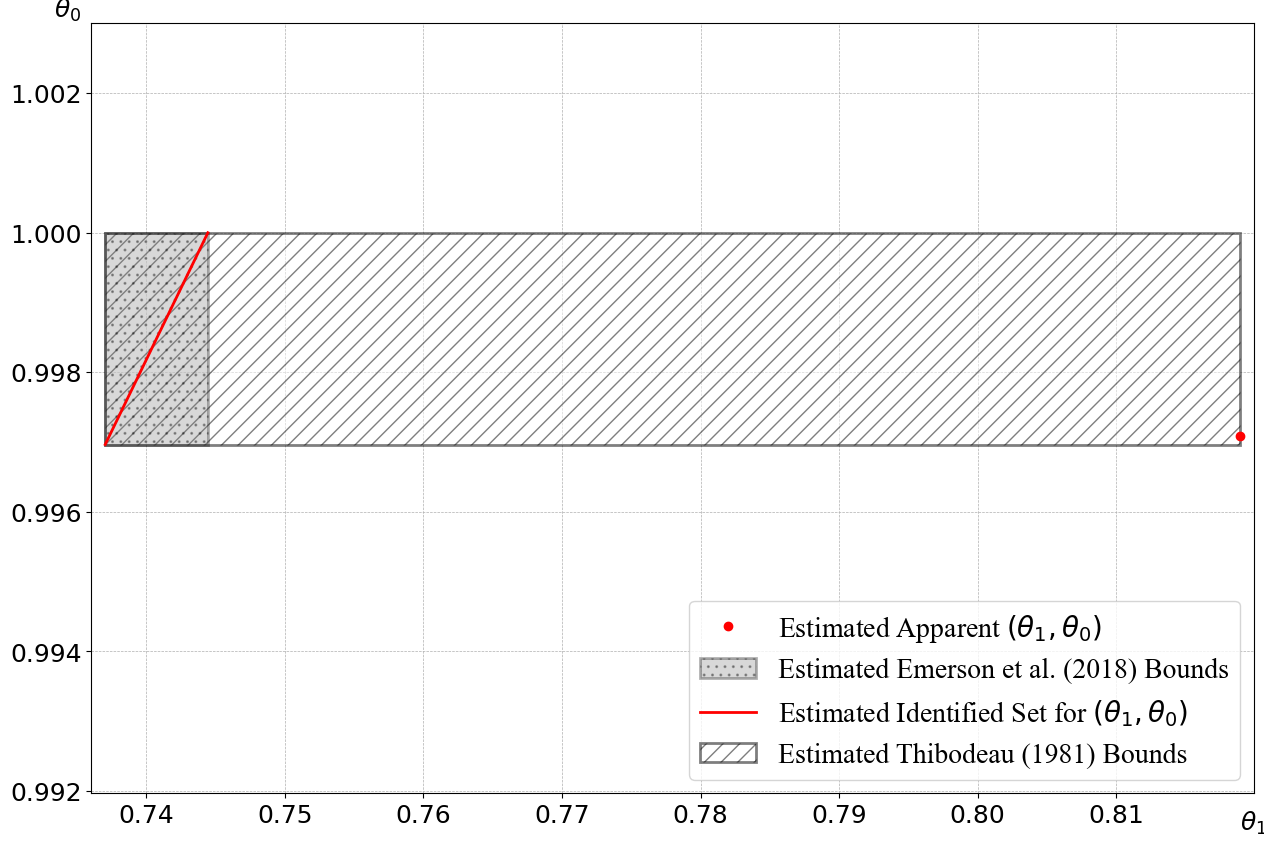} 
        \caption{\citet{shah2021performance} Sx: Methods Comparison} \label{figure_shah_sx:comparison}
    \end{subfigure}
     \hfill
    \begin{subfigure}[t]{0.41\textwidth}
        \centering
        \includegraphics[width=\linewidth]{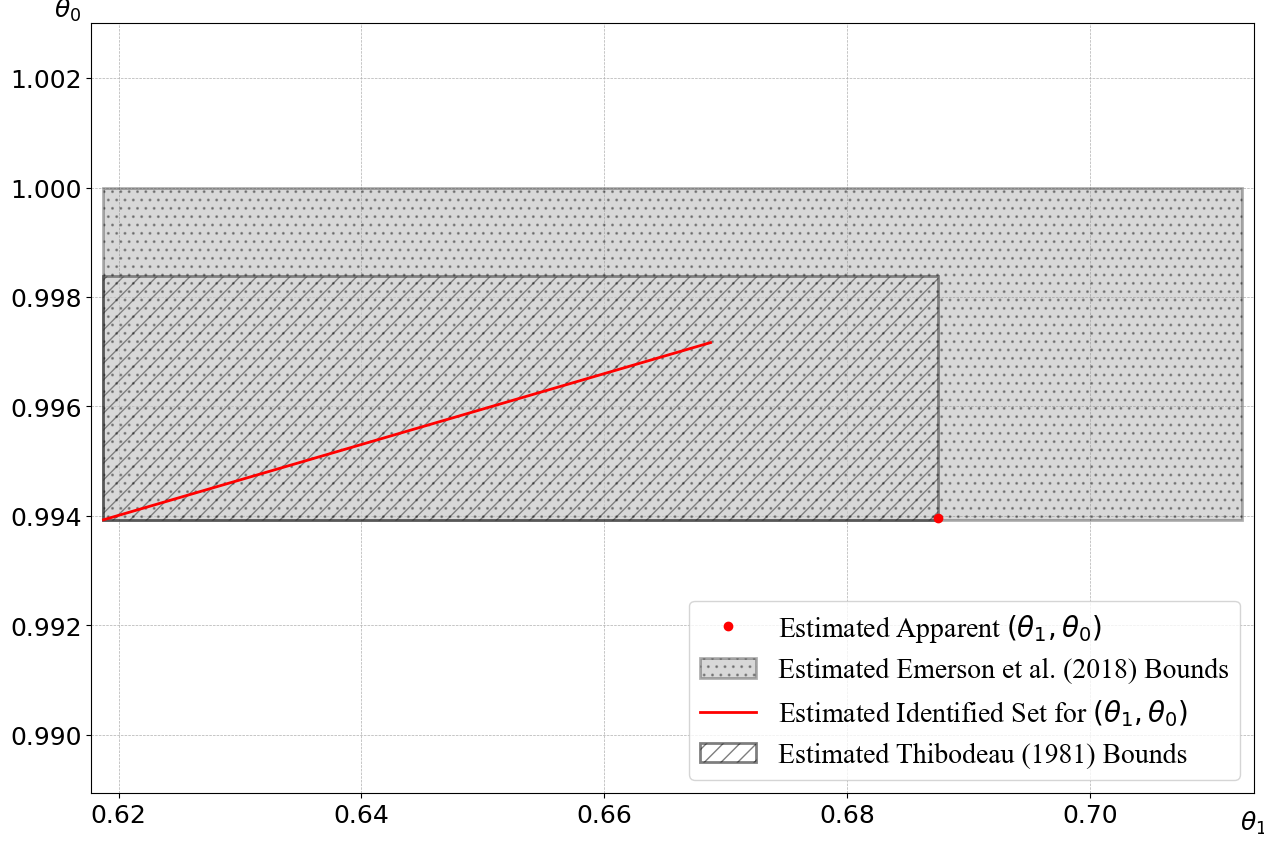} 
        \caption{\citet{shah2021performance} ASx: Methods Comparison} \label{figure_shah_asx:comparison}
    \end{subfigure}
     \begin{subfigure}[t]{0.41\textwidth}
        \centering
        \includegraphics[width=\linewidth]{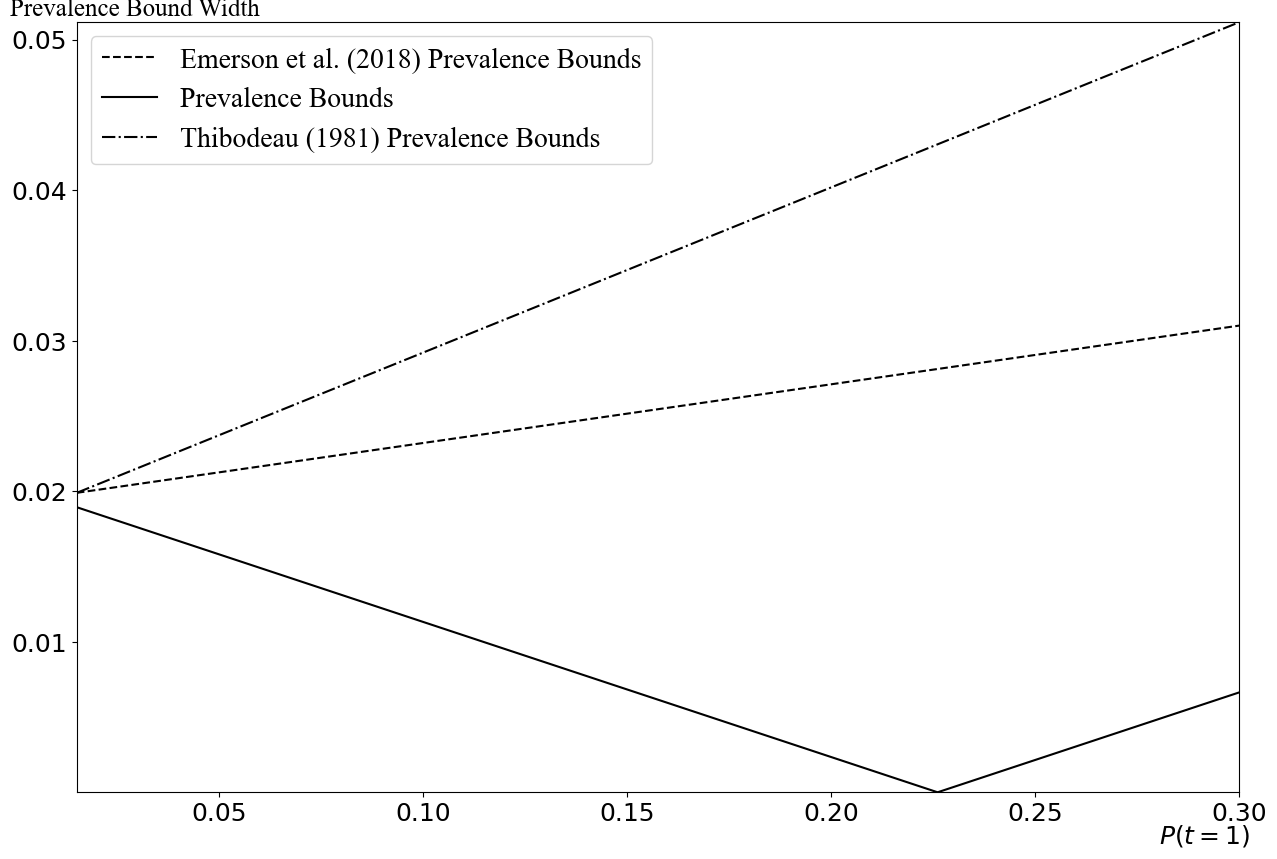} 
        \caption{EUA: Prevalence Bounds Width} \label{figure_EUA:prevalence}
    \end{subfigure}
    \hfill
    \begin{subfigure}[t]{0.41\textwidth}
        \centering
        \includegraphics[width=\linewidth]{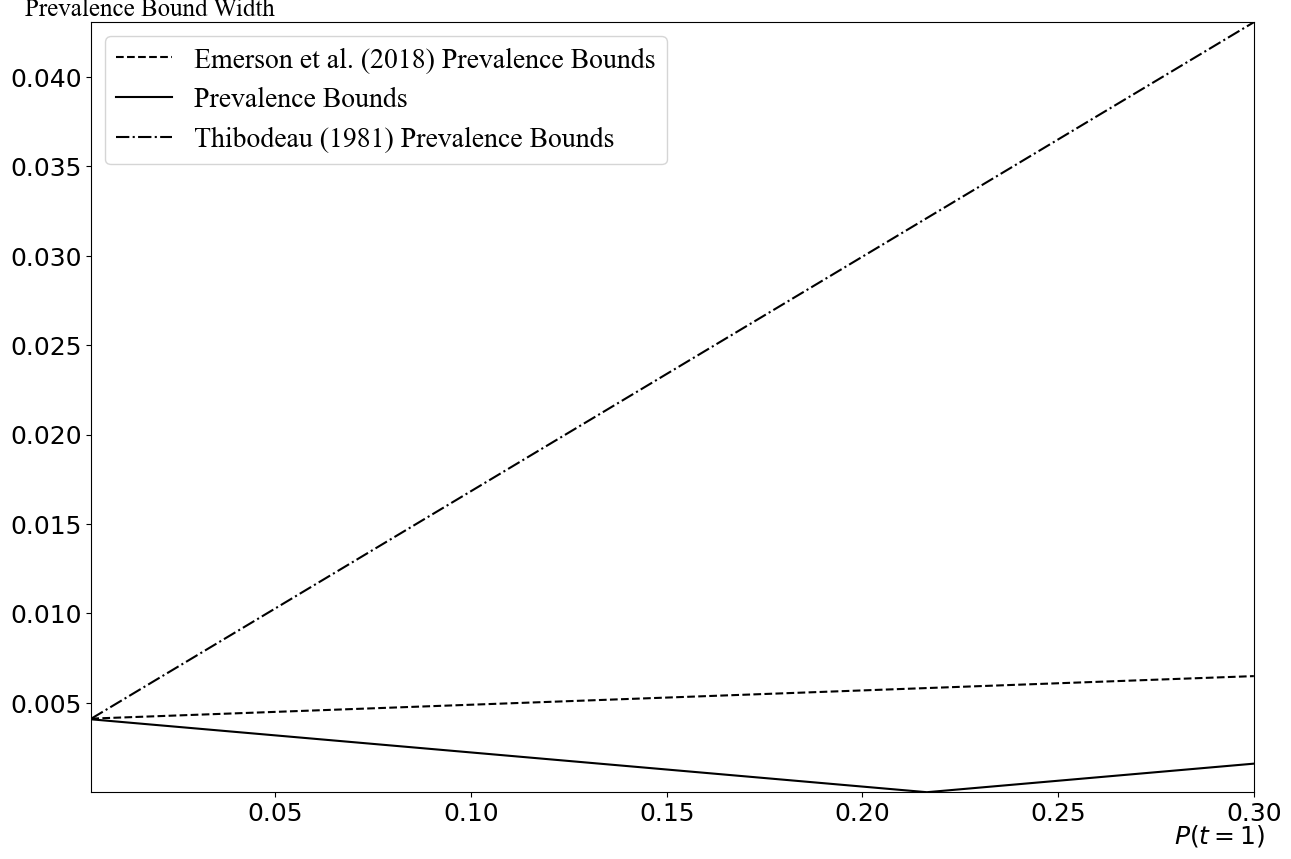} 
        \caption{\citet{shah2021performance} Sx: Prevalence Bounds Width} \label{figure_shah_sx:prevalence}
    \end{subfigure}
     \hfill
    \begin{subfigure}[t]{0.41\textwidth}
        \centering
        \includegraphics[width=\linewidth]{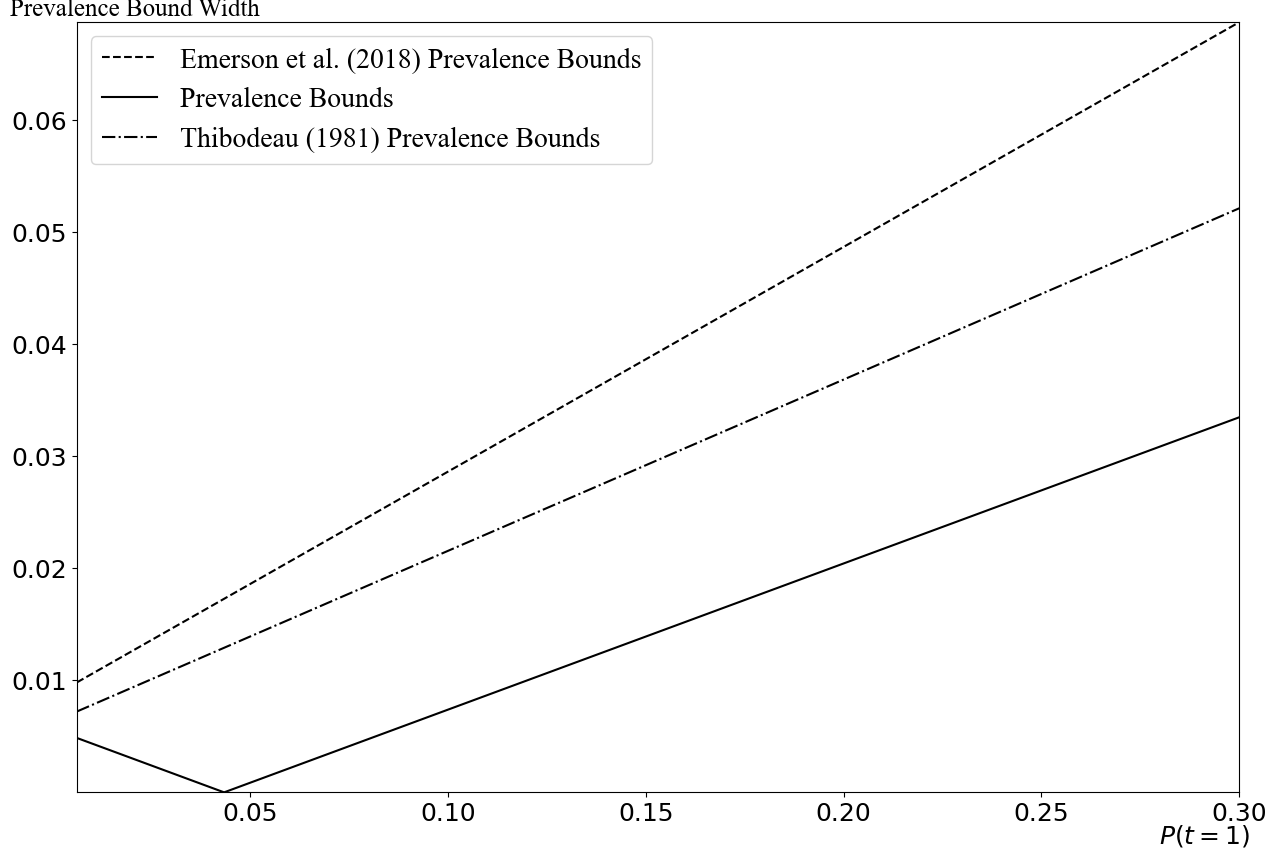} 
        \caption{\citet{shah2021performance} ASx: Prevalence Bounds Width} \label{figure_shah_asx:prevalence}
    \end{subfigure}
    \caption{Assumed $\mathcal{S}=\{(0.9,1)\}$. Sx denotes symptomatic, and ASx asymptomatic individuals. First row depicts estimates, and $95\%$ confidence sets for ``apparent'' measures and points in the identified set for $(\theta_1,\theta_0)$. Second row compares estimated identified sets with estimates by comparable methods. Bottom row compares widths of prevalence bounds implied by the estimates.}
    \label{joint_figure}
\end{figure}
\end{landscape}
\restoregeometry

The original EUA was granted based on interim results of the study in which the test exhibited estimated ``apparent'' sensitivity and specificity of $(91.7\%,100\%)$. Subsequent results of the full study yielded ``apparent'' operating characteristics estimates of $(84.6\%,98.5\%)$. Public statements and media releases erroneously cite all of the estimates as estimates of true performance.\footnote{For example: \href{https://www.bloomberg.com/press-releases/2020-12-16/abbott-s-binaxnow-covid-19-rapid-test-receives-fda-emergency-use-authorization-for-first-virtually-guided-at-home-rapid-test-u}{https://www.bloomberg.com/press-releases/2020-12-16/abbott-s-binaxnow-covid-19-rapid-test-receives-fda-emergency-use-authorization-for-first-virtually-guided-at-home-rapid-test-u}.} Both the interim and final estimates are reported on the instructions-for-use pamphlet accompanying the test. First two rows of \Cref{tab:projected_bounds} show that both estimates of ``apparent'' sensitivity lie strictly above the estimated projected upper bound for true sensitivity in all samples. Hence, sensitivity may be overstated by the ``apparent'' analog, as \citet{fitzpatrick2021buyer} suggest. The fifth and sixth rows demonstrate that final estimates of ``apparent'' specificity are at the estimated projected lower bounds for true specificity. True specificity may be understated.

In \Cref{joint_figure}: (\subref{figure_EUA:conf_set}), the estimate of ``apparent'' measures is outside the confidence set for $(\theta_1,\theta_0)$. At the $5\%$ significance level the hypothesis $H_0: (\theta_1,\theta_0)=(84.6\%,98.5\%)$ would be rejected. In other words, under the assumptions, the true sensitivity and specificity are not jointly equal to currently often cited ``apparent'' values $(84.6\%,98.5\%)$ at the ubiquitous level of significance. The argument for the same value holds in all other samples, as well as for the interim ``apparent'' estimates $(91.7\%,100\%)$.

Estimated bounds on false negative rate for symptomatic individuals within 7 days of symptom onset are $[20\%,23.9\%]$ in the final EUA study data, which is between $1.3$ and $1.55$ times larger than the corresponding ``apparent'' estimate of $15.4\%$. Comparison with the often-cited interim estimate of $8.3\%$ reveals that the estimated true false negative rate is between $2.41$ and $2.88$ times larger than the ``apparent'' analog. Data from \citet{shah2021performance} yield estimated bounds of $[25.6\%,26.3\%]$ for symptomatic and $[33.1\%,38.1\%]$ for asymptomatic individuals. These estimates suggest that the true false negative rates may be up to $3.17$ and $4.59$ times higher than ``apparent'' interim analogs for symptomatic and asymptomatic individuals, respectively. Appendix \Cref{sect_sensitivity_analysis} shows that assuming lower $s_1$ further exacerbates the difference between true and ``apparent'' false negative rates. It also notes that assuming any $0.9<s_1<1$ lessens the difference. However, the apparent measure is never contained by the set. Hence for any feasible $s_1$, apparent sensitivity overestimates true sensitivity.
\begin{remark}
 Estimated average number of infected symptomatic people who are missed by the antigen test is up to $3.17$ times higher than the test users may be led to believe by reported ``apparent'' estimates. 
\end{remark}

\citet{hadgu1999discrepant} highlights that the errors in measurement of $2.9$ percentage points for sensitivity are significant. The differences I find in this paper between the estimates of ``apparent'' and true sensitivity are substantially larger under plausible assumptions. The differences vary between $4.6$ and $8.5$ percentage points using the final EUA study data. Results from \citet{shah2021performance} exacerbate the discrepancies when compared to the final EUA study ``apparent'' sensitivity to as much $10.9$ percentage points in the symptomatic population and $22.7$ percentage points in the asymptomatic population. Even though the estimates of specificity remain close to the estimates of ``apparent'' specificity, the findings for sensitivity warrant further attention.

\begin{remark}
FDA has granted EUA to tests demonstrating at least $80\%$ estimated sensitivity. The results show that, depending on interpretation and assumed $(s_1,s_0)$, the test may not satisfy the requirement.
\end{remark}

Panels (\subref{figure_EUA:comparison}), (\subref{figure_shah_sx:comparison}), and (\subref{figure_shah_asx:comparison}) of \Cref{joint_figure} show estimates of the identified set for $(\theta_1,\theta_0)$ and compare them with estimates obtained using comparable methods. Results are represented graphically in order preserve the specific linear structure of the identified set that is lost through projection. The sharp identified set provides a substantial reduction in size in all three samples over the comparable methods, and can be very informative. Estimates of the identified set do not contain the estimates of ``apparent'' measures in any of the samples. Owing to the lack of sharpness, bounds estimated using other methods do not necessarily exclude the ``apparent'' measures. \Cref{tab:projected_bounds} shows that projected bounds on $\theta_1$ and $\theta_0$ can also be proper subsets of those produced by comparable methods. 

 \begin{remark}\label{remark_no_power_application}
     Projected bounds on $\theta_j$ from \citet{emerson2018biomarker} are equivalent to projected bounds from \Cref{bounds_theta_prop} without imposing the tendency to wrongly agree for any $y$. Rows of \Cref{tab:projected_bounds} marked by ``\textit{Projection}'' and ``\citet{emerson2018biomarker}'' thus correspond to projection bounds with and without assuming the tendency to wrongly agree for $y=1$, respectively. As mentioned by \Cref{remark_identifying_power}, the assumption may have identifying power depending on $P(t,r)$ and $(s_1,s_0)$. For data in the first two columns, estimates suggest that the assumption has no identifying power. However, among \citet{shah2021performance} ASx individuals, it effectively halves the size of the estimated identified set. 
 \end{remark}

Panels (\subref{figure_EUA:prevalence}), (\subref{figure_shah_sx:prevalence}), and (\subref{figure_shah_asx:prevalence}) of \Cref{joint_figure} depict the width of prevalence bounds implied by estimates from the three methods when extrapolated to populations screened by the antigen test. The solid line represents bound width given estimates of the identified set for $(\theta_1,\theta_0)$ and \eqref{prevalence_bounds} for various hypothetical values of $Q(t=1)$. As previously highlighted, for $Q(t=1)=P(t=1)$ prevalence becomes point-identified, despite $(\theta_1,\theta_0)$ being only partially identified. The remaining lines refer to widths of bounds in \eqref{prevalence_bounds_no_dependence}, following from estimates obtained by comparable methods which yield rectangular bounds on $(\theta_1,\theta_0)$. The resulting sharp bounds on prevalence are always proper subsets of bounds found via the other two methods. 
Benefits stemming from the particular shape of the sharp joint identified set are immediate. Even when the projected bounds are not strictly narrower compared to other methods, the identified set can yield substantially narrower bounds on derived parameters, as shown by \Cref{joint_figure}:(\subref{figure_EUA:prevalence}), (\subref{figure_shah_sx:prevalence}). 

\section{Applications Beyond Diagnostic Test Performance}\label{sect_broader_applications}

Derived results have applications that extend beyond diagnostic test performance studies. This section offers three illustrative examples, highlighting further utility of the bounds on $(\theta_1,\theta_0)$ and $Q(y=1)$. It also interprets the tendency to wrongly agree in the relevant contexts, and contrasts it with the exclusion restrictions $E[t|r,y]=E[t|r]$ and $E[t|r,y]=E[t|y]$ from \citet{cross2002regressions}.\footnote{Note that $E[t|r,y]=E[t|y]$ is equivalent to $t\independent r|y$, which is frequently considered implausible in the context of diagnostic tests, as discussed by \Cref{sect_assumptions}.}

In abstraction, suppose $P(t,r)$ is identified and $P(r|y)$ can be identified or credibly bounded, for $(t,r,y)\in\{0,1\}^3$. A salient case is the one in which a validation dataset identifying $P(r|y)$ exists, but it cannot be matched with the dataset identifying $P(t,r)$. This may happen due to legal or privacy concerns, lack of adequate identifiers, or because the two datasets are independent. Sharp bounds on $P(t,y)$ and its features follow from the results. If $(\theta_1,\theta_0)$ are bounded and they extrapolate to other populations where data only on $t$ is available, one can also sharply bound $Q(y=1)$.

\begin{example}{(Surveys and Validation Data)}
Suppose $y$ is measured using a survey response $r$, and $t$ is a binary outcome of interest for surveyed individuals. This identifies $P(t,r)$. It is well known that survey responses are susceptible to misclassification errors. This has been discussed in the contexts where $y$ is participation in government welfare programs, disability, or employment status.(\citet{poterba1986reporting}, \citet{kreider2007disability}) For example, let $y$ be true participation in the Food Stamp program, $r$ self-reported participation, and let $t$ denote whether a person has completed higher education. \citet{bollinger1997modeling} and \citet{meyer2022errors} identify $P(r|y)$ using administrative validation data for the American Community Survey, the Current Population Survey and the Survey of Income and Program Participation. One can then bound functions of $E[t|y]$ in these surveys.

Researchers may find the tendency to wrongly agree for $y=j$ credible, maintaining it to further tighten the bounds. For $y=1$ in the example above, it would mean that people who falsely report no participation are more likely not to have a higher education degree than to have one. The restriction $E[t|r,y]=E[t|r]$ holds if rates of higher education depend on program participation only through the survey response. That is, for people who provide a response $r$, higher education rates must not change with true program participation. Conversely, $E[t|r,y]=E[t|y]$ holds if for individuals who have true participation $y$, higher education rates must be the same among those who gave correct and incorrect answers. Hence, one may consider the tendency to wrongly agree to be more plausible in this context.

Results in \Cref{sect_bounding_prevalence} could be useful if one wishes to learn $Q(y=1)$ in a different population to which $E[t|y]$ extrapolates, and for which only $t$ is available. In the context of the example, if access to a dataset containing local rates of higher education is available, then local Food Stamp program participation rates can be sharply bounded using knowledge of $E[t|y]$ from the survey.
\end{example}

\begin{example}{(Protected Classes and Privacy)}
Let $t$ be an outcome of interest and $y$ a protected class. Administrative data often do not contain $y$, but its proxy $r$ may be available, identifying $P(t,r)$. Some commonly used proxies are constructed based on datasets in which both $r$ and $y$ are observed, so their performance may be known. For a $t$ of interest, we can then sharply bound various parameters of $P(t,y)$. For example, $y$ may be a latent binary indicator for a certain race, and its proxy $r$ may be constructed using the Bayesian Improved Surname Geocoding (BISG) method. Performance of BISG has been validated in several datasets, potentially providing information on $P(r|y)$. (\citet{elliott2008new}, \citet{imai2016improving}) \citet{elzayn2023measuring} consider a similar setting where $t$ is a tax audit flag, and seek to identify racial tax audit disparity $E[t|y=1]-E[t|y=0]=\theta_1+\theta_0-1$ in a dataset where only $r$ is available.\footnote{The parameter $\theta_1+\theta_0-1$ is known as the Youden's J statistic in the medical literature.} This parameter can be bounded using $P(r|y)$. One could also use results from \Cref{sect_bounding_prevalence} to sharply bound racial composition $Q(y=1)$ in a population with the same $E[t|y]$ if a dataset containing only tax audit rates is available.

The tendency to wrongly agree for $y=0$ would mean that people misclassified as being of race $r=1$ are more likely to be audited than not to be audited. This would be plausible if characteristics of individuals with $y=0$ that lead to racial misclassification also make audit the more likely outcome. The assumption $E[t|r,y]=E[t|r]$ maintains that tax audit rates for people classified as $r$ would not vary with their true race. This would hold if audit decisions depend on race only through information summarized by $r$. Conversely, $E[t|r,y]=E[t|y]$ would hold if for people of race $y$ audit rates do not vary with their classified race $r$. This would be the case if audit decisions depend on information summarized by $r$ only through true race. If any of these assumptions are plausible, researchers may use them to obtain tighter bounds.

The same arguments apply to measurement of racial disparities in health care, where $t$ can be, for example, medication nonadherence. \citet{weissman2011advancing} explain that race is often missing from medical claims data and studies validating performance of BISG in such datasets are available. (\citet{adjaye2014using})
\end{example}

\begin{example}{(Binary Classifiers)}
Let $t$ be a binary classifier whose performance is determined using an imperfect binary classifier or label $r$ as a reference. The discussion herein readily applies to any such setting. In general, $r$ may be imperfect when determined by labelers or algorithms. (\citet{cannings2020classification}) For example, labels $r$ are often obtained through services like the Amazon Mechanical Turk. Mislabeling may happen due to human error, inattentive labelers, or malicious mislabelling activity. If researchers are be able to determine misclassification rates $(s_1,s_0)$, then $(\theta_1,\theta_0)$ may be bounded. \citet{foody2010assessing} notes that $(\theta_1,\theta_0)$ are often of interest in the context of remote sensing applications, such as satellite imaging. However, reference data $r$ are commonly imperfect. \citet{carlotto2009effect} explains that in some cases, it may be possible to learn $(s_1,s_0)$ by observing the ground truth $y$ in validation studies. When the performance of $r$ cannot be validated, one can use the bounds to perform sensitivity analyses and determine how much apparent and real performance may differ for various possible values of $(s_1,s_0)$. 

Binary diagnostic tests are specific examples of binary classifiers, and the interpretation of the tendency to wrongly agree remains unchanged for general $t$ and $r$. Its plausibility can be argued based on the properties of the classifiers. \citet{foody2023challenges} notes that the exclusion restriction $E[t|r,y]=E[t|y]$ may also commonly be implausible in the context of binary classifiers, especially when $t$ and $r$ are based on the same phenomenon or process. The restriction $E[t|r,y]=E[t|r]$ asserts that $t\independent y|r$, or that $t$ cannot provide additional information about $y$ over $r$, which may be unappealing depending on the context.

Results from \Cref{sect_bounding_prevalence} also apply directly. When using the classifier $t$ to determine the prevalence of $y$ in a screening study, one can use them to obtain sharp bounds on $Q(y=1)$ if the performance of $t$ extrapolates to the screened population.
\end{example}

\section{Concluding Remarks}\label{conclusion_sect}

This paper derives the smallest possible identified set for sensitivity and specificity of a diagnostic test of interest in standard settings, when the reference test is imperfect. It formalizes an existing assumption on dependence between the reference and the test of interest, and shows how it can further reduce the size of the identified set. Finally, it develops an appropriate uniform inference procedure for the points in the identified set, enabling construction of confidence sets. The study also indicates applicability of the method beyond the context of diagnostic test performance studies.

The framework is proposed as a solution to a ubiquitous problem in diagnostic test performance studies, and it can be directly applied to existing study data to bound true test performance. Doing so demonstrates that a widely used COVID-19 antigen test tends to produce significantly more false negative results than what the currently cited figures suggest. Since other rapid COVID-19 antigen tests may exhibit similar tendencies, these findings warrant further investigation.

\clearpage
\begin{spacing}{1.2}
\printbibliography 
\end{spacing}

\newpage
\begin{appendices}

\section{Bounding Predictive Values}\label{sect_bounding_pv}
\renewcommand\theequation{\thesection.\arabic{equation}}
\setcounter{equation}{0} 

Positive predictive value (PPV) is the probability that a patient is diseased conditional on receiving a positive test result. Negative predictive value (NPV) is the probability that a patient who has tested negative is truly healthy. Clinicians are usually more concerned with knowing predictive values of a test $t$ than its sensitivity and specificity. 

The probability of the patient being diseased prior to observing a test result is referred to as a pre-test probability. For a known pre-test probability, sensitivity and specificity are often extrapolated from test performance studies to find predictive values using Bayes' theorem. With this in mind, let $Q(t,y)$ be the distribution of the clinical population of interest and suppose $(\theta_1,\theta_0)$ extrapolate to this population in the sense of Assumption \ref{ass:extrapolation}. As in \Cref{sect_bounding_prevalence}, we use $Q$ to emphasize that test performance is extrapolated and that $Q(y=1)$ may differ from the prevalence in the performance study population. Clinicians settle on a pre-test probability $\pi = Q(y=1)$ using the knowledge of local rates of infection and patients’ symptoms and characteristics. (\citet{watson2020interpreting})

\citet{manski2020bounding} provides bounds on predictive values for COVID-19 antibody tests using point identified values of $\theta_1$ and $\theta_0$, when the pre-test probability $\pi$ is bounded. The author notes that the analysis can be generalized to take bounds rather than exact values of $\theta_1$ and $\theta_0$ as inputs. \citet{ziegler2021binary} extends the analysis of predictive values when $\theta_1$ and $\theta_0$ are partially identified due to an imperfect reference test, assuming that $s_0=1$. The bounds below do not require that $s_0=1$ in the performance study.

Predictive values are defined as:
\begin{align}\label{predictive_value_definition}
\begin{split}
    &PPV = Q(y=1|t=1) = \frac{\theta_1\pi }{\theta_1\pi + (1-\theta_0)(1-\pi )}\\
    &NPV = Q(y=0|t=0)  = \frac{\theta_0 (1-\pi )}{\theta_0 (1-\pi )+(1-\theta_1)\pi }  .
\end{split}
\end{align}

Assume that the sharp identification region $\mathcal{G}_{(\theta_1,\theta_0)}(s_1,s_0)$ for $(\theta_1,\theta_0)$ and the pre-test probability of the clinician $\pi$ are known. From \eqref{predictive_value_definition}, it can be seen that both PPV and NPV increase with $\theta_1$ and $\theta_0$. Thus, the sharp bounds are:
\begin{align}\label{predictive_value_bounds_known_pre_test}
\begin{split}
    &PPV \in \left[\frac{\theta_1^L\pi}{\theta_1^L\pi+ (1-\theta_0^L)(1-\pi)},\frac{\theta_1^U\pi}{\theta_1^U\pi+ (1-\theta_0^U)(1-\pi)}\right]\\
    &NPV \in \left[\frac{\theta_0^L(1-\pi)}{\theta_0^L(1-\pi)+ (1-\theta_1^L)\pi},\frac{\theta_0^U(1-\pi)}{\theta_0^U(1-\pi)+ (1-\theta_1^U)\pi}\right].
\end{split}
\end{align}

If the clinician is not willing to settle on a single value of $\pi$, rather on a range of values $\pi\in [\pi_L,\pi_H]$, the bounds are simply:
\begin{align}\label{predictive_value_bounds_bounded_pre_test}
\begin{split}
    &PPV \in \left[\frac{\theta_1^L\pi_L}{\theta_1^L\pi_L+ (1-\theta_0^L)(1-\pi_L)},\frac{\theta_1^U\pi_H}{\theta_1^U\pi_H+ (1-\theta_0^U)(1-\pi_H)}\right]\\
    &NPV \in \left[\frac{\theta_0^L\pi_H}{\theta_0^L\pi_H+ (1-\theta_1^L)(1-\pi_H)},\frac{\theta_0^U\pi_L}{\theta_0^U\pi_L+ (1-\theta_1^U)(1-\pi_L)}\right].
\end{split}
\end{align}

The bounds are generalizable analogously to the previously outlined case for bounding prevalence when the identification region $\mathcal{G}_{(\theta_1,\theta_0)}(s_1,s_0)$ is expanded to $\mathcal{G}_{(\theta_1,\theta_0)}(\mathcal{S)}$:

\begin{align}\label{predictive_value_bounds_bounded_pre_test_unknown_reference}
\begin{split}
    &PPV \in \left[\min_{(\theta_1,\theta_0)\in\mathcal{G}_{(\theta_1,\theta_0)}(\mathcal{S)}}\left\{\frac{\theta_1\pi_L}{\theta_1\pi_L+ (1-\theta_0)(1-\pi_L)}\right\},\max_{(\theta_1,\theta_0)\in\mathcal{G}_{(\theta_1,\theta_0)}(\mathcal{S)}}\left\{\frac{\theta_1\pi_H}{\theta_1\pi_H+ (1-\theta_0)(1-\pi_H)}\right\}\right]\\
    &NPV \in \left[\min_{(\theta_1,\theta_0)\in\mathcal{G}_{(\theta_1,\theta_0)}(\mathcal{S)}}\left\{\frac{\theta_0\pi_H}{\theta_0\pi_H+ (1-\theta_1)(1-\pi_H)}\right\},\max_{(\theta_1,\theta_0)\in\mathcal{G}_{(\theta_1,\theta_0)}(\mathcal{S)}}\left\{\frac{\theta_0\pi_L}{\theta_0\pi_L+ (1-\theta_1)(1-\pi_L)}\right\}\right].
\end{split}
\end{align}

\section{Additional Moment Functions}\label{sect_additional_mom_inequalities}

\renewcommand\theequation{\thesection.\arabic{equation}}
\setcounter{equation}{0} 

This section defines moment functions for remaining identified sets in Propositions \ref{bounds_theta_prop} and \ref{bounds_wrongly_agree_prop} when the tests have a tendency to wrongly agree only for $y=0$, and for both $y=1$ and $y=0$. All proofs are collected in Appendix \ref{sect_proofs}.

Focus first on the bounds on $\theta_1$ from \Cref{bounds_theta_prop}. Following the reasoning in \Cref{finite_sample_sect}, we decompose the bounds on $\theta_1$ to construct the appropriate moment inequalities. Note that there are four values determined by the population parameters that are all lower bounds, and four values that are all upper bounds on $\theta_1$ given $(s_1,s_0)$. One lower and one upper bound are trivial since they state that $\theta_1\geq 0$ and $\theta_1\leq 1$. Both can be omitted since $\theta_1\in[0,1]$ by definition. We can then represent the bound on $\theta_1$ using six moment inequalities, corresponding to the six non-trivial boundary values of the identified set. One additional moment equality is needed to represent the joint identification region.

\begin{theoremEnd}[restate, proof at the end,
no link to proof]{proposition}\label{prop_original_inequalities}
Let the moment function $m$ be:
\begin{align}
\begin{split}\label{moment_inequalities_original}
    m(W_i,\theta) =\begin{pmatrix} m_1(W_i,\theta)\\
    m_2(W_i,\theta)\\
    m_3(W_i,\theta)\\
    m_4(W_i,\theta)\\
    m_5(W_i,\theta)\\
    m_6(W_i,\theta)\\
    m_7(W_i,\theta)\\
    \end{pmatrix} = 
    \begin{pmatrix} (-\theta_1+s_1)\frac{r_i-1+s_0}{s_1-1+s_0}+(t_i-1)r_i\\
     (-\theta_1+1-s_1)\frac{r_i-1+s_0}{s_1-1+s_0}+(r_i-1)(1-t_i)\\
    (-\theta_1+1)\frac{r_i-1+s_0}{s_1-1+s_0}+(t_i-1)\\
    \theta_1\frac{r_i-1+s_0}{s_1-1+s_0}-t_i\\
   (\theta_1-s_1)\frac{r_i-1+s_0}{s_1-1+s_0}-t_i(1-r_i)\\
    (\theta_1-1+s_1)\frac{r_i-1+s_0}{s_1-1+s_0}-t_ir_i\\
    (\theta_0-1)(1-\frac{r_i-1+s_0}{s_1-1+s_0})-\theta_1\frac{r_i-1+s_0}{s_1-1+s_0}+t_i
    \end{pmatrix}.
\end{split}
\end{align}
 Joint identification region $\mathbf{\Theta}(P)= \bigcup_{(s_1,s_0)\in\mathcal{S}}\Big(\mathcal{H}_{(\theta_1,\theta_0)}(s_1,s_0)\times\{(s_1,s_0)\}\Big)$ with $\mathcal{H}_{(\theta_1,\theta_0)}(s_1,s_0)$ defined in Proposition \ref{bounds_theta_prop} is represented by the moment function $m$. For each $\theta\in[0,1]^2\times\mathcal{S}$ such that $E_P\big(m_j(W_i,\theta)\big)\leq0\textrm{ for $j=1,\hdots,6$ and } E_P\big(m_7(W_i,\theta)\big)=0$, it must be that $\theta\in\mathbf{\Theta}(P)$. Conversely, if $\theta\in\mathbf{\Theta}(P)$, then $E_P\big(m_j(W_i,\theta)\big)\leq0\textrm{ for $j=1,\hdots,6$ and } E_P\big(m_7(W_i,\theta)\big)=0$. 
\end{theoremEnd}
\begin{proofEnd}
I prove this by finding $E_P\big(m_j(W_i,\theta)\big)$ for $j=1,2\hdots,7$ and demonstrating that the resulting system is equivalent to the bounds defined in \Cref{bounds_theta_prop} extended to $\mathbf{\Theta}(P)= \bigcup_{(s_1,s_0)\in\mathcal{S}}\Big(\mathcal{H}_{(\theta_1,\theta_0)}(s_1,s_0)\times\{(s_1,s_0)\}\Big)$. Suppose that $E_P\big(m_j(W_i,\theta)\big)\leq 0$ for $j=1,2\hdots,6$ and $E_P\big(m_7(W_i,\theta)\big)= 0$. From \eqref{moment_inequalities_original}:
\begin{align}
\begin{split}\label{expectations_moment_inequalities_original_lower}
    E_P\big(m_1(W_i,\theta)\big) &= -\theta_1 P_{s_1,s_0}(y=1)+P_{s_1,s_0}(r=1,y=1)-P(r=1)+P(t=1,r=1)\\
    &=P(t=1,r=1)-P_{s_1,s_0}(r=1,y=1)-\theta_1 P_{s_1,s_0}(y=1)\leq 0\\
    E_P\big(m_2(W_i,\theta)\big) &= (-\theta_1+1-s_1)P_{s_1,s_0}(y=1)-P(t=0,r=0)\\
        &=P_{s_1,s_0}(r=0,y=1)-P(t=0,r=0)-\theta_1 P_{s_1,s_0}(y=1)\leq 0\\
    E_P\big(m_3(W_i,\theta)\big) &= (-\theta_1+1)P_{s_1,s_0}(y=1)+P(t=1)-1\\
            &=P_{s_1,s_0}(y=1)s_1-P(r=1)+P_{s_1,s_0}(y=1)(1-s_1)+P(t=1)\\
            &-P(r=0)-\theta_1 P_{s_1,s_0}(y=1)\\
            &=P(t=1,r=1)-P_{s_1,s_0}(r=1,y=0)+P_{s_1,s_0}(r=0,y=1)-P(t=0,r=0)\\
            &-\theta_1 P_{s_1,s_0}(y=1)\leq 0.
\end{split}
\end{align}
Note further that if $\theta_1\in [0,1]$, which is true by definition, the three inequalities above yield the lower bound from \Cref{bounds_theta_prop} for $\theta_1\in\mathcal{H}_{\theta_1}(s_1,s_0)$ given an arbitrary $(s_1,s_0)\in\mathcal{S}$:
\begin{align}
    \begin{split}
    \theta_1P_{s_1,s_0}(y=1) &\geq max\Big(0,P(t=1,r=1)-P_{s_1,s_0}(r=1,y=0)\Big)\\
    &+max\Big(0,P_{s_1,s_0}(r=0,y=1)-P(t=0,r=0)\Big).\\
    \end{split}
\end{align}
This is equivalent to the lower bound for the element $\theta_1$ of $(\theta_1,\theta_0,s_1,s_0)\in\mathbf{\Theta}(P)$. Consider next:
\begin{align}
\begin{split}\label{expectations_moment_inequalities_original_upper}
    E_P\big(m_4(W_i,\theta)\big) &= \theta_1 P_{s_1,s_0}(y=1)-P(t=1)\\
    &= \theta_1 P_{s_1,s_0}(y=1)-P(t=1,r=0)-P(t=1,r=1)\leq 0\\
    E_P\big(m_5(W_i,\theta)\big) &= (\theta_1-s_1)P_{s_1,s_0}(y=1)-P(t=1,r=0)\\
        &= \theta_1 P_{s_1,s_0}(y=1)-P(t=1,r=0)-P_{s_1,s_0}(r=1,y=1)\leq 0\\
    E_P\big(m_6(W_i,\theta)\big) &= (\theta_1-1+s_1)P_{s_1,s_0}(y=1)-P(t=1,r=1)\\
    &= \theta_1 P_{s_1,s_0}(y=1) -P_{s_1,s_0}(r=0,y=1)-P(t=1,r=1) \leq 0\\
\end{split}
\end{align}
Similarly, the upper bound from \Cref{bounds_theta_prop} is obtained for the element $\theta_1$ of $(\theta_1,\theta_0,s_1,s_0)\in\mathbf{\Theta}(P)$:
\begin{align}
    \begin{split}
    \theta_1P_{s_1,s_0}(y=1) &\leq min\Big(P(t = 1, r=1-1), P_{s_1,s_0}(r=1-1,y=1)\Big)\\
    &+min\Big(P(t = 1, r=1), P_{s_1,s_0}(r=1,y=1)\Big).        
    \end{split}
\end{align}
Taking the expected value of the final component of the moment function yields:
\begin{align}
\begin{split}\label{expectations_moment_equalities}
    E_P\big(m_7(W_i,\theta)\big) &= (\theta_0-1)(1-P_{s_1,s_0}(y=1))-\theta_1 P_{s_1,s_0}(y=1)+P(t=1)= 0
\end{split}
\end{align}
It is then is true that $\theta_0 P_{s_1,s_0}(y=0) = P_{s_1,s_0}(y=0)+\theta_1P_{s_1,s_0}(y=1)-P(t=1)$. This is the linear relationship between $(\theta_1,\theta_0)$ in the identified set from \Cref{bounds_theta_prop}. Going in the other direction, it is immediate that if the two bounds and the linear relationship hold so that $\theta\in\mathbf{\Theta}(P)$, then $E_P\big(m_j(W_i,\theta)\big)\leq 0$ for $j=1,2\hdots,6$ and $E_P\big(m_7(W_i,\theta)\big)= 0$, demonstrating that the expected values of moment functions represent the joint identification region $\theta\in\mathbf{\Theta}(P)$.
\end{proofEnd}

The same reasoning applies for other bounds. Assume that the index and reference tests have a tendency to wrongly agree only for $y=0$. As in the case when the tests have a tendency to wrongly agree only for $y=1$, the three non-trivial lower-bound values are identical to the ones when there is no tendency to wrongly agree for any $y$. There are four cases for the upper bound, one of which is:
\begin{align}\label{restrict_theta_space_wrongly_agree_0}
\begin{split}
     &\theta_0\leq \left(\frac{P_{s_1,s_0}(r=1,y=0)}{2}+P_{s_1,s_0}(r=0,y=0)\right)\frac{1}{P_{s_1,s_0}(y=0)} = \frac{1+s_0}{2}
\end{split}
\end{align}
Again, this is a restriction on the parameter space, since it states only that $\theta_0\in[0,\frac{1+s_0}{2}]$. The relevant parameter space for $\theta$ when the two tests have a tendency to wrongly agree for $y=0$ is $\theta\in\bigcup_{(s_1,s_0)\in\mathcal{S}}[0,1]\times[0,\frac{1+s_0}{2}]\times\{(s_1,s_0)\}$. The restriction allows $\theta_0>s_0$, but not by more than $\frac{1-s_0}{2}$.

\begin{remark}
If the index and reference tests have a tendency to wrongly agree only for $y=0$, then the function $\Bar{m}^0$ defining moment inequalities that represent the corresponding identified set for $\theta\in\bigcup_{(s_1,s_0)\in\mathcal{S}}[0,1]\times[0,\frac{1+s_0}{2}]\times\{(s_1,s_0)\}$ would be:
\begin{align}
\begin{split}\label{moment_inequalities_wrongly_agree_0}
    \Bar{m}^0(W_i,\theta) =\begin{pmatrix} \Bar{m}^0_1(W_i,\theta)\\
    \Bar{m}^0_2(W_i,\theta)\\
    \Bar{m}^0_3(W_i,\theta)\\
    \Bar{m}^0_4(W_i,\theta)\\
    \Bar{m}^0_5(W_i,\theta)\\
    \Bar{m}^0_6(W_i,\theta)\\
    \Bar{m}^0_7(W_i,\theta)\\
    \end{pmatrix} = 
    \begin{pmatrix} (-\theta_0+s_0)\left(1-\frac{r_i-1+s_0}{s_1-1+s_0}\right)+(r_i-1)t_i\\
     (-\theta_0+1-s_0)\left(1-\frac{r_i-1+s_0}{s_1-1+s_0}\right)-t_ir_i\\
    (-\theta_0+1)\left(1-\frac{r_i-1+s_0}{s_1-1+s_0}\right)-t_i\\
    \theta_0\left(1-\frac{r_i-1+s_0}{s_1-1+s_0}\right)+(t_i-1)\\
   (\theta_0-s_0)\left(1-\frac{r_i-1+s_0}{s_1-1+s_0}\right)-r_i(1-t_i)\\
     \Big(\theta_0+\frac{-1+s_0}{2}\Big)\left(1-\frac{r_i-1+s_0}{s_1-1+s_0}\right)-(1-t_i)(1-r_i)\\
    (\theta_0-1)(1-\frac{r_i-1+s_0}{s_1-1+s_0})-\theta_1\frac{r_i-1+s_0}{s_1-1+s_0}+t_i
    \end{pmatrix}.
\end{split}
\end{align}
The proof is analogous to that of \Cref{prop_inequalities_wrongly_agree}.
\end{remark}

Finally, the same steps yield a moment function that defines the identified set when the tests have a tendency to wrongly agree for both $y=1$ and $y=0$. As in the case where the tendency exists only for $y=1$, the appropriate parameter space is $\theta\in\bigcup_{(s_1,s_0)\in\mathcal{S}}[0,\frac{1+s_1}{2}]\times[0,\frac{1+s_0}{2}]\times\{(s_1,s_0)\}$.

\begin{theoremEnd}[restate, proof at the end, no link to proof]{proposition}\label{prop_inequalities_wrongly_agree_both}
Assume that the index and reference tests have a tendency to wrongly agree for $y=1$ and $y=0$. Let the moment function $\dbar{m}$ be equal to $\Bar{m}^1$ in \eqref{moment_inequalities_wrongly_agree_1} in all components except $\dbar{m}_4(W_i,\theta)$, and $\dbar{m}_6(W_i,\theta)$:
\begin{align}\label{moment_inequalities_wrongly_agree_both}
    \begin{split}
    \dbar{m}_4(W_i,\theta) &= \theta_1\frac{r_i-1+s_0}{s_1-1+s_0}-t_i+\frac{1}{2}\Big(r_i-s_1\frac{r_i-1+s_0}{s_1-1+s_0}\Big)\\
    \dbar{m}_6(W_i,\theta) &= \Big(\theta_1+\frac{-1+s_1}{2}\Big)\frac{r_i-1+s_0}{s_1-1+s_0}-t_ir_i+\frac{1}{2}\Big(r_i-s_1\frac{r_i-1+s_0}{s_1-1+s_0}\Big)\\
    \end{split}
\end{align}
Joint identified set $\mathbf{\Theta}(P)= \bigcup_{(s_1,s_0)\in\mathcal{S}}\Big(\dbar{\mathcal{H}}_{(\theta_1,\theta_0)}(s_1,s_0)\times\{(s_1,s_0)\}\Big)$ for $\dbar{\mathcal{H}}_{(\theta_1,\theta_0)}(s_1,s_0)$ defined in Proposition \ref{bounds_wrongly_agree_prop} is represented by the moment function $\dbar{m}$. For each $\theta\in\bigcup_{(s_1,s_0)\in\mathcal{S}}[0,\frac{1+s_1}{2}]\times[0,\frac{1+s_0}{2}]\times\{(s_1,s_0)\}$ such that $E_P\big(\dbar{m}_j(W_i,\theta)\big)\leq0\textrm{ for $j=1,\hdots,6$ and } E_P\big(\dbar{m}_7(W_i,\theta)\big)=0$, it must be that $\theta\in\mathbf{\Theta}(P)$. Conversely, if $\theta\in\mathbf{\Theta}(P)$, then $E_P\big(\dbar{m}_j(W_i,\theta)\big)\leq0\textrm{ for $j=1,\hdots,6$ and } E_P\big(\dbar{m}_7(W_i,\theta)\big)=0$. 
\end{theoremEnd}
\begin{proofEnd}
The proof is analogous to the proof of \Cref{prop_original_inequalities}. From the definition of $\dbar{\mathcal{H}}_{\theta_1}(s_1,s_0)$ for $j=1$ in \Cref{bounds_wrongly_agree_prop}:
\begin{align}
\begin{split}
    \theta_1P_{s_1,s_0}(y=1) &\geq max\Big(0,P(t=1,r=1)-P_{s_1,s_0}(r=1,y=0)\Big)\\
    &+max\Big(0,P_{s_1,s_0}(r=0,y=1)-P(t=0,r=0)\Big)\\
    \theta_1P_{s_1,s_0}(y=1) &\leq min\Big(P(t = 1, r=0), \frac{P_{s_1,s_0}(r=0,y=1)}{2}\Big)\\
    &+min\Big(P(t = 1, r=1)-\frac{P_{s_1,s_0}(r=1,y=0)}{2}, P_{s_1,s_0}(r=1,y=1)\Big).
\end{split}
\end{align} Suppose that $E_P\big(\dbar{m}_j(W_i,\theta)\big)\leq 0$ for $j=1,2\hdots,6$ and $E_P\big(\dbar{m}_7(W_i,\theta)\big)= 0$. From \eqref{moment_inequalities_wrongly_agree_both}:
\begin{align}\label{expectations_moment_inequalities_wrongly_agree_both}
    \begin{split}
    E_P\big(\dbar{m}_4(W_i,\theta)\big) &= \theta_1P_{s_1,s_0}(y=1)-P(t=1)+\frac{1}{2}\Big(P(r=1)-s_1P_{s_1,s_0}(y=1)\Big)\\
    &=\theta_1 P_{s_1,s_0}(y=1)-P(t=1,r=0)-P(t=1,r=1)+\frac{P_{s_1,s_0}(r=1,y=0)}{2}\leq 0\\
    E\big(\dbar{m}_6(W_i,\theta)\big) &= \Big(\theta_1+\frac{-1+s_1}{2}\Big)P_{s_1,s_0}(y=1)-P(t=1,r=1)+\frac{1}{2}\Big(P(r=1)-s_1P_{s_1,s_0}(y=1)\Big)\\
     &=\theta_1 P_{s_1,s_0}(y=1)-\frac{P_{s_1,s_0}(r=0,y=1)}{2}-P(t=1,r=1)+\frac{P_{s_1,s_0}(r=1,y=0)}{2}\leq 0\\
    \end{split}
\end{align}
Using $E_P\big(\dbar{m}_j(W_i,\theta)\big) = E_P\big(m_j(W_i,\theta)\big)$ for $j=1,2,3,5$, $E_P\big(\dbar{m}_7(W_i,\theta)\big) = E_P\big(m_7(W_i,\theta)\big)$, \eqref{expectations_moment_inequalities_original_lower}, \eqref{expectations_moment_inequalities_original_upper}, \eqref{expectations_moment_equalities}, and \eqref{expectations_moment_inequalities_wrongly_agree_both} yields that $E_P\big(\dbar{m}_j(W_i,\theta)\big)\leq0\textrm{ for $j=1,\hdots,6$ and } E_P\big(\dbar{m}_7(W_i,\theta)\big)=0$ represent the joint identification $\mathbf{\Theta}(P)= \bigcup_{(s_1,s_0)\in\mathcal{S}}\Big(\dbar{\mathcal{H}}_{(\theta_1,\theta_0)}(s_1,s_0)\times\{(s_1,s_0)\}\Big)$ by the same argument as in the proof of \Cref{prop_original_inequalities}.
\end{proofEnd}

\section{Sensitivity Analysis}\label{sect_sensitivity_analysis}

\renewcommand\thetable{\thesection.\arabic{table}}
\setcounter{table}{0} 
\renewcommand\thefigure{\thesection.\arabic{figure}}
\setcounter{figure}{0} 

The majority of estimates obtained from the $34$ data sets used by \citet{arevalo2020false} indicate that $s_1$ may be even lower than $90\%$. To explore the implications of that possibility, I perform a sensitivity analysis. For exposition purposes, I assume $s_1\in[0.8,0.9]$, so $\mathcal{S}=[0.8,0.9]\times\{1\}$. Values $s_1<0.8$ yield the same conclusion. Estimates of the identified set as well as the corresponding $95\%$ confidence sets are found using data from each of the three samples, and presented together with findings from \Cref{application_sect} in Figures \ref{figureEUA}, \ref{figure_shah_symp} and \ref{figure_shah_asymp} to facilitate comparison. Panel (\subref{figure_EUA:b}) of each Figure depicts the results under the alternative assumption. The solid red region represents the estimated identified set for $(\theta_1,\theta_0)$. It is no longer a line. In all figures both the confidence and the estimated identified set become larger, but remain informative. \Cref{tab:projected_bounds_sensitivity_sensitivity} shows estimates of the projected bounds. Bounds for specificity are unchanged, while those for sensitivity expand only downwards. Assumed values $s_1<0.8$ accentuate this effect. The tendency of ``apparent'' sensitivity to overestimate true sensitivity increases as $s_1$ is reduced. On the other hand, allowing for $s_1>0.9$ enlarges the estimated upper bounds on sensitivity, but for values of $s_1<1$ it never surpasses ``apparent'' sensitivity. Hence, the finding that ``apparent'' sensitivity overestimates true sensitivity is robust to different assumed values of $s_1$.

\begin{figure}[H]
    \centering
    \begin{subfigure}[t]{0.49\textwidth}
        \centering
        \includegraphics[width=\linewidth]{EUA_study_known_exact.png} 
        \caption{$s_1=0.9$} \label{figure_EUA:a}
    \end{subfigure}
    \hfill
    \begin{subfigure}[t]{0.49\textwidth}
        \centering
        \includegraphics[width=\linewidth]{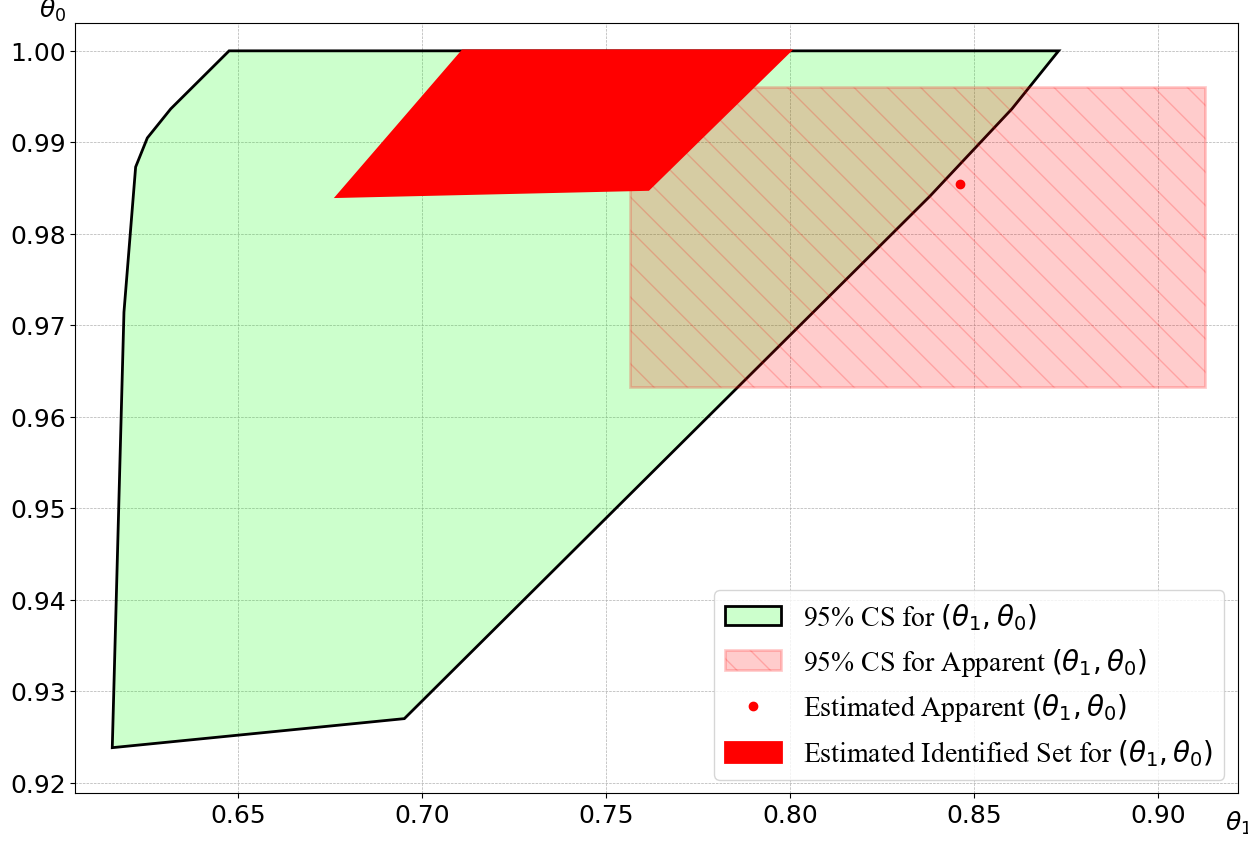} 
        \caption{$s_1\in[0.8,0.9]$} \label{figure_EUA:b}
    \end{subfigure}
    \caption{Estimates, and $95\%$ confidence sets for ``apparent'' measures and points in the identified set for $(\theta_1,\theta_0)$ in the EUA study. In panel (\protect\subref{figure_EUA:a}) $\mathcal{S}=\{(0.9,1)\}$, and $\mathcal{S}=[0.8,0.9]\times\{1\}$ in panel (\protect\subref{figure_EUA:b}).}
    \label{figureEUA}
\end{figure}

\begin{figure}
    \centering
    \begin{subfigure}[t]{0.49\textwidth}
        \centering
        \includegraphics[width=\linewidth]{Shah_symptomatic_known_exact.png} 
        \caption{$s_1=0.9$} \label{figure_shah_symp:a}
    \end{subfigure}
    \hfill
    \begin{subfigure}[t]{0.49\textwidth}
        \centering
        \includegraphics[width=\linewidth]{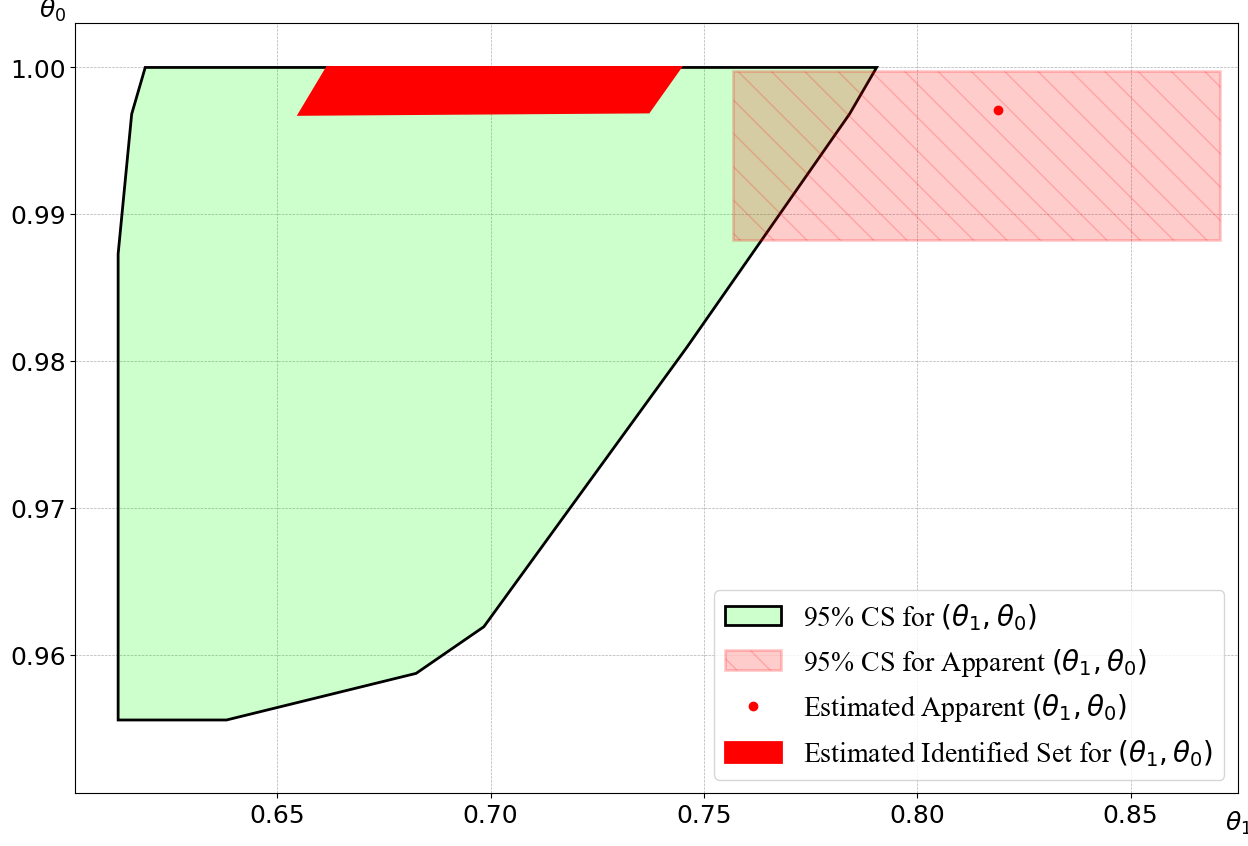} 
        \caption{$s_1\in[0.8,0.9]$} \label{figure_shah_symp:b}
    \end{subfigure}
    \caption{Estimates, and $95\%$ confidence sets for ``apparent'' measures and points in the identified set for $(\theta_1,\theta_0)$ in the symptomatic population of \citet{shah2021performance}. In panel (\protect\subref{figure_shah_symp:a}) $\mathcal{S}=\{(0.9,1)\}$, and $\mathcal{S}=[0.8,0.9]\times\{1\}$ in panel (\protect\subref{figure_shah_symp:b}).}
    \label{figure_shah_symp}
\end{figure}

\begin{figure}
    \centering
    \begin{subfigure}[t]{0.49\textwidth}
        \centering
        \includegraphics[width=\linewidth]{Shah_asymptomatic_known_exact.png} 
        \caption{$s_1=0.9$} \label{figure_shah_asymp:a}
    \end{subfigure}
    \hfill
    \begin{subfigure}[t]{0.49\textwidth}
        \centering
        \includegraphics[width=\linewidth]{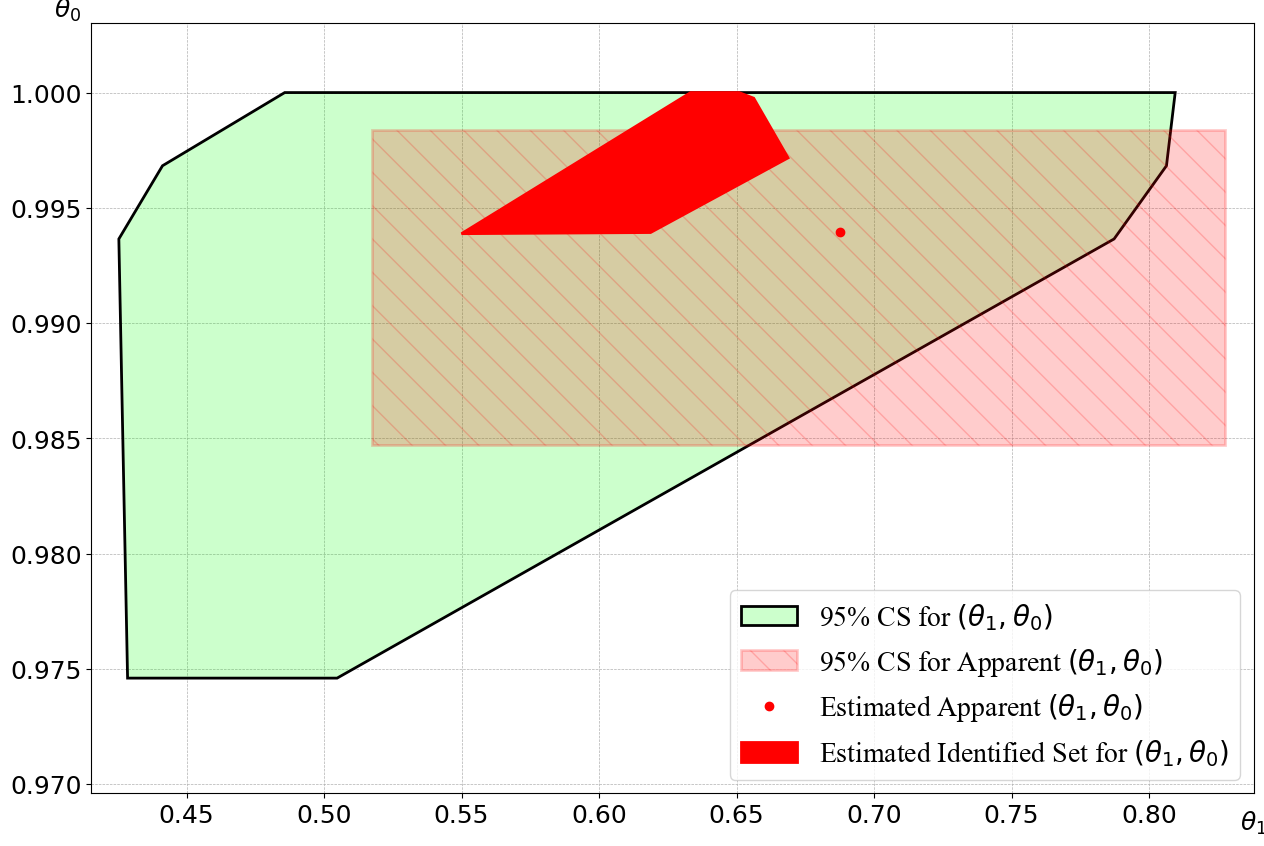} 
        \caption{$s_1\in[0.8,0.9]$} \label{figure_shah_asymp:b}
    \end{subfigure}
      \caption{Estimates, and $95\%$ confidence sets for ``apparent'' measures and points in the identified set for $(\theta_1,\theta_0)$ in the asymptomatic population of \citet{shah2021performance}. In panel (\protect\subref{figure_shah_asymp:a}) $\mathcal{S}=\{(0.9,1)\}$, and $\mathcal{S}=[0.8,0.9]\times\{1\}$ in panel (\protect\subref{figure_shah_asymp:b}).}
    \label{figure_shah_asymp}
\end{figure}

\begin{table}[htbp]
  \centering
  \begin{threeparttable}
    \caption{Estimates}
     \label{tab:projected_bounds_sensitivity_sensitivity}%
    \begin{tabular}{ccccccc}
    \toprule
     &
      \multicolumn{3}{c}{$\theta_1$ Estimates} &
      \multicolumn{3}{c}{$\theta_0$  Estimates}
      \\
\cmidrule(lr){2-4} \cmidrule(lr){5-7}   Data &
      Appar. &
      $s_1=0.9$ &
      $s_1\in[0.8,0.9]$ &
      Appar. &
      $s_1=0.9$ &
      $s_1\in[0.8,0.9]$
      \\
    \midrule
    EUA Sx &
      $0.846$ &
      $[0.761,0.800]$ &
      $[0.677,0.800]$ &
      $0.985$ &
      $[0.985,1.000]$ &
      $[0.984,1.000]$
      \\
    Shah et al. Sx &
      $0.819$ &
      $[0.737,0.744]$ &
      $[0.655,0.744]$ &
      $0.997$ &
      $[0.997,1.000]$ &
      $[0.997,1.000]$
      \\
    Shah et al. ASx &
      $0.688$ &
      $[0.619, 0.669]$ &
      $[0.550, 0.669]$ &
      $0.994$ &
      $[0.994,0.997]$ &
      $[0.994,0.997]$
      \\
    \bottomrule
    \bottomrule
    \end{tabular}%
    \begin{tablenotes}
      \small
      \item \textit{Note: }Apparent estimated values and estimated projected bounds for $(\theta_1,\theta_0)$ for different $\mathcal{S}$. Sx denotes the symptomatic, and ASx the asymptomatic individuals.
    \end{tablenotes}
  \end{threeparttable}
\end{table}%

\section{Auxiliary Results}\label{sect_aux_results}
\renewcommand\theequation{\thesection.\arabic{equation}}
\setcounter{equation}{0} 

\begin{lemma}\label{lemma_trick_equivalence}
    For a fixed $(s_1,s_0)$ and any $(j,k,l)\in\{0,1\}^3$ it holds that:
    \begin{equation}\label{trick_equivalence}
       P(t=j,r=k) - P_{s_1,s_0}(r=k,y=1-l) =P_{s_1,s_0}(r=k,y=l)-P(t=1-j,r=k).
    \end{equation}
\end{lemma}
\begin{proof}
Suppressing the subscript in $P_{s_1,s_0}$ for clarity:
\begin{align}
\begin{split}
    &P(t=j,r=k) - P(r=k,y=1-l) = \\
    &=P(t=j,r=k,y=l)+P(t=j,r=k,y=1-l)\\
    &- P(t=j,r=k,y=1-l)-P(t=1-j,r=k,y=1-l)\\
    &=P(t=j,r=k,y=l)-P(t=1-j,r=k,y=1-l)\\
    &=P(t=j,r=k,y=l)+P(t=1-j,r=k,y=l)\\
    &- P(t=1-j,r=k,y=l)-P(t=1-j,r=k,y=1-l)\\
    &=P(r=k,y=l)-P(t=1-j,r=k).
\end{split}
\end{align}
\end{proof}
\begin{lemma}\label{lemma_ordinary_frechet}
Let $P(t,r)$ and $(s_1,s_0)$ be known, and $\mathcal{H}_{\theta_j}(s_1,s_0)=[\theta_j^L,\theta_j^U]$ as in \eqref{bounds_theta_eq}. Define:
\begin{align}\label{eq:ordinary_frechet}
    \hat{\mathcal{H}}_{\theta_j}(s_1,s_0)=\Big[max\Big(0,P(t=j)-P_{s_1,s_0}(y=1-j)\Big),min\Big(P(t=j),P_{s_1,s_0}(y=j)\Big)\Big]\frac{1}{P_{s_1,s_0}(y=j)}
\end{align}
Then $\mathcal{H}_{\theta_j}(s_1,s_0)\subseteq \hat{\mathcal{H}}_{\theta_j}(s_1,s_0)$.
\end{lemma}
\begin{proof}
By \Cref{lemma_trick_equivalence}, the lower bound in \eqref{eq:ordinary_frechet} is equivalent to $ max\Big(0,P(t=j,r=j)-P_{s_1,s_0}(r=j,y=1-j)+P_{s_1,s_0}(r=1-j,y=j)-P(t=1-j,r=1-j)\Big)\frac{1}{P_{s_1,s_0}(y=j)}\leq \theta_j^L$ since the maximum of a sum of functions is at most the sum of individual maxima. Similarly, the upper bound is $min\Big(P(t=j,r=j)+P(t=j,r=1-j),P_{s_1,s_0}(r=j,y=j)+P_{s_1,s_0}(r=1-j,y=j)\Big)\frac{1}{P_{s_1,s_0}(y=j)}\geq \theta_j^U$.
\end{proof}

    \begin{lemma}\label{lem:comparison}
        For any $P(t,r)$ and $(s_1,s_0)$ such that $\theta_j^L>s_j$ it must be that $\theta_{1-j}^L<s_{1-j}$.
    \begin{proof}
        
    We prove the claim that $\theta_1^L>s_1$ implies $\theta_0^L<s_0$. Symmetrically, one can show that $\theta_0^L>s_0$ implies $\theta_1^L<s_1$.

    Suppose that $\theta_1^L>s_1$. This is equivalent to $\theta_1^LP_{s_1,s_0}(y=1)>P_{s_1,s_0}(r=1,y=1)$ by Assumption \ref{prevalence_assumption}. Taking any lower bound $\theta_j^L$ from Propositions \ref{bounds_theta_prop} and \ref{bounds_wrongly_agree_prop} yields:

    \begin{align}\label{thetal_expressions}
    \begin{split}
        \theta_1^LP_{s_1,s_0}(y=1) &= max\Big(0,P(t=1,r=1)-P_{s_1,s_0}(r=1,y=0)\Big)\\
    &+max\Big(0,P_{s_1,s_0}(r=0,y=1)-P(t=0,r=0)\Big)\\
        \theta_0^LP_{s_1,s_0}(y=0) &= max\Big(0,P(t=0,r=0)-P_{s_1,s_0}(r=0,y=1)\Big)\\
    &+max\Big(0,P_{s_1,s_0}(r=1,y=0)-P(t=1,r=1)\Big).\\
    \end{split}
    \end{align}

    Since $\theta_1^L>s_1$, it must also be that $\theta_1^L>0$ by Assumption 2. By \eqref{thetal_expressions} then $P(t=1,r=1)-P_{s_1,s_0}(r=1,y=0)>0$ or $P_{s_1,s_0}(r=0,y=1)-P(t=0,r=0)>0$. 
    
    We show that $P_{s_1,s_0}(r=0,y=1)-P(t=0,r=0)> 0$ must hold when $\theta_1^L>s_1$. By way of contradiction suppose that $P_{s_1,s_0}(r=0,y=1)-P(t=0,r=0)\leq 0$. Since $P(t=1,r=1)-P_{s_1,s_0}(r=1,y=0)>0$ or $P_{s_1,s_0}(r=0,y=1)-P(t=0,r=0)>0$, it must be $P(t=1,r=1)-P_{s_1,s_0}(r=1,y=0)>0$. Then:
    \begin{align*}
        \theta_1P_{s_1,s_0}(y=1) &= P(t=1,r=1)-P_{s_1,s_0}(r=1,y=0)>P_{s_1,s_0}(r=1,y=1)\\
        &\Longleftrightarrow P(t=1,r=1)>P(r=1)
    \end{align*}
    
    which is a contradiction, so $P_{s_1,s_0}(r=0,y=1)-P(t=0,r=0)> 0$. To complete the proof, we consider two cases.
    
    Suppose first that $P(t=1,r=1)-P_{s_1,s_0}(r=1,y=0)>0$ and $P_{s_1,s_0}(r=0,y=1)-P(t=0,r=0)> 0$. Then, it is immediate from \eqref{thetal_expressions} that $\theta^L_0 = 0$. That $\theta_{0}^L<s_{0}$ is then direct from Assumption \ref{reference_performance_assumption} since $s_1+s_0>1$ and $(s_1,s_0)\in[0,1]^2$.
    
    Finally, suppose that $P(t=1,r=1)-P_{s_1,s_0}(r=1,y=0)\leq 0$ and $P_{s_1,s_0}(r=0,y=1)-P(t=0,r=0)> 0$. By way of contradiction, suppose that $\theta_0^L\geq s_0$ From \eqref{thetal_expressions}:
        \begin{align}\label{eq:theta0_lower_exp}
        \begin{split}
            \theta_0^LP_{s_1,s_0}(y=0) &= P_{s_1,s_0}(r=1,y=0)-P(t=1,r=1)\geq P_{s_1,s_0}(r=0,y=0)\\
            &\Longleftrightarrow P(t=0,r=1)\geq P_{s_1,s_0}(r=1,y=1)+P_{s_1,s_0}(r=0,y=0)
        \end{split}
        \end{align}   
        where the second line follows by \Cref{lemma_trick_equivalence}. Similarly:
                       \begin{align}\label{eq:theta1_lower_exp}
        \begin{split}
            \theta_1^LP_{s_1,s_0}(y=1) &= P_{s_1,s_0}(r=0,y=1)-P(t=0,r=0)> P_{s_1,s_0}(r=1,y=1)\\
                   &\Longleftrightarrow P(t=1,r=0)> P_{s_1,s_0}(r=1,y=1)+P_{s_1,s_0}(r=0,y=0).
            \end{split}
        \end{align}   
    Now using $P(r=j)\geq P(t=1-j,r=j)$, \eqref{eq:theta0_lower_exp} and \eqref{eq:theta1_lower_exp}:
                       \begin{align}\label{eq:contradiction_theta}
        \begin{split}
        P(r=0)>&s_1P_{s_1,s_0}(y=1)+s_0P_{s_1,s_0}(y=0) \\
        P(r=1)\geq &s_1P_{s_1,s_0}(y=1)+s_0P_{s_1,s_0}(y=0) \\
            \end{split}
        \end{align}  
    Recall that Assumptions \ref{reference_performance_assumption} and \ref{prevalence_assumption} imply that $P(r=1)\in(1-s_0,s_1)$, so $P(r=0)\in(1-s_1,s_0)$. This together with \eqref{eq:contradiction_theta} implies $s_j>s_1P_{s_1,s_0}(y=1)+s_0P_{s_1,s_0}(y=0)$ for $j=0,1$ which is equivalent to $s_1>s_0$ and $s_0>s_1$ by Assumption \ref{prevalence_assumption}. This yields a contradiction, showing that $\theta_0^L<s_0$.    
    \end{proof}
    \end{lemma}

\begin{lemma}\label{lem:ident_power}
    Suppose that $t$ and $r$ have a tendency to wrongly agree for some $y=j$. Then:
    \begin{itemize}
        \item  $\Bar{\mathcal{H}}_{(\theta_1,\theta_0)}(s_1,s_0)\subset\mathcal{H}_{(\theta_1,\theta_0)}(s_1,s_0)$ if and only if $P(t=j,r=1-j)>\frac{P_{s_1,s_0}(r=1-j,y=j)}{2}>0$;
        \item  $\Bar{\mathcal{H}}_{(\theta_1,\theta_0)}(s_1,s_0)\subset\mathcal{H}_{(\theta_1,\theta_0)}(s_1,s_0)$ implies $s_j<1$.
    \end{itemize}
\end{lemma}

    \begin{proof}
        
    Focus on the case where the tests have the tendency to wrongly agree only for $y=1$.

    From \eqref{bounds_joint} it is immediate that $\Bar{\mathcal{H}}_{(\theta_1,\theta_0)}(s_1,s_0)\subset\mathcal{H}_{(\theta_1,\theta_0)}(s_1,s_0)$ is equivalent to $\Bar{\mathcal{H}}_{\theta_1}(s_1,s_0)\subset\mathcal{H}_{\theta_1}(s_1,s_0)$, which are intervals. Next, since lower bounds of the two intervals are $\theta_1^L$, this is further equivalent to $\Bar{\theta}_1^U<\theta_1^U$. First, assume $P(t=1,r=0)>\frac{P_{s_1,s_0}(r=0,y=1)}{2}>0$. Then:
    \begin{align*}
        \theta_1^U-\Bar{\theta}_1^U &=min\left(P(t=1,r=0),P_{s_1,s_0}(r=0,y=1)\right)- min\left(P(t=1,r=0),\frac{P_{s_1,s_0}(r=0,y=1)}{2}\right)\\
        &=min\left(P(t=1,r=0)-\frac{P_{s_1,s_0}(r=0,y=1)}{2},\frac{P_{s_1,s_0}(r=0,y=1)}{2}\right)>0.
    \end{align*}

    Next suppose $\Bar{\theta}_1^U<\theta_1^U$. This is equivalent to:
    \begin{align*}
    min\left(P(t=1,r=0),P_{s_1,s_0}(r=0,y=1)\right)> min\left(P(t=1,r=0),\frac{P_{s_1,s_0}(r=0,y=1)}{2}\right).
    \end{align*}
    By means of contradiction suppose that $\frac{P_{s_1,s_0}(r=0,y=1)}{2}= 0$ or $P(t=1,r=0)\leq\frac{P_{s_1,s_0}(r=0,y=1)}{2}$. If $\frac{P_{s_1,s_0}(r=0,y=1)}{2}=0$ then $\Bar{\theta}_1^U=\theta_1^U$. If $P(t=1,r=0)\leq\frac{P_{s_1,s_0}(r=0,y=1)}{2}$, then again $\Bar{\theta}_1^U=\theta_1^U$ proving the first claim for $y=1$. The argument for $y=0$ is symmetric.

    For the second claim, observe that $P(t=j,r=1-j)>\frac{P_{s_1,s_0}(r=1-j,y=j)}{2}>0$ implies $P_{s_1,s_0}(r=1-j,y=j)=(1-s_j)P_{s_1,s_0}(y=j)>0$ or that $s_j<1$ by Assumption \ref{prevalence_assumption}.
    \end{proof}
\section{Proofs}\label{sect_proofs}
\renewcommand\theequation{\thesection.\arabic{equation}}
\setcounter{equation}{0} 

\printProofs

\end{appendices}

\end{document}